\documentclass[]{article}

\usepackage{amsmath,amssymb,amsthm,amsfonts}   %%
\numberwithin{equation}{section}   %%%%%%%%%%%%%%
\usepackage{physics}   %%%%%%%%%%%%%%%%%%%%%%%%%%
\usepackage[english]{babel}   %%%%%%%%%%%%%%%%%%%
\usepackage[margin=3cm]{geometry}   %%%%%%%%%%%%%
\usepackage{csquotes}   %%%%%%%%%%%%%%%%%%%%%%%%%
\usepackage{mathtools}   %%%%%%%%%%%%%%%%%%%%%%%%
\usepackage{relsize}   %%%  mathlarger{}    %%%%%
\usepackage{indentfirst}  %% first paragrph %%%%%
\usepackage{thm-restate}   %%%%%%%%%%%%%%%%%%%%%%
\usepackage{mathrsfs}   %%%%%%%%%%%%%%%%%%%%%%%%%
\usepackage[dvipsnames]{xcolor}   %%%%%%%%%%%%%%%
\usepackage{paracol}    %%%%%%%%%%%%%%%%%%%%%%%%%
\usepackage{enumitem}    %%%%%%%%%%%%%%%%%%%%%%%%
\usepackage{soul}       %%%%%%%%%%%%%%%%%%%%%%%%%
\usepackage{aligned-overset}    %%%%%%%%%%%%%%%%%
\usepackage[thinc]{esdiff}  %%%%%%%%%%%%%%%%%%%%%
\usepackage{upgreek}    %%%%%%%%%%%%%%%%%%%%%%%%%
\usepackage{datetime}   %%%%%%%%%%%%%%%%%%%%%%%%%
\usepackage{tensor}     %%%%%%%%%%%%%%%%%%%%%%%%%
\usepackage{caption}    %%%%%%%%%%%%%%%%%%%%%%%%%
\usepackage{subcaption} %%%%%%%%%%%%%%%%%%%%%%%%%
\usepackage{appendix}   %%%%%%%%%%%%%%%%%%%%%%%%%
\allowdisplaybreaks

\usepackage[dvipsnames]{xcolor}     %%%%%%%%%%%%%
\usepackage{hyperref}	            %%%%%%%%%%%%%
\newcommand\myshade{85}         %%%%%%%%%%%%%%%%%
\colorlet{mylinkcolor}{violet}  %%%%%%%%%%%%%%%%%
\colorlet{mycitecolor}{YellowOrange}    %%%%%%%%%
\colorlet{myurlcolor}{Aquamarine}       %%%%%%%%%
\hypersetup{            %%%%%%%%%%%%%%%%%%%%%%%%%
    linkcolor  = blue!\myshade!black,      %%%%%%
    citecolor  = mycitecolor!\myshade!black,   %%
    urlcolor   = myurlcolor!\myshade!black, %%%%%
    colorlinks = true,      %%%%%%%%%%%%%%%%%%%%%
    hypertexnames=false,
}               %%%%%%%%%%%%%%%%%%%%%%%%%%%%%%%%%
\usepackage{cleveref}           %%%%%
\usepackage[                                %%%%%
style = alphabetic,                         %%%%%
maxbibnames=9,maxcitenames=9,               %%%%%
doi=true,isbn=true,giveninits=true,         %%%%%
backend=biber,                              %%%%%
url=false,arxiv=true]{biblatex}             %%%%%
\theoremstyle{plain}   %%%%%%%%%%%%%%%%%%%%%%%%%%
\crefname{thm}{theorem}{theorems}

\theoremstyle{plain}   %%%%%%%%%%%%%%%%%%%%%%%%%%
\crefname{assumption}{assumption}{assumptions}
\newtheorem{lemma}{Lemma}[section]   %%%%%%%%%%%%
\crefname{lemma}{lemma}{lemmas}
\newtheorem{theorem}{Theorem}[section]  %%%%%%%%%
\crefname{theorem}{theorem}{theorems}

\newtheorem{definition}[lemma]{Definition}    %
\newtheorem{prop}[lemma]{Proposition}   %%%%%%%%%
\crefname{prop}{proposition}{propositions}
\newtheorem{cor}[lemma]{Corollary}   %%%%%%%%%%%%
\newtheorem{remark}[lemma]{Remark}   %%%%%%%%%%%%
\newtheorem{dfn}{Definition}    %%%%%%%%%%%%%%%%%   
\usepackage{pageslts}   %%%%%%%%%%%%%%%%%%%%%%%%%
\usepackage{fancyhdr}   %%%%%%%%%%%%%%%%%%%%%%%%%
\newtheoremstyle{informalthmstyle}%
	{6pt}   % Space above
	{6pt}   % Space below
	{\normalfont}   % Body font
	{}      % Indent amount
	{\bfseries}   % Head font
	{.}     % Punctuation after head
	{0.5em} % Space after head
	{}      % Head spec

\theoremstyle{informalthmstyle}

\newtheoremstyle{examplestyle}%
  {3pt}{3pt}% Space above/below
  {\normalfont}% Body font
  {}% Indent
  {\bfseries}% Head font
  {.}% Punctuation after head
  {0.5em}% Space after head
  {}% Head spec
\theoremstyle{examplestyle}
\newtheorem{example}{Example}
\crefname{example}{example}{Examples}

\crefname{condition}{condition}{conditions}
\Crefname{condition}{Condition}{Conditions}

\usepackage{xparse} %%%%%%%%%%%%%%%%%%%%%%%%	new mb.., mc.., ms       %%%     
\ExplSyntaxOn       %%%%%%%%%%%%%%%%%%%%%%%%                             %%%
\NewDocumentCommand{\definealphabet}{mmmm}{%                             %%%
  \int_step_inline:nnn { `#3 } { `#4 }{%                                 %%%
    \cs_new_protected:cpx { #1 \char_generate:nn { ##1 }{ 11 } }{%       %%%
      \exp_not:N #2 { \char_generate:nn { ##1 } { 11 } }}}}              %%%           
\ExplSyntaxOff      %%%%%%%%%%%%%%%%%%%%%%%%                             %%%
\definealphabet{mb}{\mathbb}{A}{Z}      %%%%%%%%                         %%%  
\definealphabet{mc}{\mathcal}{A}{Z}     %%%%%%%%                         %%%
\definealphabet{scr}{\mathscr}{A}{Z}    %%%%%%%%                         %%%
\definealphabet{mf}{\mathfrak}{A}{Z}    %%%%%%%%                         %%%
\DeclarePairedDelimiterX{\corr}[2]{\mathrm{corr}(}{)}{#1,#2}       %%%%%%%%%%%
\DeclarePairedDelimiterX{\cov}[2]{\mathrm{cov}(}{)}{#1,#2}         %%%%%%%%%%%
\DeclarePairedDelimiterX{\dis}[2]{\mathrm{d}(}{)}{#1, #2} %%%%%%%%%%%%%%%%%%%%
\DeclareMathOperator{\pr}{\mathsf{pr}}           %%%%%%%%%   text{Pr} probability  %%
\newcommand{\matrices}{\mbM_d}
\DeclarePairedDelimiterX{\inner}[2]{\langle}{\rangle}{#1| #2} %%% <1,  2> %%
\renewcommand{\tr}[1]{\operatorname{Tr}\left(#1\right)} %   tr[] sclaes %%%%

\let\oldker\ker         %%%%%%%%%%%%%%%%%%%%%%%%%%%%%%%%%%%%%%%%%%%%%%%%%%%%
\renewcommand{\ker}[1]{\oldker\left(#1\right)}  %%%%%%%%%%%%%%%%%%%%%%%%%%%%
\newcommand{\tracezero}{\mathsf{H}_0}

\title{Asymptotic Replacement for Quantum Channel Products with Applications to Inhomogeneous Matrix Product States}
\usepackage{authblk}

\author[]{Lubashan Pathirana\thanks{lpk@math.ku.dk}}
\affil[]{Department of Mathematical Sciences and QMATH, University of Copenhagen, Denmark.}
\date{}

\begin{document}
\pagenumbering{arabic}
\lhead{\thepage}
\maketitle
\vspace{-1cm}

%%%%%%%%%%%%%%%%%%%%%%%%%%%%%%%%%%%%%%%%%%%%%%%%%%%%%%%%%%%%%%%%%%%%%%%%%%%%
%%%%%%%%%%%%%%%%%%%%%%%%%%%%%%%%%%%%%%%%%%%%%%%%%%%%%%%%%%%%%%%%%%%%%%%%%%%%

\begin{abstract}

    We develop a product-level trace-Dobrushin theory for finite-dimensional quantum channel products and apply it to deterministic and stationary random inhomogeneous matrix product states in left-canonical CPTP gauge.
    For a product of channels, the centered trace-Dobrushin coefficient quantifies the residual dependence on the input state, and its decay is the criterion for trace-norm forgetting.
    In the deterministic setting, this decay is equivalent to asymptotic replacement by a moving replacement channel.
    For two-sided products, pullback forgetting produces a unique boundary state,  which determines the canonical replacement family.
    For stationary random CPTP cocycles, submultiplicativity of the product coefficient yields a trace-Dobrushin Lyapunov exponent.
    We prove that the almost sure negativity of this exponent is equivalent to quenched trace-norm memory loss and gives exponential forward and pullback convergence to a unique dynamically stationary random replacement channel.
    When the \(\varrho\)-mixing profile of the channel environment tends to zero, we obtain annealed super-polynomial estimates, while independence gives annealed exponential estimates.
    Finally, we transfer these estimates to inhomogeneous matrix product states whose auxiliary transfer maps are CPTP.
    These channel estimates transfer to deterministic and stationary random inhomogeneous MPS, giving infinite-volume limits of trace-closed finite-volume states, quantitative boundary stability, and correlation bounds governed by the same auxiliary product coefficients.

\end{abstract}

%%%%%%%%%%%%%%%%%%%%%%%%%%%%%%%%%%%%%%%%%%%%%%%%%%%%%%%%%%%%%%%%%%%%%%%%%%%%
%%%%%%%%%%%%%%%%%%%%%%%%%%%%%%%%%%%%%%%%%%%%%%%%%%%%%%%%%%%%%%%%%%%%%%%%%%%%

 % {\hypersetup{linkcolor=Aquamarine!65!black}\tableofcontents}

%%%%%%%%%%%%%%%%%%%%%%%%%%%%%%%%%%%%%%%%%%%%%%%%%%%%%%%%%%%%%%%%%%%%%%%%%%%%
%%%%%%%%%%%%%%%%%%%%%%%%%%%%%%%%%%%%%%%%%%%%%%%%%%%%%%%%%%%%%%%%%%%%%%%%%%%%

\section{Introduction}
\label{sec:introduction}

%%%%%%%%%%%%%%%%%%%%%%%%%%%%%%%%%%%%%%%%%%%%%%%%%%%%%%%%%%%%%%%%%%%%%%%%%%%%
%%%%%%%%%%%%%%%%%%%%%%%%%%%%%%%%%%%%%%%%%%%%%%%%%%%%%%%%%%%%%%%%%%%%%%%%%%%%

Quantum channels are the basic dynamical objects of finite-dimensional quantum information theory.
If \(\mathcal H\simeq \mbC^d\), a quantum state is a density matrix
\[
	\mcS_d
	:=
	\{\rho\in M_d(\mbC): \rho\geq 0,\ \tr{\rho}=1\}.
\]
A quantum channel is a linear map
\[
	\Phi:M_d(\mbC)\longrightarrow M_d(\mbC)
\]
which is completely positive and trace preserving (CPTP).
Equivalently, by the Choi--Kraus theorem, \(\Phi\) admits a representation
\[
	\Phi(X)
	=
	\sum_{i=1}^r K_iXK_i^*,
	\qquad
	\sum_{i=1}^r K_i^*K_i=I.
\]
Here $K^*$ denotes the conjugate transpose of the matrix $K$. 
Such maps describe the evolution of finite-dimensional open quantum systems and the state transformations available in quantum information theory.
They are the noncommutative analogues of stochastic matrices: trace preservation is conservation of total mass, while complete positivity is stability under extension by an arbitrary finite-dimensional reference system \cite{Stinespring1955,Choi1975,Kraus1983,NielsenChuang2000,Watrous2018}.

The long-time behavior of quantum channels is a central problem in quantum information theory and mathematical physics.
For a single homogeneous channel \(\Phi\), one asks whether repeated application of \(\Phi\) erases the initial state:
\[
	\norm{\Phi^n(\rho)-\Phi^n(\sigma)}_1
	\longrightarrow 0,
	\qquad
	\rho,\sigma\in\mcS_d,
\]
where $\norm{\,\cdot\,}$ denotes the trace-norm. 
Under the usual aperiodicity or primitivity assumptions, this loss of memory is equivalent to convergence toward a unique stationary state,
\[
	\Phi^n(\rho)
	\longrightarrow
	\rho_*,
	\qquad
	\Phi(\rho_*)=\rho_*.
\]
This homogeneous theory is closely related to the spectral theory of completely positive maps, quantum Perron--Frobenius theory, and the distinction between ergodicity, mean ergodicity, mixing, and primitivity \cite{EvansHoeghKrohn1978,FannesNachtergaeleWerner1992,BurgarthEtAl2013,CarbonePautrat2016}.

The present paper studies a more flexible problem.
We consider products of quantum channels, which need not be identical.
Given a time interval \(\mcI\subseteq\mbZ\) and a sequence of channels \((\Phi_t)_{t\in\mcI}\), we write
\[
	\Phi_{t:s}
	:=
	\Phi_{t-1}\circ\Phi_{t-2}\circ\cdots\circ\Phi_s,
	\qquad
	s<t.
\]
Such products arise when the noise mechanism changes in time, when an open system is driven by a time-dependent environment, and when site-dependent matrix product state tensors are represented through finite-dimensional auxiliary transfer channels.
They also arise in stationary random models, where the channel at time \(n\) is sampled along a probability-preserving base system,
\[
	\Phi_{\theta^n\omega},
	\qquad
	(\Omega,\mcF,\pr,\theta).
\]
The corresponding random products form the cocycle
\[
	\Phi_\omega^{(n)}
	:=
	\Phi_{\theta^{n-1}\omega}
	\circ\cdots\circ
	\Phi_\omega.
\]
Random and time-inhomogeneous quantum processes of this type have been studied from several complementary viewpoints, including random quantum cocycles, ergodic quantum processes, and stationary limit theorems \cite{asym,PS23,MovassaghSchenker2021,MovassaghSchenker2022}.

The main feature of the inhomogeneous setting is that there is usually no single limiting state.
Even if the product forgets its initial state, the output may converge toward a moving state rather than toward a fixed density matrix.
Thus, the natural limiting object is not necessarily a fixed channel.
It is a family of replacement channels.
For a state \(\rho\in\mcS_d\), the replacement channel with output \(\rho\) is
\[
	R_\rho(X)
	:=
	\tr{X}\rho.
\]
It erases the input and replaces it by \(\rho\).
We say, informally, that a family of channel products is asymptotically replacing when it converges in induced trace norm to an appropriate family of replacement channels.

The coefficient that detects this replacement behavior is the centered trace-Dobrushin coefficient.
For a positive trace-preserving map \(\Phi\), define
\[
	\kappa_{\rm tr}(\Phi)
	:=
	\sup_{\substack{X=X^*\\ \tr{X}=0\\ X\neq0}}
	\dfrac{\norm{\Phi(X)}_1}{\norm{X}_1}.
\]
Equivalently, by \Cref{prop:dobr},
\[
	\kappa_{\rm tr}(\Phi)
	=
	\sup_{\substack{\rho,\sigma\in\mcS_d\\ \rho\neq\sigma}}
	\dfrac{\norm{\Phi(\rho)-\Phi(\sigma)}_1}{\norm{\rho-\sigma}_1}
	=
	\dfrac{1}{2}
	\sup_{\rho,\sigma\in\mcS_d}
	\norm{\Phi(\rho)-\Phi(\sigma)}_1.
\]
Thus \(\kappa_{\rm tr}(\Phi)\) is the sharp trace-distance contraction coefficient of the channel.
This coefficient is standard at the one-step level.
The point of the present paper is to use the exact coefficient of long products,
\[
	\kappa_{t:s}
	:=
	\kappa_{\rm tr}(\Phi_{t:s}),
\]
as the intrinsic product-level measure of residual trace-norm memory.
Equivalently, \(\kappa_{t:s}\) is one half of the trace-norm diameter of the evolved state space \(\Phi_{t:s}(\mcS_d)\).
It vanishes exactly when the product \(\Phi_{t:s}\) is a replacement channel.
Its decay is therefore the intrinsic product-level notion of loss of memory.

This viewpoint should be distinguished from computable sufficient criteria.
Markov--Dobrushin minorization and quantum Doeblin coefficients provide upper bounds on \(\kappa_{\rm tr}(\Phi)\), often through common lower bounds or semidefinite programs \cite{Dobrushin1956,Dobrushin1970,AccardiLuSouissi2021,SouissiBarhoumi2025,Hirche2024}.
They are useful estimates, but the intrinsic object in this paper is the exact product coefficient \(\kappa_{\rm tr}(\Phi_{t:s})\).
It should also be distinguished from Hilbert--Birkhoff projective contraction.
Projective metrics are natural for positive cone maps and non-normalized transfer operators \cite{ReebKastoryanoWolf2011}.
For CPTP products, trace preservation fixes the trace-zero hyperplane, so the centered trace-norm coefficient is the direct measure of state-memory loss.

A final motivation comes from matrix product states.
Matrix product states provide finite-dimensional descriptions of one-dimensional quantum many-body systems \cite{FannesNachtergaeleWerner1992,PerezGarciaVerstraeteWolfCirac2007,VerstraeteMurgCirac2008}.
When the tensors are placed in a left-canonical CPTP gauge, the auxiliary transfer maps are quantum channels.
Asymptotic replacement of the auxiliary products then yields boundary stability, local thermodynamic limits, and correlation estimates controlled by the transfer product across the separating gap.
In this CPTP-gauge setting, the results below extend and complement existing work on homogeneous and random MPS limits and random MPS correlation lengths \cite{Souissi_2025,MovassaghSchenker2021,MovassaghSchenker2022,MPS}.

%%%%%%%%%%%%%%%%%%%%%%%%%%%%%%%%%%%%%%%%%%%%%%%%%%%%%%%%%%%%%%%%%%%%%%%%%%%%
%%%%%%%%%%%%%%%%%%%%%%%%%%%%%%%%%%%%%%%%%%%%%%%%%%%%%%%%%%%%%%%%%%%%%%%%%%%%

\subsection{Main results}
\label{subsec:main-results}

%%%%%%%%%%%%%%%%%%%%%%%%%%%%%%%%%%%%%%%%%%%%%%%%%%%%%%%%%%%%%%%%%%%%%%%%%%%%
%%%%%%%%%%%%%%%%%%%%%%%%%%%%%%%%%%%%%%%%%%%%%%%%%%%%%%%%%%%%%%%%%%%%%%%%%%%%

\subsubsection*{Deterministic inhomogeneous products}
\label{sec:intro-det-replacement}

%%%%%%%%%%%%%%%%%%%%%%%%%%%%%%%%%%%%%%%%%%%%%%%%%%%%%%%%%%%%%%%%%%%%%%%%%%%%
%%%%%%%%%%%%%%%%%%%%%%%%%%%%%%%%%%%%%%%%%%%%%%%%%%%%%%%%%%%%%%%%%%%%%%%%%%%%

Let \((\Phi_n)_{n\in\mcI}\) be a deterministic sequence of CPTP maps on \(\matrices\).
Here \(\mcI=\mbN_0\) in the one-sided setting and \(\mcI=\mbZ\) in the two-sided setting.
For admissible times \(s<t\), write
\[
	\Phi_{t:s}
	:=
	\Phi_{t-1}\circ\Phi_{t-2}\circ\cdots\circ\Phi_s,
	\qquad
	\kappa_{t:s}
	:=
	\kappa_{\rm tr}(\Phi_{t:s}).
\]

\begin{definition}[Strong asymptotic replacement from an initial time]
\label{def:strong-asymptotic-replacement-from-s}
	Fix an initial time \(s\in\mcI\).
	We say that the products \((\Phi_{t:s})_{t>s}\) are strongly asymptotically replacing from \(s\) if there are states \(\eta_{t:s}\in\mcS_d\) such that, with
	\[
		R_{t:s}(X)
		:=
		\tr{X}\eta_{t:s},
	\]
	one has
	\[
		\lim_{t\to\infty}
		\norm{\Phi_{t:s}-R_{t:s}}_{1\to1}
		=
		0.
	\]
\end{definition}
Here $\norm{\,\cdot\,}_{1\to 1}$ denotes the norm induced on superoperators by the trace norm (\(\norm{\,\cdot\,}_1\)) on matrices. 

The definition is a channel-level formulation of loss of memory.
The output center may depend on \(t\), so the limiting object is a family of replacement channels rather than a single fixed channel.
The next theorem shows that this map-level property is exactly characterized by the product coefficient \(\kappa_{t:s}\).
For a reference state \(\tau\in\mcS_d\), define
\[
	R_{t:s}^{(\tau)}(X)
	:=
	\tr{X}\Phi_{t:s}(\tau).
\]

\begin{restatable}[]{thm}{detReplacementApproximationTheorem}
\label{thm:det-replacement-approximation}
	Fix an initial time \(s\in\mcI\).
	Then the following are equivalent.
	\begin{enumerate}
		\item
        \(\displaystyle\lim_{t\to\infty}\kappa_{t:s}=0\).
		
		\item
		The products \((\Phi_{t:s})_{t>s}\) are strongly asymptotically replacing from \(s\).

		\item
        For every reference state \(\tau\in\mcS_d\), 
        \(\norm{\Phi_{t:s}-R_{t:s}^{(\tau)}}_{1\to1}\to0\) as \(t\to\infty\).
        
        \item
        For some reference state \(\tau\in\mcS_d\), 
        \(\norm{\Phi_{t:s}-R_{t:s}^{(\tau)}}_{1\to1}\to0\) as \(t\to\infty\).
	\end{enumerate}
	Moreover, for every \(s<t\) and every \(\tau\in\mcS_d\),
	\[
		\kappa_{t:s}
		\le
		\norm{\Phi_{t:s}-R_{t:s}^{(\tau)}}_{1\to1}
		\le
		4\kappa_{t:s}.
	\]
	Thus, the quantitative decay of \(\kappa_{t:s}\) is equivalent, up to universal constants, to the quantitative convergence of the product to the reference replacement family.
\end{restatable}

In the two-sided setting, the limiting replacement family can be identified intrinsically by pulling the initial time to the remote past.
This removes the arbitrary reference state from the limiting channel.

\begin{restatable}[]{thm}{detPullbackCanonicalReplacement}
\label{thm:pullback-canonical-replacement}
	Assume that \(\mcI=\mbZ\).
	The following are equivalent.
	\begin{enumerate}
		\item
		For every fixed terminal time \(t\in\mbZ\),
		\(\displaystyle\lim_{s\to-\infty}\kappa_{t:s}=0\).

		\item
        There exists a family of states \((\eta_t)_{t\in\mbZ}\subset\mcS_d\) such that, for 
		\(R_t(X)=\tr{X}\eta_t\), one has, for every fixed terminal time \(t\in\mbZ\),
		\(\norm{\Phi_{t:s}-R_t}_{1\to1}\to0\) as \(s\to-\infty\).
	\end{enumerate}
	When these conditions hold, the states \((\eta_t)_{t\in\mbZ}\) are unique.
	Writing them as \((\rho_t)_{t\in\mbZ}\), one has
	\[
		\rho_{t+1}
		=
		\Phi_t(\rho_t),
		\qquad
		t\in\mbZ.
	\]
	Moreover, for every reference state \(\tau\in\mcS_d\),
	\[
		\rho_t
		=
		\lim_{s\to-\infty}\Phi_{t:s}(\tau),
		\qquad
		\lim_{s\to-\infty}
		\norm{R_{t:s}^{(\tau)}-R_t}_{1\to1}
		=
		0.
	\]
\end{restatable}

When the equivalent conditions in \Cref{thm:pullback-canonical-replacement} hold, we say that the sequence has pullback strong asymptotic replacement.
The theorem shows that the canonical pullback centers form a pullback boundary state \((\rho_t)_{t\in\mbZ}\).
In particular, the consistency relation
\[
    \rho_{t+1}=\Phi_t(\rho_t)
\]
is not an additional hypothesis; it follows from asymptotic replacement from the remote past.
Assume further that the products lose memory in the forward direction, namely that
\[
    \kappa_{t:s}\longrightarrow 0
    \qquad\text{as }t\to\infty
\]
for every fixed \(s\in\mbZ\).  
Then the same canonical pullback boundary family is also the forward replacement family.  
Indeed, the boundary-state identity gives
\[
    \rho_t=\Phi_{t:s}(\rho_s),
    \qquad s<t,
\]
and hence
\[
    R_t=R_{t:s}^{(\rho_s)}.
\]
Therefore, for every \(s<t\) and every reference state \(\tau\in\mcS_d\),
\[
    \norm{\Phi_{t:s}-R_t}_{1\to1}
    \le
    4\kappa_{t:s},
    \qquad
    \norm{R_{t:s}^{(\tau)}-R_t}_{1\to1}
    \le
    2\kappa_{t:s}.
\]
Consequently, for every fixed initial time \(s\in\mbZ\),
\[
    \norm{\Phi_{t:s}-R_t}_{1\to1}
    \longrightarrow 0,
    \qquad
    \norm{R_{t:s}^{(\tau)}-R_t}_{1\to1}
    \longrightarrow 0
\]
as \(t\to\infty\).  Thus, under pullback memory loss and forward memory loss, the reference-dependent forward replacement family is asymptotic to the canonical pullback replacement family.

%%%%%%%%%%%%%%%%%%%%%%%%%%%%%%%%%%%%%%%%%%%%%%%%%%%%%%%%%%%%%%%%%%%%%%%%%%%%
%%%%%%%%%%%%%%%%%%%%%%%%%%%%%%%%%%%%%%%%%%%%%%%%%%%%%%%%%%%%%%%%%%%%%%%%%%%%

\subsubsection*{Random products over a probability-preserving base}
\label{sec:intro-random-replacement}

%%%%%%%%%%%%%%%%%%%%%%%%%%%%%%%%%%%%%%%%%%%%%%%%%%%%%%%%%%%%%%%%%%%%%%%%%%%%
%%%%%%%%%%%%%%%%%%%%%%%%%%%%%%%%%%%%%%%%%%%%%%%%%%%%%%%%%%%%%%%%%%%%%%%%%%%%

We next place the deterministic inhomogeneous theory over a stationary random base.
Let \((\Omega,\mcF,\pr,\theta)\) be an invertible probability-preserving system, and let \(\omega\mapsto \Phi_\omega\) be a measurable family of CPTP maps on \(\matrices\).
For \(n\ge1\), set
\[
   \Phi_\omega^{(n)}
   :=
   \Phi_{\theta^{n-1}\omega}
   \circ
   \Phi_{\theta^{n-2}\omega}
   \circ\cdots\circ
   \Phi_\omega,
   \qquad
   \Phi_\omega^{(0)}:=\mathrm{id}.
\]
More generally, for integers \(s<t\), write
\[
   \Phi_{\omega;t:s}
   :=
   \Phi_{\theta^{t-1}\omega}
   \circ
   \Phi_{\theta^{t-2}\omega}
   \circ\cdots\circ
   \Phi_{\theta^s\omega},
   \qquad
   \Phi_{\omega;s:s}:=\mathrm{id}.
\]
Thus \(\Phi_\omega^{(n)}=\Phi_{\omega;n:0}\).
We write \(\kappa_{\omega;t:s}:=\kappa_{\rm tr}(\Phi_{\omega;t:s})\) and \(\kappa_n(\omega):=\kappa_{\omega;n:0}\), so that \(\kappa_{\omega;t:s}=\kappa_{t-s}(\theta^s\omega)\).

Related random and inhomogeneous quantum-process frameworks appear in \cite{asym,PS23,MovassaghSchenker2021,MovassaghSchenker2022}.
A recurring theme in these works is the construction of dynamically compatible random states or boundary-state fields.

\begin{definition}[Dynamically stationary random state]
\label{def:dynamically-stationary-random-state}
    A dynamically stationary random state for the cocycle is a measurable map \(\rho:\Omega\to\mcS_d\), written \(\omega\mapsto\rho_\omega\), such that
    \[
       \Phi_\omega(\rho_\omega)
       =
       \rho_{\theta\omega}
       \qquad
       \text{for }\pr\text{-a.e. }\omega.
    \]
    It is unique if every other such map agrees with it \(\pr\)-almost surely.
\end{definition}

By \Cref{thm:random-lyapunov}, the limit
\[
   \lambda_{\rm tr}(\omega)
   =
   \lim_{n\to\infty}\frac{1}{n}\log\kappa_n(\omega),
   \qquad
   \log0:=-\infty,
\]
exists for \(\pr\)-a.e. \(\omega\) and defines a \(\theta\)-invariant random variable \(\lambda_{\rm tr}:\Omega\to[-\infty,0]\).
We call \(\lambda_{\rm tr}\) the trace-Dobrushin Lyapunov exponent.
It is constant on ergodic components; in particular, if \(\theta\) is ergodic, then \(\lambda_{\rm tr}\) is deterministic almost surely.

For each fixed realization \(\omega\), the sequence \(\Phi_\omega,\Phi_{\theta\omega},\Phi_{\theta^2\omega},\ldots\) is an ordinary deterministic time-inhomogeneous sequence of channels.
Hence, the deterministic replacement results apply fiberwise.
The probabilistic structure is used to quantify the decay of the corresponding product coefficients.
The relevant hypothesis in the general probability-preserving setting is \(\lambda_{\rm tr}(\omega)<0\) for \(\pr\)-a.e. \(\omega\).
Criteria for verifying this almost-sure negativity are proved in \Cref{sec:random-negative-criteria}.
The trace-Dobrushin criterion below is formulated directly at the level of product contraction.
In particular, eventual strict positivity gives one sufficient mechanism for a negative exponent, but it is not necessary; see \Cref{cor:strict-positive-block-negative-exponent} and \Cref{ex:amplitude-random-strictness}.

\begin{restatable}[]{thm}{randomrates}
\label{thm:negative-exponent-replacement}
    Assume that \(\lambda_{\rm tr}(\omega)<0\) with probability one.
    Then there exists a unique dynamically stationary random state \(\rho:\Omega\to\mcS_d\) in the sense of \Cref{def:dynamically-stationary-random-state}.
    Let \(R_\omega(X):=\tr{X}\rho_\omega\).
    Then there exist a measurable \(\theta\)-invariant random variable \(\beta:\Omega\to(-\infty,0)\), a full-measure \(\theta\)-invariant set \(\Omega_*\subseteq\Omega\), and almost surely finite measurable random variables \(C_\beta^+\) and \(C_\beta^-\) such that, for every \(\omega\in\Omega_*\) and every \(n\ge1\),
    \[
    \begin{aligned}
       \norm{\Phi_{\omega;n:0}-R_{\theta^n\omega}}_{1\to1}
       &\le
       4C_\beta^+(\omega)e^{\beta(\omega)n},
       \\
       \norm{\Phi_{\omega;0:-n}-R_\omega}_{1\to1}
       &\le
       4C_\beta^-(\omega)e^{\beta(\omega)n}.
    \end{aligned}
    \]
    The corresponding state-level estimates hold with the constant \(4\) replaced by \(2\).
\end{restatable}

We now impose a stochastic decorrelation assumption on the unit-time channel sequence.
Set \(X_j(\omega):=\Phi_{\theta^j\omega}\) for \(j\in\mbZ\).
Define
\[
   \mcG_{-\infty}^k:=\sigma(X_j:j<k),
   \qquad
   \mcG_k^\infty:=\sigma(X_j:j\ge k).
\]
For two sub-\(\sigma\)-algebras \(\mcA,\mcB\subseteq\mcF\), define the maximal-correlation coefficient by
\[
   \varrho_{\rm mc}(\mcA,\mcB)
   :=
   \sup
   \left\{
      \left|\mbE[UV]\right|:
      U\in L^2(\mcA),\,
      V\in L^2(\mcB),\,
      \mbE[U]=\mbE[V]=0,\,
      \norm{U}_2=\norm{V}_2=1
   \right\}.
\]
For \(m\ge1\), set
\[
   \varrho_m
   :=
   \sup_{k\in\mbZ}
   \varrho_{\rm mc}
   \left(
      \mcG_{-\infty}^k,\mcG_{k+m}^{\infty}
   \right).
\]
We say that the channel environment is \(\varrho\)-mixing if \(\varrho_m\to0\).
This is a decorrelation assumption on the random environment only.
The loss of quantum memory is supplied separately by the negative trace-Dobrushin Lyapunov exponent.
If the two-sided sequence \((X_j)_{j\in\mbZ}\) is independent, then \(\varrho_m=0\) for every \(m\ge1\).

\begin{restatable}[]{thm}{randomannealeddeviations}
\label{thm:random-annealed-deviations}
    Assume that \(\lambda_{\rm tr}(\omega)<0\) with probability one, and let \(\rho:\Omega\to\mcS_d\) be the stationary random state from \Cref{thm:negative-exponent-replacement}.
    Put \(R_\omega(X):=\tr{X}\rho_\omega\), and define
    \[
       \Delta_n^{\rm op}(\omega)
       :=
       \max
       \left\{
          \norm{\Phi_{\omega;n:0}-R_{\theta^n\omega}}_{1\to1},
          \norm{\Phi_{\omega;0:-n}-R_\omega}_{1\to1}
       \right\}.
    \]
    Then the following hold.
    \begin{enumerate}
        \item
        If \(\varrho_m\to0\), then for every \(p\in\mbN\) there is \(C_p<\infty\) such that \(\mbE[\Delta_n^{\rm op}]\le C_p n^{-p}\) for all \(n\ge1\).

        \item
        If \((X_j)_{j\in\mbZ}\) are independent, then there are \(C<\infty\) and \(\gamma>0\) such that \(\mbE[\Delta_n^{\rm op}]\le Ce^{-\gamma n}\) for all \(n\ge1\).
    \end{enumerate}
\end{restatable}

%%%%%%%%%%%%%%%%%%%%%%%%%%%%%%%%%%%%%%%%%%%%%%%%%%%%%%%%%%%%%%%%%%%%%%%%%%%%
%%%%%%%%%%%%%%%%%%%%%%%%%%%%%%%%%%%%%%%%%%%%%%%%%%%%%%%%%%%%%%%%%%%%%%%%%%%%

\subsection{Inhomogeneous Matrix Product States}
\label{sec:intro-mps}

%%%%%%%%%%%%%%%%%%%%%%%%%%%%%%%%%%%%%%%%%%%%%%%%%%%%%%%%%%%%%%%%%%%%%%%%%%%%
%%%%%%%%%%%%%%%%%%%%%%%%%%%%%%%%%%%%%%%%%%%%%%%%%%%%%%%%%%%%%%%%%%%%%%%%%%%%

The deterministic product theory has a direct application to deterministic inhomogeneous matrix product states in left-canonical gauge.
We use the standard terminology of matrix product states from \cite{PerezGarciaVerstraeteWolfCirac2007}.
We write the auxiliary transfer maps in the Schrödinger picture, so the left-canonical condition gives trace-preserving auxiliary channels.
Equivalently, the dual Heisenberg-picture maps are unital.
This is the same convention used in channel-based approaches to inhomogeneous MPS thermodynamic limits \cite{Souissi_2025}.
This channel viewpoint also connects with random MPS, including the homogeneous random setting of \cite{LancienPerezGarcia2022} and the random or ergodic constructions in \cite{MovassaghSchenker2021,MovassaghSchenker2022,MPS}.
The advantage of the present formulation is that the exact trace-Dobrushin coefficient \(\kappa_{\rm tr}\) gives quantitative boundary-stability and trace-closed thermodynamic-limit estimates.
The resulting matrix product states are deterministic inhomogeneous MPS.
They are not assumed to be translation invariant.
The translation-invariant case is recovered, in this gauge, when the Kraus family \(\{K_i^{[n]}\}_i\) is independent of \(n\).

Let \(\mcK\) be the finite-dimensional physical single-site Hilbert space.
Let \(\mcH\) be the finite-dimensional auxiliary, or bond, Hilbert space.
We write \(B(\mcH)\) for the algebra of bounded operators on \(\mcH\).
We also write \(\mcS(\mcH)
    :=
    \{
        \rho\in B(\mcH):
        \rho\ge0,\,
        \tr{\rho}=1
    \}
\)
for the set of auxiliary density matrices.
Since \(\mcH\) is finite-dimensional, \(B(\mcH)\) is simply the full matrix algebra on \(\mcH\).
Fix an orthonormal basis \(\{\ket{1},\ldots,\ket{d_{\mcK}}\}\) of \(\mcK\).
For each site \(n\ge1\), let
\[
    \left\{
        K_i^{[n]}
    \right\}_{i=1}^{d_{\mcK}}
    \subset B(\mcH)
\]
be deterministic site-dependent MPS tensors.
We assume the left-canonical condition
\[
    \sum_{i=1}^{d_{\mcK}}
    \left(K_i^{[n]}\right)^*
    K_i^{[n]}
    =
    I_{\mcH}.
\]
Equivalently, the auxiliary transfer map
\[
    \Phi_n(\rho)
    :=
    \sum_{i=1}^{d_{\mcK}}
    K_i^{[n]}\rho\left(K_i^{[n]}\right)^*
\]
is completely positive and trace-preserving on \(B(\mcH)\).

For \(1\le a\le b\), write
\[
    \mcK_{[a,b]}
    :=
    \bigotimes_{r=a}^{b}\mcK,
    \qquad
    \mcA_{[a,b]}
    :=
    B(\mcK_{[a,b]}).
\]
We identify \(\mcA_{[1,m]}\) with a subalgebra of \(\mcA_{[1,n]}\), for \(m\le n\), through the embedding \(X\mapsto X\otimes I_{\mcK}^{\otimes(n-m)}\).
The one-sided quasi-local algebra is
\[
    \mcA_{\mbN}
    :=
    \overline{
        \bigcup_{m\ge1}
        \mcA_{[1,m]}
    }^{\norm{\cdot}_\infty}.
\]
A local observable is an element of some finite-volume algebra \(\mcA_{[1,m]}\).
The thermodynamic-limit statements below concern the expectations of fixed local observables as the volume \(n\) tends to infinity.

For \(0\le m\le n\), define the right-tail auxiliary transfer product by
\[
    \Theta_{m,n}
    :=
    \begin{cases}
        \Phi_{m+1}\circ\Phi_{m+2}\circ\cdots\circ\Phi_n,
        & m<n,\\
        \operatorname{id}_{B(\mcH)},
        & m=n.
    \end{cases}
\]
Thus \(\Theta_{m,n}\) is a right-tail, or backward auxiliary, product; in the deterministic-product notation of \Cref{sec:det}, it becomes a pullback product after reversing the site index.
For a multi-index \(\mathbf i=(i_a,\ldots,i_b)\), set
\[
    \ket{\mathbf i}
    :=
    \ket{i_a}\otimes\cdots\otimes\ket{i_b},
    \qquad
    K_{\mathbf i}^{[a,b]}
    :=
    K_{i_a}^{[a]}K_{i_{a+1}}^{[a+1]}\cdots K_{i_b}^{[b]}.
\]
For \(X\in\mcA_{[a,b]}\), define the inserted auxiliary transfer map by
\[
    \widehat X_{[a,b]}(Y)
    :=
    \sum_{\mathbf i,\mathbf j}
    \bra{\mathbf i}X\ket{\mathbf j}
    K_{\mathbf j}^{[a,b]}Y
    \left(K_{\mathbf i}^{[a,b]}\right)^*,
    \qquad
    Y\in B(\mcH).
\]
In particular,
\[
    \widehat I_{[a,b]}
    =
    \Phi_a\circ\Phi_{a+1}\circ\cdots\circ\Phi_b,
\]
and, for \(m<n\),
\[
    \widehat I_{[m+1,n]}
    =
    \Theta_{m,n}.
\]

For \(n\ge1\), define the trace-closed finite-volume MPS vector by
\[
    \ket{\Psi_n}
    :=
    \sum_{i_1,\ldots,i_n=1}^{d_{\mcK}}
    \tr{
        K_{i_1}^{[1]}\cdots K_{i_n}^{[n]}
    }
    \ket{i_1\cdots i_n}.
\]
The trace in the coefficient closes the auxiliary bond.
Whenever \(\inner{\Psi_n}{\Psi_n}\ne0\), let \(\varphi_n\) be the normalized \(n\)-site state on \(\mcA_{[1,n]}\) given by
\[
    \varphi_n(Z)
    :=
    \frac{
        \bra{\Psi_n}Z\ket{\Psi_n}
    }{
        \inner{\Psi_n}{\Psi_n}
    },
    \qquad
    Z\in\mcA_{[1,n]}.
\]
If \(X\in\mcA_{[1,m]}\) and \(m\le n\), we also write \(\varphi_n(X)\) for the expectation of the canonically embedded observable,
\[
    \varphi_n(X)
    :=
    \varphi_n
    \left(
        X\otimes I_{\mcK}^{\otimes(n-m)}
    \right).
\]
Thus the thermodynamic limit concerns \(\varphi_n(X)\) for fixed \(m\in\mbN\) and fixed \(X\in\mcA_{[1,m]}\) as \(n\to\infty\).

\begin{restatable}[]{thm}{introMpsThermodynamicLimit}
\label{thm:intro-mps-thermodynamic-limit}
    Assume that, for every fixed \(q\in\mbN_0\), \(\kappa_{\rm tr}\left(\Theta_{q,n}\right)\to0\) as \(n\to\infty\).
    Then there exists a unique right boundary sequence
    \[
        \rho_r\in\mcS(\mcH),
        \qquad
        \rho_r=\Phi_r(\rho_{r+1}),
        \qquad
        r\ge1,
    \]
    and a unique state \(\varphi_\infty\) on the one-sided quasi-local algebra \(\mcA_{\mbN}\) such that, for every \(m\ge1\) and every \(X\in\mcA_{[1,m]}\),
    \[
        \varphi_\infty(X)
        =
        \tr{
            \widehat X_{[1,m]}(\rho_{m+1})
        }.
    \]
    Let \(\ket{\Psi_n}\) be the trace-closed finite-volume MPS vector, and let \(\varphi_n\) be the normalized \(n\)-site state whenever \(\inner{\Psi_n}{\Psi_n}\ne0\).
    For every fixed local observable \(X\in\mcA_{[1,m]}\), the normalized expectations
    \[
        \varphi_n(X)
        =
        \frac{
            \bra{\Psi_n}
            \left(
                X\otimes I_{\mcK}^{\otimes(n-m)}
            \right)
            \ket{\Psi_n}
        }{
            \inner{\Psi_n}{\Psi_n}
        },
        \qquad
        n\ge m,
    \]
    are defined for all sufficiently large \(n\), and
    \[
        \varphi_n(X)
        \longrightarrow
        \varphi_\infty(X)
        \qquad
        \text{as }n\to\infty.
    \]
\end{restatable}

The same right-tail memory-loss mechanism also controls spatial clustering.
Indeed, once the observable \(B\) is inserted, the only auxiliary information connecting it to the observable \(A\) must pass through the identity-transfer product across the gap.

\begin{restatable}[]{thm}{introMpsCorrelationBound}
\label{thm:intro-mps-correlation-bound}
    Assume the hypotheses of \Cref{thm:intro-mps-thermodynamic-limit}.
    Let \(A\in\mcA_{[p,q]}\) and \(B\in\mcA_{[r,s]}\), where \(1\le p\le q\) and \(q+1<r\le s\).
    Regard \(A\) and \(B\) as observables on \(\mcA_{\mbN}\) through the canonical embeddings.
    Then
    \[
        \left|
            \varphi_\infty(AB)
            -
            \varphi_\infty(A)\varphi_\infty(B)
        \right|
        \le
        4
        \norm{A}_\infty
        \norm{B}_\infty
        \kappa_{\rm tr}\left(\Theta_{q,r-1}\right).
    \]
\end{restatable}

The quantitative estimates behind the preceding theorems are expressed entirely in terms of the same trace-Dobrushin coefficients.
For the local thermodynamic limit, the relevant coefficient is the right-tail coefficient \(\kappa_{\rm tr}(\Theta_{m,n})\), where \(X\in\mcA_{[1,m]}\) is fixed and \(n\to\infty\).
Under the hypothesis of \Cref{thm:intro-mps-thermodynamic-limit}, for every fixed \(m\ge1\), the decay \(\kappa_{\rm tr}(\Theta_{m,n})\to0\) implies that, for all sufficiently large \(n\), the trace-closed normalization satisfies \(\inner{\Psi_n}{\Psi_n}\ne0\).
Moreover, for every fixed local observable \(X\in\mcA_{[1,m]}\), one has, for all sufficiently large \(n\),
\[
    \left|
        \varphi_n(X)
        -
        \varphi_\infty(X)
    \right|
    \le
    16D_{\mcH}^2
    \norm{X}_\infty
    \kappa_{\rm tr}\left(\Theta_{m,n}\right),
    \qquad
    D_{\mcH}:=\dim\mcH.
\]
Therefore, one may use \Cref{sec:det-rates} to obtain the convergence rates above. 

%%%%%%%%%%%%%%%%%%%%%%%%%%%%%%%%%%%%%%%%%%%%%%%%%%%%%%%%%%%%%%%%%%%%%%%%%%%%
%%%%%%%%%%%%%%%%%%%%%%%%%%%%%%%%%%%%%%%%%%%%%%%%%%%%%%%%%%%%%%%%%%%%%%%%%%%%

\subsubsection*{Random MPS with CPTP cocycles}
\label{sec:intro-random-mps}

%%%%%%%%%%%%%%%%%%%%%%%%%%%%%%%%%%%%%%%%%%%%%%%%%%%%%%%%%%%%%%%%%%%%%%%%%%%%
%%%%%%%%%%%%%%%%%%%%%%%%%%%%%%%%%%%%%%%%%%%%%%%%%%%%%%%%%%%%%%%%%%%%%%%%%%%%

We now pass from deterministic inhomogeneous MPS to stationary random inhomogeneous MPS.
Let \((\Omega,\mcF,\pr,\theta)\) be an invertible probability-preserving system.
Let
\[
    \left\{
        K_i(\omega)
    \right\}_{i=1}^{d_{\mcK}}
    \subset B(\mcH)
\]
be a measurable family of left-canonical tensors, meaning that
\[
    \sum_{i=1}^{d_{\mcK}}
    K_i(\omega)^*K_i(\omega)
    =
    I_{\mcH}
\]
for \(\pr\)-a.e. \(\omega\).
After restricting to a full-measure \(\theta\)-invariant set, we assume this identity holds along every orbit.
The associated CPTP transfer map is
\[
    \Phi_\omega(Y)
    :=
    \sum_{i=1}^{d_{\mcK}}
    K_i(\omega)Y K_i(\omega)^*,
    \qquad
    Y\in B(\mcH).
\]
Thus \(\omega\mapsto\Phi_\omega\) is a stationary CPTP cocycle.

For a fixed realization \(\omega\), place the tensor sampled at \(\theta^{-n}\omega\) on site \(n\):
\[
    K_i^{[n]}(\omega)
    :=
    K_i(\theta^{-n}\omega),
    \qquad
    n\ge1.
\]
The resulting MPS is inhomogeneous for each fixed realization \(\omega\), while the tensor sequence is stationary in distribution.
We use the deterministic MPS notation fiberwise.
Thus \(\widehat X_{\omega;[a,b]}\), \(\Theta_{\omega;m,n}\), \(\ket{\Psi_n^\omega}\), and \(\varphi_{n,\omega}\) denote the corresponding objects built from the realization \(\left(K_i^{[n]}(\omega)\right)_{n\ge1}\).

With this indexing convention, the right-tail product is
\[
    \Theta_{\omega;m,n}
    =
    \Phi_{\theta^{-(m+1)}\omega}
    \circ
    \Phi_{\theta^{-(m+2)}\omega}
    \circ
    \cdots
    \circ
    \Phi_{\theta^{-n}\omega}.
\]
Equivalently,
\[
    \Theta_{\omega;m,n}
    =
    \Phi_{\theta^{-m}\omega;0:-(n-m)}.
\]
Thus the random MPS right-tail product is exactly a pullback product for the random CPTP cocycle.
This is the bridge between the random-channel replacement theory of \Cref{thm:negative-exponent-replacement} and the random MPS thermodynamic limit.

The random MPS works \cite{MovassaghSchenker2021,MovassaghSchenker2022,MPS} are based on irreducibility or eventual strict-positivity mechanisms for transfer products.
In the CPTP-gauge setting considered here, eventual strict positivity is one sufficient mechanism for a negative trace-Dobrushin Lyapunov exponent, but it is not necessary.
The hypothesis below is formulated directly at the contraction level through \(\lambda_{\rm tr}<0\).
Consequently, at the level of CPTP transfer cocycles, the present results apply beyond the eventual strict-positivity regime.

For \(n\ge1\), define the trace-closed random finite-volume MPS vector by
\[
    \ket{\Psi_n^\omega}
    :=
    \sum_{i_1,\ldots,i_n=1}^{d_{\mcK}}
    \tr{
        K_{i_1}^{[1]}(\omega)
        \cdots
        K_{i_n}^{[n]}(\omega)
    }
    \ket{i_1\cdots i_n}.
\]
Whenever \(\inner{\Psi_n^\omega}{\Psi_n^\omega}\ne0\), let \(\varphi_{n,\omega}\) be the normalized \(n\)-site state on \(\mcA_{[1,n]}\), namely
\[
    \varphi_{n,\omega}(Z)
    :=
    \frac{
        \bra{\Psi_n^\omega}Z\ket{\Psi_n^\omega}
    }{
        \inner{\Psi_n^\omega}{\Psi_n^\omega}
    },
    \qquad
    Z\in\mcA_{[1,n]}.
\]
If \(X\in\mcA_{[1,m]}\) and \(m\le n\), we also write
\[
    \varphi_{n,\omega}(X)
    :=
    \varphi_{n,\omega}
    \left(
        X\otimes I_{\mcK}^{\otimes(n-m)}
    \right).
\]

\begin{restatable}[]{thm}{introRandomMpsThermodynamicLimit}
\label{thm:intro-random-mps-thermodynamic-limit}
    Assume that \(\lambda_{\rm tr}<0\) \(\pr\)-a.s.
    Then there exist a measurable \(\theta\)-invariant random variable \(\beta:\Omega\to(-\infty,0)\), a full-measure \(\theta\)-invariant set \(\Omega_*\subseteq\Omega\), and an a.s. finite random variable \(C_\beta^-\) such that the following holds.
    There exists a unique stationary random boundary state
    \[
        \rho:\Omega\to\mcS(\mcH),
        \qquad
        \Phi_\omega(\rho_\omega)=\rho_{\theta\omega},
    \]
    up to \(\pr\)-a.e. equality.
    The representative may be chosen so that the stationarity relation holds for every \(\omega\in\Omega_*\).
    There exists a weak-\(*\) measurable random state field
    \[
        \varphi_\infty:\Omega\to\mcS(\mcA_{\mbN}),
        \qquad
        \omega\mapsto\varphi_{\infty,\omega},
    \]
    unique up to \(\pr\)-a.e. equality among weak-\(*\) measurable fields satisfying the local formula below.
    For every \(\omega\in\Omega_*\), every \(m\ge1\), and every \(X\in\mcA_{[1,m]}\),
    \[
        \varphi_{\infty,\omega}(X)
        =
        \tr{
            \widehat X_{\omega;[1,m]}
            \left(
                \rho_{\theta^{-m}\omega}
            \right)
        }.
    \]
    For every \(\omega\in\Omega_*\) and every fixed local observable \(X\in\mcA_{[1,m]}\), the normalized trace-closed expectations \(\varphi_{n,\omega}(X)\) are defined for all sufficiently large \(n\), and for all sufficiently large \(n\),
    \[
        \left|
            \varphi_{n,\omega}(X)
            -
            \varphi_{\infty,\omega}(X)
        \right|
        \le
        16D_{\mcH}^2
        \norm{X}_\infty
        C_\beta^-\left(\theta^{-m}\omega\right)
        e^{\beta(\omega)(n-m)},
        \qquad
        D_{\mcH}:=\dim\mcH.
    \]
    In particular,
    \[
        \varphi_{n,\omega}(X)
        \longrightarrow
        \varphi_{\infty,\omega}(X)
        \qquad
        \text{as }n\to\infty
    \]
    for every fixed local observable \(X\) and every \(\omega\in\Omega_*\).
\end{restatable}

The same quenched trace-Dobrushin rate gives spatial clustering in the random infinite-volume state.
For two separated observables, the auxiliary information connecting the right observable to the left observable must pass through the random identity-transfer product across the gap.

\begin{restatable}[]{thm}{introRandomMpsCorrelationBound}
\label{thm:intro-random-mps-correlation-bound}
    Assume that \(\lambda_{\rm tr}<0\) \(\pr\)-a.s.
    Let \(\Omega_*\), \(\beta\), \(C_\beta^-\), and the weak-\(*\) measurable representative \(\varphi_\infty\) be chosen as in \Cref{thm:intro-random-mps-thermodynamic-limit}.
    Let \(A\in\mcA_{[p,q]}\) and \(B\in\mcA_{[r,s]}\), where \(1\le p\le q\) and \(q+1<r\le s\).
    Regard \(A\) and \(B\) as observables on \(\mcA_{\mbN}\) through the canonical embeddings.
    Set
    \[
        L:=r-q-1.
    \]
    Then, for every \(\omega\in\Omega_*\),
    \[
        \left|
            \varphi_{\infty,\omega}(AB)
            -
            \varphi_{\infty,\omega}(A)
            \varphi_{\infty,\omega}(B)
        \right|
        \le
        4
        \norm{A}_\infty
        \norm{B}_\infty
        C_\beta^-\left(\theta^{-q}\omega\right)
        e^{\beta(\omega)L}.
    \]
    Thus the random infinite-volume MPS has quenched exponential clustering across gaps.
\end{restatable}

We note that one may obtain high probability annealed decay similar to the techniques developed in \Cite{MPS}. We show one such estimate in \Cref{cor:random-mps-high-probability-clustering}. 

%%%%%%%%%%%%%%%%%%%%%%%%%%%%%%%%%%%%%%%%%%%%%%%%%%%%%%%%%%%%%%%%%%%%%%%%%%%%
%%%%%%%%%%%%%%%%%%%%%%%%%%%%%%%%%%%%%%%%%%%%%%%%%%%%%%%%%%%%%%%%%%%%%%%%%%%%

\section{The centered trace-Dobrushin coefficient}
\label{sec:coefficient}

%%%%%%%%%%%%%%%%%%%%%%%%%%%%%%%%%%%%%%%%%%%%%%%%%%%%%%%%%%%%%%%%%%%%%%%%%%%%
%%%%%%%%%%%%%%%%%%%%%%%%%%%%%%%%%%%%%%%%%%%%%%%%%%%%%%%%%%%%%%%%%%%%%%%%%%%%

We begin with the coefficient used throughout the paper.
This is not a new one-step coefficient.
It is the standard trace-distance contraction coefficient, often called the trace-norm Dobrushin contraction coefficient; see, for instance, \cite{HiaiRuskai2016,Hirche2024}.
Our use of it is product-level.
For a product of channels, \(2\kappa_{\rm tr}\) is exactly the trace-norm diameter of the evolved state space, and hence measures the residual memory of the input state.
The next sections show that the decay of this state-level memory coefficient is equivalent, up to universal constants, to map-level asymptotic replacement.

Let \(\matrices:=M_d(\mathbb C)\) denote the space of \(d\times d\) complex matrices.
The case \(d=1\) is trivial, so throughout this section we assume \(d\ge2\).
Let
\[
    \mcS_d
    :=
    \{\rho\in\matrices:\rho\ge0,\ \tr{\rho}=1\}
\]
be the state space.
We write \(\|\cdot\|_1\) for the trace norm.
The centered self-adjoint trace-zero subspace is
\[
    \tracezero
    :=
    \{X\in\matrices^{\rm sa}:\tr{X}=0\}.
\]
Here \(\matrices^{\rm sa}\) denotes the real vector space of self-adjoint matrices in \(\matrices\).
That is,
\[
    \matrices^{\rm sa}
    :=
    \{X\in\matrices:X=X^*\}.
\]
This is the natural centered space for trace-preserving dynamics.
Every difference of two density matrices lies in \(\tracezero\).
Every positive trace-preserving linear map sends \(\tracezero\) into itself.

\begin{definition}[Centered trace-Dobrushin coefficient]
\label{def:centered-trace-dobrushin}
    Let \(T:\matrices\to\matrices\) be positive and trace preserving.
    The centered trace-Dobrushin coefficient of \(T\) is
    \[
        \kappa_{\rm tr}(T)
        :=
        \sup_{X\in\tracezero\setminus\{0\}}
        \frac{\|T(X)\|_1}{\|X\|_1}.
    \]
    Equivalently, \(\kappa_{\rm tr}(T)\) is the \(1\to1\) operator norm of \(T\) restricted to the real Banach space \((\tracezero,\|\cdot\|_1)\).
\end{definition}

\begin{prop}
\label{prop:dobr}
    For every positive trace-preserving map \(T:\matrices\to\matrices\),
    \begin{equation}
    \label{eq:kappa-state-forms}
        \kappa_{\rm tr}(T)
        =
        \sup_{\substack{\rho,\sigma\in\mcS_d\\ \rho\ne\sigma}}
        \frac{\|T(\rho)-T(\sigma)\|_1}{\|\rho-\sigma\|_1}
        =
        \frac12
        \sup_{\rho,\sigma\in\mcS_d}
        \|T(\rho)-T(\sigma)\|_1 .
    \end{equation}
    Equivalently,
    \begin{equation}
    \label{eq:kappa-rank-one-form}
        \kappa_{\rm tr}(T)
        =
        \frac12
        \sup
        \left\{
            \|T(E-F)\|_1:
            E,F \text{ orthogonal rank-one projections}
        \right\}.
    \end{equation}
    Consequently \(0\le\kappa_{\rm tr}(T)\le1\).
    Moreover, \(\kappa_{\rm tr}(T)=0\) if and only if \(T\) is a replacement channel, namely \(T(X)=\tr{X}\tau\) for some \(\tau\in\mcS_d\).
\end{prop}

\begin{proof}
    If \(\rho,\sigma\in\mcS_d\), then \(\rho-\sigma\in\tracezero\).
    Hence
    \[
        \|T(\rho)-T(\sigma)\|_1
        \le
        \kappa_{\rm tr}(T)\|\rho-\sigma\|_1 .
    \]
    This proves
    \[
        \sup_{\substack{\rho,\sigma\in\mcS_d\\ \rho\ne\sigma}}
        \frac{\|T(\rho)-T(\sigma)\|_1}{\|\rho-\sigma\|_1}
        \le
        \kappa_{\rm tr}(T).
    \]

    Conversely, let \(X\in\tracezero\setminus\{0\}\).
    Write the positive-negative decomposition of the self-adjoint matrix \(X\) as
    \[
        X=X_+-X_-,
        \qquad
        X_\pm\ge0,
        \qquad
        X_+X_-=0.
    \]
    Since \(\tr{X}=0\), one has
    \[
        \tr{X_+}
        =
        \tr{X_-}
        =
        \frac12\|X\|_1 .
    \]
    Both traces are positive because \(X\ne0\) and \(\tr{X}=0\).
    Define
    \[
        \rho_+:=\frac{X_+}{\tr{X_+}},
        \qquad
        \rho_-:=\frac{X_-}{\tr{X_-}}.
    \]
    Then \(\rho_+,\rho_-\in\mcS_d\), their supports are orthogonal, and
    \[
        X
        =
        \frac{\|X\|_1}{2}
        (\rho_+-\rho_-).
    \]
    Since \(\|\rho_+-\rho_-\|_1=2\), it follows that
    \[
        \frac{\|T(X)\|_1}{\|X\|_1}
        =
        \frac12\|T(\rho_+)-T(\rho_-)\|_1
        =
        \frac{\|T(\rho_+)-T(\rho_-)\|_1}{\|\rho_+-\rho_-\|_1}.
    \]
    Taking the supremum over \(X\in\tracezero\setminus\{0\}\) gives the reverse inequality in the first identity of \eqref{eq:kappa-state-forms}.
    The same decomposition also gives
    \[
        \kappa_{\rm tr}(T)
        \le
        \frac12
        \sup_{\rho,\sigma\in\mcS_d}
        \|T(\rho)-T(\sigma)\|_1 .
    \]
    The opposite inequality follows from the first identity in \eqref{eq:kappa-state-forms} and from the bound \(\|\rho-\sigma\|_1\le2\) for states.
    This proves \eqref{eq:kappa-state-forms}.

    We next prove the rank-one formula.
    Let \(\rho_+\) and \(\rho_-\) be states with orthogonal supports.
    Choose spectral decompositions
    \[
        \rho_+
        =
        \sum_i p_iE_i,
        \qquad
        \rho_-
        =
        \sum_j q_jF_j,
    \]
    where \(E_i\) and \(F_j\) are rank-one projections and \(E_iF_j=0\) for all \(i,j\).
    Since \(\sum_i p_i=\sum_j q_j=1\), we have
    \[
        \rho_+-\rho_-
        =
        \sum_{i,j}p_iq_j(E_i-F_j).
    \]
    By convexity of the trace norm,
    \[
        \|T(\rho_+-\rho_-)\|_1
        \le
        \sup_{i,j}\|T(E_i-F_j)\|_1.
    \]
    Taking the supremum over all orthogonally supported \(\rho_+,\rho_-\in\mcS_d\) gives
    \[
        \kappa_{\rm tr}(T)
        \le
        \frac12
        \sup
        \left\{
            \|T(E-F)\|_1:
            E,F \text{ orthogonal rank-one projections}
        \right\}.
    \]
    The reverse inequality follows because every pair of orthogonal rank-one projections is a pair of states with trace distance \(2\).
    This proves \eqref{eq:kappa-rank-one-form}.

    Positivity and trace preservation imply trace-norm contractivity on self-adjoint matrices.
    Indeed, if \(Y=Y_+-Y_-\) is the positive-negative decomposition of a self-adjoint matrix \(Y\), then \(T(Y_+)\) and \(T(Y_-)\) are positive.
    Hence
    \[
        \|T(Y)\|_1
        \le
        \tr{T(Y_+)}+\tr{T(Y_-)}
        =
        \tr{Y_+}+\tr{Y_-}
        =
        \|Y\|_1 .
    \]
    Applying this to \(Y\in\tracezero\) gives \(0\le\kappa_{\rm tr}(T)\le1\).

    Finally, suppose \(\kappa_{\rm tr}(T)=0\).
    Then \(T(\rho)=T(\sigma)\) for all \(\rho,\sigma\in\mcS_d\).
    Fix \(\rho_0\in\mcS_d\) and set \(\tau:=T(\rho_0)\).
    Since \(T\) is positive and trace preserving, \(\tau\in\mcS_d\).
    If \(A\ge0\) and \(A\ne0\), then \(A/\tr{A}\in\mcS_d\), and therefore
    \[
        T(A)
        =
        \tr{A}\tau .
    \]
    The same formula holds for \(A=0\).
    By decomposing a self-adjoint matrix into its positive and negative parts, and then decomposing a general matrix into real and imaginary self-adjoint parts, linearity gives
    \[
        T(X)=\tr{X}\tau
        \qquad
        \text{for every }X\in\matrices .
    \]
    The converse implication is immediate.
\end{proof}

\begin{remark}
\label{rem:complex-tracezero-control}
	The coefficient \(\kappa_{\rm tr}(T)\) is defined on the real self-adjoint trace-zero space \(\tracezero\), because this is the space generated by differences of states.
	It also controls the complex trace-zero space up to a universal factor.
	Indeed, if \(Z\in\matrices\) satisfies \(\tr{Z}=0\), write
	\[
		Z=A+iB,
		\qquad
		A=A^*,\quad B=B^*,\quad \tr{A}=\tr{B}=0.
	\]
	Then
	\[
		\norm{A}_1\leq\norm{Z}_1,
		\qquad
		\norm{B}_1\leq\norm{Z}_1,
	\]
	and hence
	\[
		\norm{T(Z)}_1
		\leq
		\norm{T(A)}_1+\norm{T(B)}_1
		\leq
		2\kappa_{\rm tr}(T)\norm{Z}_1.
	\]
	This is why decay of \(\kappa_{\rm tr}\) gives induced \(1\to1\) convergence to replacement channels with only universal constants.
\end{remark}

Equation \eqref{eq:kappa-state-forms} gives the precise memory-loss interpretation of \(\kappa_{\rm tr}\).
The number \(2\kappa_{\rm tr}(T)\) is exactly the trace-norm diameter of the evolved state space:
\[
    \sup_{\rho,\sigma\in\mcS_d}
    \norm{T(\rho)-T(\sigma)}_1
    =
    2\kappa_{\rm tr}(T).
\]
Thus \(\kappa_{\rm tr}(T)=0\) exactly when \(T\) is a replacement channel.
For an inhomogeneous product \(\Phi_{t:s}\), the identity
\[
    \sup_{\rho,\sigma\in\mcS_d}
    \norm{\Phi_{t:s}(\rho)-\Phi_{t:s}(\sigma)}_1
    =
    2\kappa_{\rm tr}(\Phi_{t:s})
\]
says that \(\kappa_{\rm tr}(\Phi_{t:s})\) is the exact residual trace-norm memory after the product.
Consequently,
\[
    \kappa_{\rm tr}(\Phi_{t:s})\longrightarrow0
\]
is exactly uniform trace-norm loss of memory of the initial state.
The deterministic results below show that this state-level memory loss is equivalent, up to universal constants, to convergence of the product toward a family of replacement channels.
The key property that makes the coefficient suitable for products is the following submultiplicativity.

\begin{prop}
\label{prop:submult}
    If \(S,T:\matrices\to\matrices\) are positive and trace preserving, then
    \[
        \kappa_{\rm tr}(S\circ T)
        \le
        \kappa_{\rm tr}(S)\kappa_{\rm tr}(T).
    \]
\end{prop}

\begin{proof}
    Since \(T\) is positive, it maps self-adjoint matrices to self-adjoint matrices.
    Since \(T\) is trace preserving, it maps trace-zero matrices to trace-zero matrices.
    Hence \(T(\tracezero)\subseteq\tracezero\).

    Let \(X\in\tracezero\).
    If \(T(X)=0\), then the desired estimate is immediate.
    If \(T(X)\ne0\), then \(T(X)\in\tracezero\setminus\{0\}\), and therefore
    \[
        \|S(T(X))\|_1
        \le
        \kappa_{\rm tr}(S)\|T(X)\|_1
        \le
        \kappa_{\rm tr}(S)\kappa_{\rm tr}(T)\|X\|_1.
    \]
    The same inequality is also true when \(T(X)=0\).
    Taking the supremum over \(X\in\tracezero\setminus\{0\}\) gives the result.
\end{proof}

Consequently, if \(s<r<t\), then
\[
    \kappa_{\rm tr}(\Phi_{t:s})
    \le
    \kappa_{\rm tr}(\Phi_{t:r})\kappa_{\rm tr}(\Phi_{r:s}).
\]
This block submultiplicativity is the basic mechanism behind the deterministic and random product criteria developed below.

%%%%%%%%%%%%%%%%%%%%%%%%%%%%%%%%%%%%%%%%%%%%%%%%%%%%%%%%%%%%%%%%%%%%%%%%%%%%
%%%%%%%%%%%%%%%%%%%%%%%%%%%%%%%%%%%%%%%%%%%%%%%%%%%%%%%%%%%%%%%%%%%%%%%%%%%%

\subsection{Comparison with one-step contraction criteria}
\label{subsec:comparison-one-step-criteria}

%%%%%%%%%%%%%%%%%%%%%%%%%%%%%%%%%%%%%%%%%%%%%%%%%%%%%%%%%%%%%%%%%%%%%%%%%%%%
%%%%%%%%%%%%%%%%%%%%%%%%%%%%%%%%%%%%%%%%%%%%%%%%%%%%%%%%%%%%%%%%%%%%%%%%%%%%

The coefficient \(\kappa_{\rm tr}\) is a standard trace-distance contraction coefficient.
Our use of it is product-level.
For an inhomogeneous product \(\Phi_{t:s}\), the quantity \(\kappa_{\rm tr}(\Phi_{t:s})\) is the exact residual-memory coefficient of the whole block.
Its decay to zero is precisely uniform trace-norm loss of memory.
The replacement results in \Cref{sec:det} show that this state-level memory loss is equivalent, up to universal constants, to map-level asymptotic replacement.
This distinguishes the present approach from one-step sufficient criteria such as Markov--Dobrushin and Doeblin minorization.
Also, we may use certain results in this section (also see \Cref{subsec:bistochastic-hs-gaps} and \Cref{app:comparison-coefficients} in the Appendix) as sufficient criteria to obtain decay of  \(\kappa_{\rm tr}(\Phi_{t:s})\) that is needed in subsequent sections. 

In the notation often used for contraction coefficients, the trace-distance contraction coefficient is
\[
    \eta^{\rm Tr}(\Phi)
    :=
    \sup_{\substack{\rho,\sigma\in\mcS_d\\ \rho\ne\sigma}}
    \frac{\norm{\Phi(\rho)-\Phi(\sigma)}_1}{\norm{\rho-\sigma}_1}.
\]
By \Cref{prop:dobr}, one has \(\eta^{\rm Tr}(\Phi)=\kappa_{\rm tr}(\Phi)\).
The point of the present paper is not this one-step identity, but the use of the exact product coefficient \(\kappa_{\rm tr}(\Phi_{t:s})\) as the intrinsic memory coefficient for an entire block.

Let
\[
    \mcP_1
    :=
    \left\{
        P\in\matrices:
        P=P^*=P^2,\ \tr{P}=1
    \right\}
\]
denote the set of rank-one projections.
The state-level Markov--Dobrushin coefficient is based on a common lower bound for the family \(\{\Phi(P):P\in\mcP_1\}\).

\begin{definition}
\label{def:state-md-minorization}
    Let \(\Phi:\matrices\to\matrices\) be CPTP.
    Define
    \[
        \alpha_{\rm MD}(\Phi)
        :=
        \sup
        \left\{
            \tr{B}:
            B\ge0,\ 
            B\le \Phi(P)
            \text{ for every }P\in\mcP_1
        \right\}.
    \]
    The supremum is attained in finite dimension.
    We also write
    \[
        \eta_{\rm MD}(\Phi)
        :=
        1-\alpha_{\rm MD}(\Phi).
    \]
\end{definition}
Thus \(\alpha_{\rm MD}(\Phi)\) is the mass of the largest common lower bound, while \(\eta_{\rm MD}(\Phi)=1-\alpha_{\rm MD}(\Phi)\) is the resulting one-step contraction bound.
The number \(\eta_{\rm MD}(\Phi)\) is not, in general, the exact trace-distance contraction coefficient.

This is a state-level minorization condition.
It asks whether all output states \(\Phi(P)\), with \(P\) pure, share a common positive part.
In the inhomogeneous-channel framework of \cite{Souissi2026}, the same idea is formulated through a trace-maximal common lower bound \(B_\Phi\in I^{\rm Tr}_{\rm MD}(\Phi)\), and the corresponding estimate can be written as
\[
    \norm{\Phi(\rho)-\Phi(\sigma)}_1
    \le
   \left(1-\tr{B_\Phi}\right)\norm{\rho-\sigma}_1.
\]
The homogeneous MPS application in \cite{Souissi_2025} uses positivity of such a Markov--Dobrushin lower bound to obtain exponential mixing and a thermodynamic limit.

\begin{prop}
\label{prop:md-bound}
    For every CPTP map \(\Phi\),
    \[
        \kappa_{\rm tr}(\Phi)
        \le
        \eta_{\rm MD}(\Phi)
        =
        1-\alpha_{\rm MD}(\Phi).
    \]
\end{prop}

\begin{proof}
    Let \(B\ge0\) be an admissible common lower bound and put \(a:=\tr{B}\).
    Since \(B\le\Phi(P)\) and \(\tr{\Phi(P)}=1\) for every \(P\in\mcP_1\), one has \(0\le a\le1\).
    If \(a=1\), then \(B\) is a state.
    Since \(B\le\Phi(P)\) and \(\tr{\Phi(P)}=1\), we have \(\Phi(P)=B\) for every \(P\in\mcP_1\).
    Hence \(\Phi\) is constant on pure states, therefore constant on \(\mcS_d\), and by linearity it is a replacement channel.
    Thus \(\kappa_{\rm tr}(\Phi)=0\).
    
    Assume \(a<1\).
    If \(\rho\in\mcS_d\) has spectral decomposition \(\rho=\sum_i p_iP_i\), then \(B\le\Phi(P_i)\) for each \(i\), and therefore \(\Phi(\rho)\ge B\).
    Thus
    \[
        \omega_\rho
        :=
        \frac{\Phi(\rho)-B}{1-a}
    \]
    is a state.

    Let \(X\in\tracezero\setminus\left\{0\right\}\).
    As in the proof of \Cref{prop:dobr}, write
    \[
        X
        =
        \frac{\norm{X}_1}{2}\left(\rho_+-\rho_-\right),
    \]
    where \(\rho_+,\rho_-\in\mcS_d\).
    Applying the preceding decomposition to \(\rho_+\) and \(\rho_-\), we obtain
    \[
        \Phi(\rho_+)-\Phi(\rho_-)
        =
        \left(1-a\right)\left(\omega_{\rho_+}-\omega_{\rho_-}\right).
    \]
    Hence
    \[
        \norm{\Phi(\rho_+)-\Phi(\rho_-)}_1
        \le
        2\left(1-a\right).
    \]
    Therefore
    \[
        \frac{\norm{\Phi(X)}_1}{\norm{X}_1}
        =
        \frac12\norm{\Phi(\rho_+)-\Phi(\rho_-)}_1
        \le
        1-a.
    \]
    Taking the supremum over \(X\in\tracezero\setminus\left\{0\right\}\) gives \(\kappa_{\rm tr}(\Phi)\le1-\tr{B}\).
    Optimizing over admissible \(B\) proves the claim.
\end{proof}

Thus, every state-level Markov--Dobrushin lower bound gives a valid one-step upper bound on \(\kappa_{\rm tr}\).
The converse is false.
A channel may contract trace distance strictly without admitting a nonzero common lower bound for all pure-state outputs.
For example, the qubit amplitude-damping channel with parameter \(0<\gamma<1\),
\[
    \Gamma_\gamma(\rho)
    =
    K_0\rho K_0^*
    +
    K_1\rho K_1^*,
    \qquad
    K_0=
    \begin{pmatrix}
        1&0\\
        0&\sqrt{1-\gamma}
    \end{pmatrix},
    \qquad
    K_1=
    \begin{pmatrix}
        0&\sqrt{\gamma}\\
        0&0
    \end{pmatrix},
\]
satisfies \(\kappa_{\rm tr}(\Gamma_\gamma)=\sqrt{1-\gamma}<1\), but \(\alpha_{\rm MD}(\Gamma_\gamma)=0\).
Indeed, \(\Gamma_\gamma(|0\rangle\langle0|)=|0\rangle\langle0|\), so any common lower bound \(B\) must have the form \(B=b|0\rangle\langle0|\) with \(b\ge0\).
For a pure input \(|\psi\rangle=\alpha|0\rangle+\beta|1\rangle\), the condition \(B\le\Gamma_\gamma(|\psi\rangle\langle\psi|)\) forces \(b\le\gamma|\beta|^2\).
Letting \(|\beta|\downarrow0\) gives \(b=0\).
Thus \(\alpha_{\rm MD}(\Gamma_\gamma)=0\), even though \(\kappa_{\rm tr}(\Gamma_\gamma)=\sqrt{1-\gamma}<1\).

Other contraction mechanisms fit different geometries.
CP-order Doeblin coefficients impose an algebraic minorization of the whole channel, rather than only a common lower bound on pure-state outputs.
Hilbert--Birkhoff projective coefficients are adapted to cone-preserving maps and are especially useful beyond the trace-preserving setting.
Entropy, Riemannian, and \(\chi^2\) coefficients are adapted to faithful reference-state geometries.
Quantum Wasserstein--Dobrushin coefficients control spatial influence in many-body systems \cite{BakshiLiuMoitraTang2025}.
Unconstrained diamond-norm analogues are not the right object for ordinary state-memory loss, because an untouched ancillary reference system can preserve trace-zero information even when the system channel itself is replacing.
We discuss these comparisons in detail in \Cref{app:comparison-coefficients}.

%%%%%%%%%%%%%%%%%%%%%%%%%%%%%%%%%%%%%%%%%%%%%%%%%%%%%%%%%%%%%%%%%%%%%%%%%%%%
%%%%%%%%%%%%%%%%%%%%%%%%%%%%%%%%%%%%%%%%%%%%%%%%%%%%%%%%%%%%%%%%%%%%%%%%%%%%

\section{Deterministic inhomogeneous products}
\label{sec:det}

%%%%%%%%%%%%%%%%%%%%%%%%%%%%%%%%%%%%%%%%%%%%%%%%%%%%%%%%%%%%%%%%%%%%%%%%%%%%
%%%%%%%%%%%%%%%%%%%%%%%%%%%%%%%%%%%%%%%%%%%%%%%%%%%%%%%%%%%%%%%%%%%%%%%%%%%%

We now pass from a single channel to deterministic products of channels.
Let \((\Phi_n)_{n\in\mcI}\) be a deterministic sequence of CPTP maps on \(\matrices\).
The index set is either one-sided, \(\mcI=\mbN_0\), or two-sided, \(\mcI=\mbZ\), depending on whether only forward products or pullback limits are being discussed.

%%%%%%%%%%%%%%%%%%%%%%%%%%%%%%%%%%%%%%%%%%%%%%%%%%%%%%%%%%%%%%%%%%%%%%%%%%%%
%%%%%%%%%%%%%%%%%%%%%%%%%%%%%%%%%%%%%%%%%%%%%%%%%%%%%%%%%%%%%%%%%%%%%%%%%%%%

\subsection{Product notation and fixed-start replacement}
\label{sec:det-product-criterion}

%%%%%%%%%%%%%%%%%%%%%%%%%%%%%%%%%%%%%%%%%%%%%%%%%%%%%%%%%%%%%%%%%%%%%%%%%%%%
%%%%%%%%%%%%%%%%%%%%%%%%%%%%%%%%%%%%%%%%%%%%%%%%%%%%%%%%%%%%%%%%%%%%%%%%%%%%

For \(s<t\), we use the chronological convention
\[
	\Phi_{t:s}
	:=
	\Phi_{t-1}\circ\Phi_{t-2}\circ\cdots\circ\Phi_s,
	\qquad
	\Phi_{s:s}
	:=
	\operatorname{id}.
\]
Thus \(\Phi_{t:s}\) maps a state prepared at time \(s\) to its state at time \(t\).
The cocycle identity is
\[
	\Phi_{t:s}
	=
	\Phi_{t:u}\circ\Phi_{u:s},
	\qquad
	s<u<t.
\]
We write
\[
	\kappa_{t:s}
	:=
	\kappa_{\rm tr}\bigl(\Phi_{t:s}\bigr).
\]
By \Cref{prop:submult}, the product coefficients satisfy
\[
	\kappa_{t:s}
	\le
	\kappa_{t:u}\kappa_{u:s},
	\qquad
	s<u<t.
\]
This is the basic deterministic submultiplicativity used throughout the sequel.

For \(\zeta\in\mcS_d\), let
\[
	R_\zeta(X)
	:=
	\tr{X}\zeta.
\]
We call \(R_\zeta\) the replacement channel with output \(\zeta\).
If the state \(\zeta\) depends on the time window, then the corresponding replacement channel is a moving replacement channel.

The first point is that \(\kappa_{t:s}\) is the exact state-level memory coefficient of the product.
Since \(\Phi_{t:s}\) is CPTP, \Cref{prop:dobr} applied to \(T=\Phi_{t:s}\) gives
\begin{equation}
\label{eq:det-product-diameter-form}
	\sup_{\rho,\sigma\in\mcS_d}
	\norm{\Phi_{t:s}(\rho)-\Phi_{t:s}(\sigma)}_1
	=
	2\kappa_{t:s}.
\end{equation}
Equivalently,
\[
	\kappa_{t:s}
	=
	\dfrac{1}{2}
	\operatorname{diam}_1\bigl(\Phi_{t:s}(\mcS_d)\bigr),
\]
where
\[
	\operatorname{diam}_1(A)
	:=
	\sup_{x,y\in A}\norm{x-y}_1.
\]
Thus, the product coefficient is one half of the trace-norm diameter of the evolved state space.

\begin{definition}[Asymptotic fixed-start trace-memory loss]
\label{def:fixed-start-memory-loss}
	Fix an initial time \(s\in\mcI\).
	The product sequence has asymptotic trace-memory loss from \(s\) if
	\[
		\lim_{t\to\infty}
		\sup_{\rho,\sigma\in\mcS_d}
		\norm{\Phi_{t:s}(\rho)-\Phi_{t:s}(\sigma)}_1
		=
		0.
	\]
\end{definition}

By \eqref{eq:det-product-diameter-form}, for fixed \(s\),
\Cref{def:fixed-start-memory-loss} is equivalent to
\[
	\lim_{t\to\infty}\kappa_{t:s}=0.
\]
Thus \(\kappa_{t:s}\to0\) is exactly uniform trace-norm loss of memory of the initial state.
The next estimate upgrades this state-level statement to an induced \(1\to1\) map-norm statement by comparing the whole product with a replacement channel.

For a reference state \(\tau\in\mcS_d\), define
\[
	Z_{t:s}^{(\tau)}
	:=
	\Phi_{t:s}(\tau),
	\qquad
	R_{t:s}^{(\tau)}(X)
	:=
	\tr{X}Z_{t:s}^{(\tau)}.
\]
Thus \(R_{t:s}^{(\tau)}\) is the replacement channel centered at the evolved reference state.

For a linear map \(L:\matrices\to\matrices\), we write
\[
	\norm{L}_{1\to1}
	:=
	\sup_{X\in\matrices\setminus\{0\}}
	\frac{\norm{L(X)}_1}{\norm{X}_1},
\]
where \(\norm{\,\cdot\,}_1\) denotes the trace norm, equivalently the Schatten \(1\)-norm, on \(\matrices\).

\begin{prop}
\label{prop:det-radius-replacement}
	For every \(s<t\), every reference state \(\tau\in\mcS_d\), and every \(\zeta\in\mcS_d\),
	\[
		\kappa_{t:s}
		\le
		\norm{\Phi_{t:s}-R_\zeta}_{1\to1}
	\qquad\text{and}\qquad
		\norm{\Phi_{t:s}-R_{t:s}^{(\tau)}}_{1\to1}
		\le
		4\kappa_{t:s}.
	\]
	Consequently,
	\[
		\kappa_{t:s}
		\le
		\inf_{\zeta\in\mcS_d}
		\norm{\Phi_{t:s}-R_\zeta}_{1\to1}
		\le
		\norm{\Phi_{t:s}-R_{t:s}^{(\tau)}}_{1\to1}
		\le
		4\kappa_{t:s}.
	\]
\end{prop}

\begin{proof}
	First fix \(\zeta\in\mcS_d\).
	If \(X=X^*\), \(\tr{X}=0\), and \(X\ne0\), then \(R_\zeta(X)=0\).
	Therefore
	\[
		\norm{\Phi_{t:s}-R_\zeta}_{1\to1}
		\ge
		\sup_{\substack{X=X^*\\ \tr{X}=0\\ X\ne0}}
		\frac{\norm{\Phi_{t:s}(X)}_1}{\norm{X}_1}
		=
		\kappa_{t:s}.
	\]
	This proves the lower bound.

	We now prove the upper bound for the reference-state replacement.
	Fix \(\tau\in\mcS_d\).
	Since \(\Phi_{t:s}\) is CPTP, we have
	\[
		Z_{t:s}^{(\tau)}
		=
		\Phi_{t:s}(\tau)
		\in\mcS_d.
	\]
	For \(X\in\matrices\), set
	\[
		Y:=X-\tr{X}\tau.
	\]
	Then \(\tr{Y}=0\), and
	\[
		\Phi_{t:s}(X)-R_{t:s}^{(\tau)}(X)
		=
		\Phi_{t:s}(X)-\tr{X}\Phi_{t:s}(\tau)
		=
		\Phi_{t:s}(Y).
	\]
	By \Cref{rem:complex-tracezero-control}, applied to the positive trace-preserving map \(\Phi_{t:s}\),
	\[
		\norm{\Phi_{t:s}(Y)}_1
		\le
		2\kappa_{t:s}\norm{Y}_1.
	\]
	Moreover,
	\[
		\norm{Y}_1
		\le
		\norm{X}_1+|\tr{X}|\,\norm{\tau}_1.
	\]
	Since \(\tau\in\mcS_d\), we have \(\norm{\tau}_1=1\).
	By trace duality,
	\[
		|\tr{X}|
		=
		|\tr{IX}|
		\le
		\norm{I}_\infty\norm{X}_1
		=
		\norm{X}_1.
	\]
	Hence
	\[
		\norm{Y}_1
		\le
		2\norm{X}_1.
	\]
	Combining the last two estimates gives
	\[
		\norm{\Phi_{t:s}(X)-R_{t:s}^{(\tau)}(X)}_1
		\le
		4\kappa_{t:s}\norm{X}_1.
	\]
	Taking the supremum over \(X\ne0\) proves
	\[
		\norm{\Phi_{t:s}-R_{t:s}^{(\tau)}}_{1\to1}
		\le
		4\kappa_{t:s}.
	\]

	Finally, since \(R_{t:s}^{(\tau)}=R_\zeta\) with \(\zeta=\Phi_{t:s}(\tau)\), the middle inequality
	\[
		\inf_{\zeta\in\mcS_d}
		\norm{\Phi_{t:s}-R_\zeta}_{1\to1}
		\le
		\norm{\Phi_{t:s}-R_{t:s}^{(\tau)}}_{1\to1}
	\]
	is immediate.
	The claimed chain of inequalities follows.
\end{proof}

\begin{remark}
	The factor \(4\) comes from two elementary losses.
	First, \(\kappa_{\rm tr}\) is defined on the real self-adjoint trace-zero space, while the induced \(1\to1\) norm is taken over all complex matrices.
	This gives the factor \(2\) in \Cref{rem:complex-tracezero-control}.
	Second, replacing \(X\) by the centered matrix \(X-\tr{X}\tau\) costs another factor at most \(2\).
	If the \(1\to1\) norm were restricted to self-adjoint inputs, the first loss would disappear and the constant \(4\) could be replaced by \(2\).
\end{remark}

The product-to-replacement estimate is the technical bridge between the state-level memory coefficient and the channel-level replacement formulation.
It also shows that different reference states give asymptotically the same moving replacement family whenever \(\kappa_{t:s}\to0\).
Indeed, for states \(a,b\in\mcS_d\), define \(R_a(X):=\tr{X}a\) and \(R_b(X):=\tr{X}b\).
Then
\[
	\norm{R_a-R_b}_{1\to1}
	=
	\norm{a-b}_1.
\]
The upper bound follows from
\[
	\norm{R_a(X)-R_b(X)}_1
	=
	|\tr{X}|\,\norm{a-b}_1
	\le
	\norm{X}_1\norm{a-b}_1.
\]
The reverse inequality is attained by taking \(X\) to be any state.
Therefore, for \(\tau,\tau'\in\mcS_d\),
\[
	\norm{R_{t:s}^{(\tau)}-R_{t:s}^{(\tau')}}_{1\to1}
	=
	\norm{\Phi_{t:s}(\tau)-\Phi_{t:s}(\tau')}_1
	\le
	2\kappa_{t:s},
\]
where the last inequality follows from \eqref{eq:det-product-diameter-form}.
Thus, the choice of reference state affects the center only by an error controlled by the same product coefficient.

Recall \Cref{def:strong-asymptotic-replacement-from-s}.
We now prove the fixed-start replacement theorem stated in the introduction.

\detReplacementApproximationTheorem*

\begin{proof}
	Fix the initial time \(s\in\mcI\).
	For every \(t>s\) and every \(\tau\in\mcS_d\), \Cref{prop:det-radius-replacement} gives
	\[
		\kappa_{t:s}
		\le
		\norm{\Phi_{t:s}-R_{t:s}^{(\tau)}}_{1\to1}
		\le
		4\kappa_{t:s}.
	\]
	This proves the quantitative estimate in the theorem.

	We now prove the equivalence of the four statements.
	Assume first that \(\kappa_{t:s}\to0\) as $t \to\infty$
	Then the upper bound gives, for every reference state \(\tau\in\mcS_d\),
	\[
		\norm{\Phi_{t:s}-R_{t:s}^{(\tau)}}_{1\to1}
		\le
		4\kappa_{t:s}
		\longrightarrow0.
	\]
	Hence statement \(1\) implies statement \(3\).
	In particular, statement \(1\) also implies statement \(4\).

	Moreover, taking any fixed \(\tau\in\mcS_d\) and setting \(\eta_{t:s}:=\Phi_{t:s}(\tau)\), we have \(R_{t:s}^{(\tau)}(X)=\tr{X}\eta_{t:s}\).
	The same convergence therefore shows that the products are strongly asymptotically replacing from \(s\).
	Thus statement \(1\) implies statement \(2\).

	Statement \(3\) trivially implies statement \(4\).

	Assume next that statement \(4\) holds for some reference state \(\tau\in\mcS_d\).
	Then the lower bound in \Cref{prop:det-radius-replacement} gives
	\[
		\kappa_{t:s}
		\le
		\norm{\Phi_{t:s}-R_{t:s}^{(\tau)}}_{1\to1}
		\longrightarrow0.
	\]
	Hence statement \(4\) implies statement \(1\).

	Finally, assume statement \(2\).
	Then there are states \(\eta_{t:s}\in\mcS_d\) such that, with \(S_{t:s}(X):=\tr{X}\eta_{t:s}\),
	one has
	\[
		\norm{\Phi_{t:s}-S_{t:s}}_{1\to1}
		\longrightarrow0.
	\]
	For each \(t>s\), the map \(S_{t:s}\) is the replacement channel \(R_{\eta_{t:s}}\).
	Applying the lower bound in \Cref{prop:det-radius-replacement} with \(\zeta=\eta_{t:s}\), we obtain
	\[
		\kappa_{t:s}
		\le
		\norm{\Phi_{t:s}-S_{t:s}}_{1\to1}
		\longrightarrow0.
	\]
	Thus statement \(2\) implies statement \(1\).

	The four statements are therefore equivalent.
\end{proof}

%%%%%%%%%%%%%%%%%%%%%%%%%%%%%%%%%%%%%%%%%%%%%%%%%%%%%%%%%%%%%%%%%%%%%%%%%%%%
%%%%%%%%%%%%%%%%%%%%%%%%%%%%%%%%%%%%%%%%%%%%%%%%%%%%%%%%%%%%%%%%%%%%%%%%%%%%

\subsection{Pullback boundary states and moving replacement channels}
\label{sec:det-pullback}

%%%%%%%%%%%%%%%%%%%%%%%%%%%%%%%%%%%%%%%%%%%%%%%%%%%%%%%%%%%%%%%%%%%%%%%%%%%%
%%%%%%%%%%%%%%%%%%%%%%%%%%%%%%%%%%%%%%%%%%%%%%%%%%%%%%%%%%%%%%%%%%%%%%%%%%%%

Two-sided inhomogeneous products have a natural pullback viewpoint: the terminal time is fixed, while the initial time is sent to the remote past.
If products ending at \(t\) forget remote initial states in trace distance, then the state observed at time \(t\) becomes independent of the state inserted far in the past.
The resulting time-indexed state is a pullback boundary state, or entrance law, and the corresponding channel-level limit is a moving replacement family
\[
	R_t(X)=\tr{X}\rho_t .
\]

Throughout this subsection, assume that the sequence \((\Phi_n)_{n\in\mbZ}\) is indexed by all integers.
For \(s<t\), recall that
\[
	\Phi_{t:s}
	=
	\Phi_{t-1}\circ\Phi_{t-2}\circ\cdots\circ\Phi_s .
\]

\begin{definition}[Pullback boundary state]
\label{def:pullback-boundary-state}
	A pullback boundary state, also called an entrance law, for \((\Phi_n)_{n\in\mbZ}\) is a sequence of states
	\[
		(\rho_t)_{t\in\mbZ}\subset\mcS_d
	\]
	such that
	\[
		\rho_{t+1}=\Phi_t(\rho_t)
		\qquad\text{for all }t\in\mbZ.
	\]
	Equivalently, for every \(s<t\),
	\[
		\rho_t=\Phi_{t:s}(\rho_s).
	\]
\end{definition}

The term entrance law is standard in the theory of time-inhomogeneous Markov processes.
The phrase pullback boundary state emphasizes the construction that will arise below: under pullback memory loss, the state \(\rho_t\) is obtained as the limit of states transported from the remote past.

\begin{definition}[Pullback trace-memory loss]
\label{def:pullback-memory-loss}
	The sequence \((\Phi_n)_{n\in\mbZ}\) has pullback trace-memory loss if, for every fixed terminal time \(t\in\mbZ\),
	\[
		\lim_{s\to-\infty}\kappa_{t:s}=0.
	\]
\end{definition}

By \Cref{prop:dobr}, applied to the product \(\Phi_{t:s}\),
\[
	\sup_{\rho,\sigma\in\mcS_d}
	\norm{\Phi_{t:s}(\rho)-\Phi_{t:s}(\sigma)}_1
	=
	2\kappa_{t:s}.
\]
Thus, pullback trace-memory loss is exactly the statement that two arbitrary states inserted farther and farther in the past become indistinguishable when observed at the fixed terminal time \(t\).

The next result is the state-level deterministic pullback replacement theorem.
It says that pullback decay of the exact trace-Dobrushin coefficient is precisely convergence, uniformly over input states, to a unique moving center \(\rho_t\).
The induced \(1\to1\) replacement statement follows afterward from \Cref{prop:det-radius-replacement}.

\begin{prop}
\label{prop:pullback-entrance-law}
	Let \((\Phi_n)_{n\in\mbZ}\) be a two-sided deterministic sequence of CPTP maps.
	The following are equivalent.

	\begin{enumerate}
		\item \textbf{Pullback trace-memory loss.}
		For every fixed terminal time \(t\in\mbZ\), we have that \(\kappa_{t:s}\to 0\) as \(s\to-\infty\).
	
		\item \textbf{State-level pullback convergence.}
		There exists a pullback boundary state \((\rho_t)_{t\in\mbZ}\subset\mcS_d\)
		such that, for every fixed terminal time \(t\in\mbZ\),
		\[
			\lim_{s\to-\infty} \,
			\sup_{\sigma\in\mcS_d}
			\norm{\Phi_{t:s}(\sigma)-\rho_t}_1
			=
			0 .
		\]
	\end{enumerate}
	When these conditions hold, the pullback boundary state is unique.
	Moreover, for every reference state \(\tau\in\mcS_d\) and every \(t\in\mbZ\),
	\[
		\rho_t
		=
		\lim_{s\to-\infty}\Phi_{t:s}(\tau),
	\]
	and this limit is independent of \(\tau\).
    For every \(s<t\), one also has
	\[
		\kappa_{t:s}
		\le
		\sup_{\sigma\in\mcS_d}
		\norm{\Phi_{t:s}(\sigma)-\rho_t}_1
		\le
		2\kappa_{t:s}.
	\]
	Consequently, for every fixed terminal time \(t\), pullback trace-memory loss is equivalent to uniform state-level convergence of the evolved state space \(\Phi_{t:s}(\mcS_d)\) to the moving center \(\rho_t\) as \(s\to-\infty\).
\end{prop}

\begin{proof}
	Assume first that \(\kappa_{t:s}\to0\) as \(s\to-\infty\), for every fixed terminal time \(t\in\mbZ\).
	Fix once and for all a reference state \(\tau_0\in\mcS_d\).
	For each \(t\in\mbZ\) and each \(s<t\), set \(z_{t,s}:=\Phi_{t:s}(\tau_0)\).
	We first show that, for each fixed \(t\), the family \((z_{t,s})_{s<t}\) is Cauchy in trace norm as \(s\to-\infty\).
	Let \(r<s<t\).
	By the cocycle identity,
	\[
		\Phi_{t:r}(\tau_0)
		=
		\Phi_{t:s}\bigl(\Phi_{s:r}(\tau_0)\bigr).
	\]
	Both \(\Phi_{s:r}(\tau_0)\) and \(\tau_0\) are states.
	Therefore, by \Cref{prop:dobr} applied to the product \(\Phi_{t:s}\),
	\[
		\norm{\Phi_{t:r}(\tau_0)-\Phi_{t:s}(\tau_0)}_1
		\le
		2\kappa_{t:s}.
	\]
	Now let \(\varepsilon>0\).
	Choose \(s_0<t\) such that \(2\kappa_{t:u}<\varepsilon\) for all  \(u\le s_0\).
	If \(r,s\le s_0\), then after relabelling we may assume \(r<s<t\), and the preceding estimate gives
	\[
		\norm{z_{t,r}-z_{t,s}}_1<\varepsilon.
	\]
	Thus \((z_{t,s})_{s<t}\) is Cauchy as \(s\to-\infty\).
	Since \(\mcS_d\) is complete in trace norm, the limit
	\[
		\rho_t
		:=
		\lim_{s\to-\infty}\Phi_{t:s}(\tau_0)
	\]
	exists and belongs to \(\mcS_d\).
	The limit is independent of the reference state.
	Indeed, if \(\tau\in\mcS_d\), then
	\[
		\norm{\Phi_{t:s}(\tau)-\Phi_{t:s}(\tau_0)}_1
		\le
		2\kappa_{t:s}
		\longrightarrow0.
	\]
	Hence, for every \(\tau\in\mcS_d\), we get \(\rho_t
		=
		\lim_{s\to-\infty}\Phi_{t:s}(\tau)\).

	We now show that \((\rho_t)_{t\in\mbZ}\) is a pullback boundary state.
	By continuity of the CPTP map \(\Phi_t\),
	\[
		\Phi_t(\rho_t)
		=
		\lim_{s\to-\infty}\Phi_t\bigl(\Phi_{t:s}(\tau_0)\bigr)
		=
		\lim_{s\to-\infty}\Phi_{t+1:s}(\tau_0)
		=
		\rho_{t+1}.
	\]
	Thus \(\rho_{t+1}=\Phi_t(\rho_t)\), for all  \(t\in\mbZ\). 
	Iterating this identity gives, for every \(s<t\),
	\[
		\rho_t=\Phi_{t:s}(\rho_s).
	\]
	Therefore, for every \(s<t\) and every \(\sigma\in\mcS_d\),
	\[
		\norm{\Phi_{t:s}(\sigma)-\rho_t}_1
		=
		\norm{\Phi_{t:s}(\sigma)-\Phi_{t:s}(\rho_s)}_1
		\le
		2\kappa_{t:s}.
	\]
	Taking the supremum over \(\sigma\in\mcS_d\) gives
	\[
		\sup_{\sigma\in\mcS_d}
		\norm{\Phi_{t:s}(\sigma)-\rho_t}_1
		\le
		2\kappa_{t:s}.
	\]
	For fixed \(t\), the right-hand side tends to zero as \(s\to-\infty\).
	Hence statement \(1\) implies statement \(2\).

	We next prove the lower estimate.
	This estimate is valid for any state \(\eta\in\mcS_d\).
	By \Cref{prop:dobr},
	\[
		2\kappa_{t:s}
		=
		\sup_{\alpha,\beta\in\mcS_d}
		\norm{\Phi_{t:s}(\alpha)-\Phi_{t:s}(\beta)}_1.
	\]
	For any \(\alpha,\beta\in\mcS_d\), the triangle inequality gives
	\[
		\norm{\Phi_{t:s}(\alpha)-\Phi_{t:s}(\beta)}_1
		\le
		\norm{\Phi_{t:s}(\alpha)-\eta}_1
		+
		\norm{\Phi_{t:s}(\beta)-\eta}_1.
	\]
	Hence
	\[
		2\kappa_{t:s}
		\le
		2
		\sup_{\sigma\in\mcS_d}
		\norm{\Phi_{t:s}(\sigma)-\eta}_1,
	\]
	and therefore
	\[
		\kappa_{t:s}
		\le
		\sup_{\sigma\in\mcS_d}
		\norm{\Phi_{t:s}(\sigma)-\eta}_1.
	\]
	Applying this with \(\eta=\rho_t\) and combining it with the upper estimate above gives, for every \(s<t\),
	\[
		\kappa_{t:s}
		\le
		\sup_{\sigma\in\mcS_d}
		\norm{\Phi_{t:s}(\sigma)-\rho_t}_1
		\le
		2\kappa_{t:s}.
	\]

	Conversely, suppose that there exists a pullback boundary state \((\rho_t)_{t\in\mbZ}\subset\mcS_d\) such that, for every fixed terminal time \(t\),
	\[
		\sup_{\sigma\in\mcS_d}
		\norm{\Phi_{t:s}(\sigma)-\rho_t}_1
		\longrightarrow0
		\qquad
		\text{as }s\to-\infty.
	\]
	The lower estimate, applied with \(\eta=\rho_t\), gives
	\[
		\kappa_{t:s}
		\le
		\sup_{\sigma\in\mcS_d}
		\norm{\Phi_{t:s}(\sigma)-\rho_t}_1
		\longrightarrow0.
	\]
	Thus statement \(2\) implies statement \(1\).

	It remains to prove uniqueness.
	Let \((\eta_t)_{t\in\mbZ}\) be another pullback boundary state.
	Then, for every \(s<t\), we have that \(\eta_t=\Phi_{t:s}(\eta_s)\).
	Therefore,
	\[
		\norm{\eta_t-\rho_t}_1
		=
		\norm{\Phi_{t:s}(\eta_s)-\rho_t}_1
		\le
		\sup_{\sigma\in\mcS_d}
		\norm{\Phi_{t:s}(\sigma)-\rho_t}_1.
	\]
	Letting \(s\to-\infty\) gives \(\eta_t=\rho_t\).
	Hence, the pullback boundary state is unique.

	The reference-state formula
	\[
		\rho_t
		=
		\lim_{s\to-\infty}\Phi_{t:s}(\tau)
	\]
	for every \(\tau\in\mcS_d\), and the independence of \(\tau\), were proved in the construction above.
	This completes the proof.
\end{proof}

The preceding proposition identifies the pullback limit at the level of states.
We now upgrade this state-level convergence to convergence of the whole product,
as a linear map, toward a canonical moving replacement channel.

\detPullbackCanonicalReplacement*

\begin{proof}
	Assume first that, for every fixed terminal time \(t\in\mbZ\),
	\[
		\lim_{s\to-\infty}\kappa_{t:s}=0.
	\]
	Then by \Cref{prop:pullback-entrance-law}, there is a unique pullback boundary state
	\[
		(\rho_t)_{t\in\mbZ}\subset\mcS_d.
	\]
	Moreover, for every \(t\in\mbZ\) and every reference state \(\tau\in\mcS_d\),
	\[
		\rho_t
		=
		\lim_{s\to-\infty}\Phi_{t:s}(\tau),
	\]
	and this limit is independent of \(\tau\).

	Define
	\[
		R_t(X):=\tr{X}\rho_t.
	\]
	Since \((\rho_t)_{t\in\mbZ}\) is a pullback boundary state, for every \(s<t\),
	\[
		\rho_t=\Phi_{t:s}(\rho_s).
	\]
	Hence
	\[
		R_t=R_{t:s}^{(\rho_s)}.
	\]
	Applying \Cref{prop:det-radius-replacement} to the product \(\Phi_{t:s}\) with reference state \(\rho_s\), we obtain
	\[
		\norm{\Phi_{t:s}-R_t}_{1\to1}
		\le
		4\kappa_{t:s}.
	\]
	Also, applying the lower bound in \Cref{prop:det-radius-replacement} with \(\zeta=\rho_t\), we obtain
	\[
		\kappa_{t:s}
		\le
		\norm{\Phi_{t:s}-R_t}_{1\to1}.
	\]
	Thus, for every \(s<t\),
	\[
		\kappa_{t:s}
		\le
		\norm{\Phi_{t:s}-R_t}_{1\to1}
		\le
		4\kappa_{t:s}.
	\]
	In particular, for every fixed terminal time \(t\),
	\[
		\lim_{s\to-\infty}
		\norm{\Phi_{t:s}-R_t}_{1\to1}
		=
		0.
	\]
	This proves statement \(2\).

	Conversely, suppose that statement \(2\) holds.
	Thus there are states \((\eta_t)_{t\in\mbZ}\subset\mcS_d\) such that, with
	\[
		S_t(X):=\tr{X}\eta_t,
	\]
	one has, for every fixed terminal time \(t\),
	\[
		\lim_{s\to-\infty}
		\norm{\Phi_{t:s}-S_t}_{1\to1}
		=
		0.
	\]
	By the lower bound in \Cref{prop:det-radius-replacement}, applied to the product \(\Phi_{t:s}\) with \(\zeta=\eta_t\), we have
	\[
		\kappa_{t:s}
		\le
		\norm{\Phi_{t:s}-S_t}_{1\to1}.
	\]
	Hence, for every fixed terminal time \(t\),
	\[
		\lim_{s\to-\infty}\kappa_{t:s}=0.
	\]
	This proves statement \(1\).

	We now identify the replacement centers and prove their uniqueness.
	By statement \(1\) and \Cref{prop:pullback-entrance-law}, there is a unique pullback boundary state
	\[
		(\rho_t)_{t\in\mbZ}\subset\mcS_d
	\]
	satisfying
	\[
		\rho_t
		=
		\lim_{s\to-\infty}\Phi_{t:s}(\tau)
	\]
	for every \(t\in\mbZ\) and every reference state \(\tau\in\mcS_d\).

	Let \((\eta_t)_{t\in\mbZ}\) be any family of states satisfying statement \(2\), and write
	\[
		S_t(X):=\tr{X}\eta_t.
	\]
	Fix \(t\in\mbZ\) and \(\tau\in\mcS_d\).
	Since \(\tr{\tau}=1\), we have
	\[
		S_t(\tau)=\eta_t.
	\]
	Therefore
	\[
		\norm{\eta_t-\rho_t}_1
		=
		\norm{S_t(\tau)-\rho_t}_1.
	\]
	Using the triangle inequality,
	\[
		\norm{S_t(\tau)-\rho_t}_1
		\le
		\norm{S_t(\tau)-\Phi_{t:s}(\tau)}_1
		+
		\norm{\Phi_{t:s}(\tau)-\rho_t}_1.
	\]
	The first term satisfies
	\[
		\norm{S_t(\tau)-\Phi_{t:s}(\tau)}_1
		\le
		\norm{S_t-\Phi_{t:s}}_{1\to1}\norm{\tau}_1
		=
		\norm{S_t-\Phi_{t:s}}_{1\to1},
	\]
	and tends to \(0\) as \(s\to-\infty\) by statement \(2\).
	The second term tends to \(0\) by the reference-state formula in \Cref{prop:pullback-entrance-law}.
	Therefore
	\[
		\eta_t=\rho_t.
	\]
	Since \(t\) was arbitrary, the replacement centers are unique and coincide with the canonical pullback boundary state.

	It remains to prove the convergence of the reference-state replacement channels.
	For every \(t\in\mbZ\) and every \(\tau\in\mcS_d\),
	\[
		R_{t:s}^{(\tau)}(X)-R_t(X)
		=
		\tr{X}\bigl(\Phi_{t:s}(\tau)-\rho_t\bigr).
	\]
	Hence
	\[
		\norm{R_{t:s}^{(\tau)}-R_t}_{1\to1}
		=
		\norm{\Phi_{t:s}(\tau)-\rho_t}_1.
	\]
	By \Cref{prop:pullback-entrance-law},
	\[
		\norm{\Phi_{t:s}(\tau)-\rho_t}_1
		\le
		2\kappa_{t:s}.
	\]
	Since \(\kappa_{t:s}\to0\) as \(s\to-\infty\), we obtain
	\[
		\lim_{s\to-\infty}
		\norm{R_{t:s}^{(\tau)}-R_t}_{1\to1}
		=
		0.
	\]
	The identity
	\[
		\rho_{t+1}=\Phi_t(\rho_t)
	\]
	and the reference-state formula was already obtained from \Cref{prop:pullback-entrance-law}.
	This completes the proof.
\end{proof}

%%%%%%%%%%%%%%%%%%%%%%%%%%%%%%%%%%%%%%%%%%%%%%%%%%%%%%%%%%%%%%%%%%%%%%%%%%%
%%%%%%%%%%%%%%%%%%%%%%%%%%%%%%%%%%%%%%%%%%%%%%%%%%%%%%%%%%%%%%%%%%%%%%%%%%%

\subsection{Deterministic Rates}
\label{sec:det-rates}

%%%%%%%%%%%%%%%%%%%%%%%%%%%%%%%%%%%%%%%%%%%%%%%%%%%%%%%%%%%%%%%%%%%%%%%%%%%
%%%%%%%%%%%%%%%%%%%%%%%%%%%%%%%%%%%%%%%%%%%%%%%%%%%%%%%%%%%%%%%%%%%%%%%%%%%

The preceding subsections identify \(\kappa_{t:s}\) as the exact product-level coefficient for fixed-start and pullback asymptotic replacement.
We now record the deterministic rate mechanisms that will be used later.
The estimates are stated for a single product \(\Phi_{t:s}\), rather than uniformly over all windows.
This is the form needed for the deterministic MPS bounds, because a local observable with support up to time \(m\) only sees the auxiliary tail product after that support.

Recall that \(\mcI\) denotes the set of admissible starting times.
Thus \(\mcI=\mbN_0\) in the one-sided setting and \(\mcI=\mbZ\) in the two-sided setting.
For \(s<t\), set
\[
	d_{t:s}
	:=
	\sup_{\rho,\sigma\in\mcS_d}
	\norm{\Phi_{t:s}(\rho)-\Phi_{t:s}(\sigma)}_1.
\]
By \Cref{prop:dobr},
\[
	d_{t:s}=2\kappa_{t:s}.
\]
Thus, every upper bound on \(\kappa_{t:s}\) immediately gives a trace-diameter bound.
By \Cref{prop:det-radius-replacement}, every such upper bound also gives an induced \(1\to1\) replacement bound:
\[
	\norm{\Phi_{t:s}-R_{t:s}^{(\tau)}}_{1\to1}
	\le
	4\kappa_{t:s},
	\qquad
	R_{t:s}^{(\tau)}(X)
	:=
	\tr{X}\Phi_{t:s}(\tau).
\]

\begin{prop}
\label{prop:det-contraction-clocks}
	Let \(s<t\).
	Suppose that \(a_j\in[0,1]\) and
	\[
		\kappa_{\rm tr}(\Phi_j)
		\le
		1-a_j,
		\qquad
		j=s,\ldots,t-1.
	\]
	Then
	\[
		\kappa_{t:s}
		\le
		\prod_{j=s}^{t-1}(1-a_j)
		\le
		\exp\!\left(-\sum_{j=s}^{t-1}a_j\right).
	\]
	For every \(r\in(0,1)\), define
	\[
		G_r(s,t)
		:=
		\#\{j\in\{s,\ldots,t-1\}:a_j\ge r\}.
	\]
	Then
	\[
		\kappa_{t:s}
		\le
		(1-r)^{G_r(s,t)}.
	\]
	Consequently, for every reference state \(\tau\in\mcS_d\),
	\[
		d_{t:s}
		\le
		2\prod_{j=s}^{t-1}(1-a_j)
		\le
		2\exp\!\left(-\sum_{j=s}^{t-1}a_j\right),
	\]
	and
	\[
		\norm{\Phi_{t:s}-R_{t:s}^{(\tau)}}_{1\to1}
		\le
		4\prod_{j=s}^{t-1}(1-a_j)
		\le
		4\exp\!\left(-\sum_{j=s}^{t-1}a_j\right).
	\]
	In particular, for every \(r\in(0,1)\),
	\[
		d_{t:s}
		\le
		2(1-r)^{G_r(s,t)},
		\qquad
		\norm{\Phi_{t:s}-R_{t:s}^{(\tau)}}_{1\to1}
		\le
		4(1-r)^{G_r(s,t)}.
	\]
\end{prop}

\begin{proof}
	By submultiplicativity,
	\[
		\kappa_{t:s}
		\le
		\prod_{j=s}^{t-1}\kappa_{\rm tr}(\Phi_j)
		\le
		\prod_{j=s}^{t-1}(1-a_j).
	\]
	The exponential bound follows from the elementary inequality
	\[
		1-a\le e^{-a},
		\qquad
		a\in[0,1].
	\]
	Fix \(r\in(0,1)\).
	For every \(j\in\{s,\ldots,t-1\}\) with \(a_j\ge r\), the corresponding factor satisfies
	\[
		1-a_j\le 1-r.
	\]
	All remaining factors are at most \(1\).
	Therefore
	\[
		\prod_{j=s}^{t-1}(1-a_j)
		\le
		(1-r)^{G_r(s,t)},
	\]
	and hence
	\[
		\kappa_{t:s}
		\le
		(1-r)^{G_r(s,t)}.
	\]
	The estimates for \(d_{t:s}\) and for the replacement channel follow from
	\[
		d_{t:s}=2\kappa_{t:s}
	\]
	and from \Cref{prop:det-radius-replacement}.
\end{proof}

The preceding estimate is a one-step criterion.
The genuinely product-level mechanism is obtained by counting contractive blocks.
Suppose that \(s<t\) and that
\[
	\left[u_1,v_1\right),\ldots,\left[u_m,v_m\right)
	\subseteq
	\left[s,t\right)
\]
are ordered disjoint subintervals.
This means that
\[
	s\le u_1<v_1\le u_2<v_2\le\cdots\le u_m<v_m\le t.
\]
If \(\kappa_{v_r:u_r}\le q_r\le1\) for \(r=1,\ldots,m\), then repeated use of submultiplicativity gives
\begin{equation}
\label{eq:disjoint-block-factorization}
	\kappa_{t:s}
	\le
	\prod_{r=1}^m q_r,
\end{equation}
with the empty product interpreted as \(1\).
Indeed, the product \(\Phi_{t:s}\) factors into the selected blocks and the intervening gap products.
Every gap product is CPTP and hence has trace-Dobrushin coefficient at most \(1\).

Fix \(q\in\left(0,1\right)\).
Let \(\mcB_q\left(s,t\right)\) be the maximum number of pairwise disjoint nonempty subintervals
\[
	\left[u_r,v_r\right)
	\subseteq
	\left[s,t\right)
\]
such that
\[
	\kappa_{v_r:u_r}\le q.
\]
By \Cref{eq:disjoint-block-factorization},
\[
	\kappa_{t:s}
	\le
	q^{\mcB_q\left(s,t\right)},
	\qquad
	d_{t:s}
	\le
	2q^{\mcB_q\left(s,t\right)}.
\]
For every \(\tau\in\mcS_d\), this also gives
\[
	\norm{\Phi_{t:s}-R_{t:s}^{\left(\tau\right)}}_{1\to1}
	\le
	4q^{\mcB_q\left(s,t\right)}.
\]
Thus, memory loss can be certified by contractive products even when no useful one-step contraction estimate is available.

\begin{definition}
\label{def:uniform-good-block}
	Fix integers \(\ell\ge1\), \(M\ge\ell\), and \(q\in\left(0,1\right)\).
	The sequence satisfies the uniform \(\left(\ell,M,q\right)\)-good-block condition if every admissible interval
	\[
		\left[r,r+M\right),
		\qquad
		r,r+M\in\mcI,
	\]
	contains a length-\(\ell\) subblock
	\[
		\left[u,u+\ell\right)
		\subseteq
		\left[r,r+M\right)
	\]
	such that
	\[
		\kappa_{u+\ell:u}\le q.
	\]
\end{definition}

\begin{prop}
\label{prop:uniform-good-block}
	Assume the uniform \(\left(\ell,M,q\right)\)-good-block condition.
	Then for every admissible product \(\Phi_{t:s}\),
	\[
		\kappa_{t:s}
		\le
		q^{\left\lfloor\frac{t-s}{M}\right\rfloor}.
	\]
	Consequently, for every reference state \(\tau\in\mcS_d\),
	\[
		d_{t:s}
		\le
		2q^{\left\lfloor\frac{t-s}{M}\right\rfloor},
		\qquad
		\norm{\Phi_{t:s}-R_{t:s}^{\left(\tau\right)}}_{1\to1}
		\le
		4q^{\left\lfloor\frac{t-s}{M}\right\rfloor}.
	\]
	In particular,
	\[
		\kappa_{t:s}
		\le
		q^{-1}
		\exp\!\left(-\frac{\left|\log q\right|}{M}\left(t-s\right)\right).
	\]
\end{prop}

\begin{proof}
	Set \(L:=t-s\) and \(m:=\left\lfloor L/M\right\rfloor\).
	We construct \(m\) disjoint \(q\)-contractive blocks inside \(\left[s,t\right)\).

	Set \(r_0=s\).
	Suppose that \(0\le j\le m-1\) and that \(r_j\) has already been chosen with
	\[
		r_j\le s+jM.
	\]
	Then
	\[
		r_j+M
		\le
		s+\left(j+1\right)M
		\le
		s+mM
		\le
		t.
	\]
	Hence \(\left[r_j,r_j+M\right)\subseteq\left[s,t\right)\).
	By the uniform good-block condition, there exists a length-\(\ell\) subblock
	\[
		\left[u_j,u_j+\ell\right)
		\subseteq
		\left[r_j,r_j+M\right)
	\]
	such that
	\[
		\kappa_{u_j+\ell:u_j}\le q.
	\]
	Set \(r_{j+1}:=u_j+\ell\).
	Then
	\[
		r_{j+1}
		\le
		r_j+M
		\le
		s+\left(j+1\right)M.
	\]
	The next block begins after the previous one ends.
	Thus, the selected blocks are pairwise disjoint.

	This constructs \(\left\lfloor\left(t-s\right)/M\right\rfloor\) disjoint \(q\)-contractive blocks in \(\left[s,t\right)\).
	Hence
	\[
		\mcB_q\left(s,t\right)
		\ge
		\left\lfloor\frac{t-s}{M}\right\rfloor.
	\]
	The block-count estimate gives
	\[
		\kappa_{t:s}
		\le
		q^{\left\lfloor\frac{t-s}{M}\right\rfloor}.
	\]
	The bounds on \(d_{t:s}\) and \(R_{t:s}^{\left(\tau\right)}\) follow from \(d_{t:s}=2\kappa_{t:s}\) and \Cref{prop:det-radius-replacement}.
	Finally,
	\[
		q^{\left\lfloor\frac{t-s}{M}\right\rfloor}
		\le
		q^{-1}\exp\!\left(-\frac{\left|\log q\right|}{M}\left(t-s\right)\right),
	\]
	because
	\[
		\left\lfloor\frac{t-s}{M}\right\rfloor
		\ge
		\frac{t-s}{M}-1.
	\]
\end{proof}

A useful special case is uniform contraction at one fixed block length.
If there are \(L_0\ge1\) and \(q\in\left(0,1\right)\) such that
\[
	\kappa_{u+L_0:u}\le q
	\qquad
	\text{for every admissible }u,
\]
then \Cref{prop:uniform-good-block} applies with \(\ell=M=L_0\).
Consequently,
\[
	\kappa_{t:s}
	\le
	q^{\left\lfloor\frac{t-s}{L_0}\right\rfloor},
	\qquad
	d_{t:s}
	\le
	2q^{\left\lfloor\frac{t-s}{L_0}\right\rfloor},
\]
and
\[
	\norm{\Phi_{t:s}-R_{t:s}^{\left(\tau\right)}}_{1\to1}
	\le
	4q^{\left\lfloor\frac{t-s}{L_0}\right\rfloor}.
\]
This is the deterministic analogue of a uniform Doeblin condition.
The bounded-gap good-block condition is weaker.
It does not require every block of a fixed length to contract, only that contractive blocks occur with uniformly bounded gaps.

The block mechanism detects contraction that is invisible to one-step criteria.

\begin{example}
\label{ex:alternating-dephasing}
	Let \(D_Z\) be complete dephasing in the computational basis of a qubit:
	\[
		D_Z\left(\rho\right)
		=
		P_0\rho P_0+P_1\rho P_1,
		\qquad
		P_i=\left|i\right\rangle\left\langle i\right|.
	\]
	Let \(D_X\) be complete dephasing in the Hadamard basis:
	\[
		D_X\left(\rho\right)
		=
		P_+\rho P_+ + P_-\rho P_-,
		\qquad
		P_\pm=\left|\pm\right\rangle\left\langle\pm\right|.
	\]
	Then
	\[
		\kappa_{\rm tr}\left(D_Z\right)=1,
		\qquad
		\kappa_{\rm tr}\left(D_X\right)=1.
	\]
	Indeed, \(D_Z\) preserves the trace distance between \(P_0\) and \(P_1\).
	The channel \(D_X\) preserves the trace distance between \(P_+\) and \(P_-\).
	The one-step Markov--Dobrushin lower bound is also zero for both maps, since in each case there are pure inputs whose outputs are orthogonal.
	However,
	\[
		D_X\circ D_Z
		=
		R_{I/2},
		\qquad
		R_{I/2}\left(X\right)=\tr{X}\frac{I}{2}.
	\]
	Hence
	\[
		\kappa_{\rm tr}\left(D_X\circ D_Z\right)=0.
	\]
	Thus a two-step block can have perfect trace-memory loss even though every one-step coefficient is trivial.
\end{example}

\begin{remark}[Certificates for good blocks]
	Every occurrence of a condition of the form \(\kappa_{v:u}\le q\) may be replaced by any rigorous certificate implying that bound.
	Such a certificate may come from an exact computation of \(\kappa_{\rm tr}\), a Markov--Dobrushin lower bound for the whole block, a CP-order Doeblin bound for the whole block, a projective-diameter estimate for the whole block, or any other valid contraction estimate.
	The proofs use only submultiplicativity and the stated upper bounds on selected block coefficients.
\end{remark}

The preceding criteria separate two independent mechanisms.
One-step clocks quantify the accumulation of small contractions at individual times.
Good-block clocks quantify contraction that appears only after composing several noncontractive channels.
Both mechanisms are used below in the deterministic MPS convergence estimates.

%%%%%%%%%%%%%%%%%%%%%%%%%%%%%%%%%%%%%%%%%%%%%%%%%%%%%%%%%%%%%%%%%%%%%%%%%%%
%%%%%%%%%%%%%%%%%%%%%%%%%%%%%%%%%%%%%%%%%%%%%%%%%%%%%%%%%%%%%%%%%%%%%%%%%%%

\section{Random CPTP cocycles}
\label{sec:random}

%%%%%%%%%%%%%%%%%%%%%%%%%%%%%%%%%%%%%%%%%%%%%%%%%%%%%%%%%%%%%%%%%%%%%%%%%%%
%%%%%%%%%%%%%%%%%%%%%%%%%%%%%%%%%%%%%%%%%%%%%%%%%%%%%%%%%%%%%%%%%%%%%%%%%%%

We now pass from deterministic inhomogeneous products to stationary random products.
Let \(\left(\Omega,\mcF,\pr,\theta\right)\) be an invertible probability-preserving system.
We use the standard terminology for measure-preserving systems and ergodicity from \cite{walters2000introduction}.
Let \(\omega\mapsto \Phi_\omega\) be a measurable map from \(\Omega\) into the finite-dimensional space of linear maps on \(\matrices\), and assume that each \(\Phi_\omega\) is CPTP.
For \(n\ge1\), define the forward random product
\[
   \Phi_\omega^{(n)}
   :=
   \Phi_{\theta^{n-1}\omega}
   \circ
   \Phi_{\theta^{n-2}\omega}
   \circ\cdots\circ
   \Phi_\omega,
   \qquad
   \Phi_\omega^{(0)}:=\mathrm{id}.
\]
The associated trace-Dobrushin product coefficient is
\[
   \kappa_n(\omega)
   :=
   \kappa_{\rm tr}\left(\Phi_\omega^{(n)}\right).
\]
We also set \(\kappa_0(\omega):=1\).
Since \(\matrices\) is finite-dimensional, the space of linear maps on \(\matrices\) is finite-dimensional.
For each fixed \(j\), the map \(\omega\mapsto \Phi_{\theta^j\omega}\) is measurable.
The composition map on linear operators is continuous in finite dimensions.
Hence \(\omega\mapsto \Phi_\omega^{(n)}\) is measurable for every fixed \(n\).
Every CPTP map preserves \(\tracezero\).
By definition, \(\kappa_{\rm tr}(\Phi)\) is the \(1\to1\) operator norm of the restriction of \(\Phi\) to \(\tracezero\), where \(\tracezero\) is equipped with the trace norm.
The operator norm depends continuously on the linear map in finite dimensions.
Therefore \(\omega\mapsto \kappa_{\rm tr}(\Phi_\omega^{(n)})\) is measurable.

The base transformation supplies the stationary environment for the sequence of channels.
The contraction of the channel product is a separate property.
It is measured by the decay of \(\kappa_n(\omega)\).

\begin{definition}[Base ergodicity]
    The base transformation \(\theta\) is ergodic if every \(E\in\mcF\) satisfying \(\theta^{-1}E=E\) has \(\pr(E)\in\{0,1\}\).
    Equivalently, \(\theta\) is ergodic if every \(E\in\mcF\) satisfying \(\pr(E\triangle\theta^{-1}E)=0\) has \(\pr(E)\in\{0,1\}\).
    Here \(\triangle\) denotes the symmetric difference of sets.
\end{definition}

We also refer the reader to \cite{walters2000introduction} for other equivalent formulations of ergodicity.  

\begin{remark}
    Ergodicity of \(\theta\) is a property of the randomness.
    It does not by itself imply contraction of the products \(\Phi_\omega^{(n)}\).
    The channel-level memory loss studied below is the decay of \(\kappa_n(\omega)\).
    Any stronger temporal decorrelation assumption needed for annealed or high-probability estimates will be stated separately.
    For example, let \(\theta\) be an i.i.d. shift and let \(\Phi_\omega(X)=U_\omega XU_\omega^*\) be a random unitary conjugation.
    Then
    \[
       \Phi_\omega^{(n)}(X)
       =
       V_{\omega,n}X V_{\omega,n}^*,
       \qquad
       V_{\omega,n}
       :=
       U_{\theta^{n-1}\omega}\cdots U_\omega.
    \]
    Thus, every product \(\Phi_\omega^{(n)}\) is again a unitary conjugation.
    Since the trace norm is unitarily invariant, unitary conjugations act isometrically on \(\tracezero\).
    Hence, when \(d\ge2\), one has \(\kappa_n(\omega)=1\) for all \(n\).
    Thus, even a strongly mixing or i.i.d. base can generate a channel cocycle with no trace-Dobrushin contraction.
\end{remark}

We now introduce the pathwise notions of memory loss used in the random setting.
The first one is the coefficient-level formulation.
It says that, in almost every environment, the product eventually forgets all trace-zero perturbations of the initial state.

\begin{definition}[Quenched trace-memory loss]
\label{def:quenched-trace-memory-loss}
    The random channel cocycle has quenched trace-memory loss if \(\kappa_n(\omega)\to0\) for \(\pr\)-a.e. \(\omega\).
\end{definition}

By \Cref{prop:dobr}, this is equivalent to uniform forgetting of pairs of initial states along almost every realization.
Indeed,
\[
    \sup_{\rho,\sigma\in\mcS_d}
    \norm{
        \Phi_\omega^{(n)}(\rho)-\Phi_\omega^{(n)}(\sigma)
    }_1
    =
    2\kappa_n(\omega).
\]
Consequently, quenched trace-memory loss holds if and only if
\[
    \sup_{\rho,\sigma\in\mcS_d}
    \norm{
        \Phi_\omega^{(n)}(\rho)-\Phi_\omega^{(n)}(\sigma)
    }_1
    \longrightarrow0
\]
for \(\pr\)-a.e. \(\omega\).
Thus \(\kappa_n(\omega)\) is the random analogue of the deterministic memory coefficient \(\kappa_{t:s}\).

The same loss of memory can also be expressed at the channel level.
For a fixed realization \(\omega\), the product \(\Phi_\omega^{(n)}\) is an ordinary deterministic inhomogeneous product, so the natural random analogue of strong asymptotic replacement is obtained by requiring the deterministic replacement property for \(\pr\)-a.e. fiber.

\begin{definition}[Quenched strong asymptotic replacement]
\label{def:quenched-strong-asymptotic-replacement}
    The random channel cocycle is quenched strongly asymptotically replacing if there is a full-measure set \(\Omega_0\subseteq\Omega\) such that, for every \(\omega\in\Omega_0\), there are states \(\eta_{\omega,n}\in\mcS_d\) with the following property.
    If
    \[
       R_{\omega,n}(X)
       :=
       \tr{X}\eta_{\omega,n},
    \]
    then
    \[
       \lim_{n\to\infty}
       \norm{\Phi_\omega^{(n)}-R_{\omega,n}}_{1\to1}
       =
       0.
    \]
\end{definition}

%%%%%%%%%%%%%%%%%%%%%%%%%%%%%%%%%%%%%%%%%%%%%%%%%%%%%%%%%%%%%%%%%%%%%%%%%%%%
%%%%%%%%%%%%%%%%%%%%%%%%%%%%%%%%%%%%%%%%%%%%%%%%%%%%%%%%%%%%%%%%%%%%%%%%%%%%

\subsection{Trace-Dobrushin Lyapunov exponent}
\label{sec:random-lyapunov}

%%%%%%%%%%%%%%%%%%%%%%%%%%%%%%%%%%%%%%%%%%%%%%%%%%%%%%%%%%%%%%%%%%%%%%%%%%%%
%%%%%%%%%%%%%%%%%%%%%%%%%%%%%%%%%%%%%%%%%%%%%%%%%%%%%%%%%%%%%%%%%%%%%%%%%%%%

We now pass from the random product coefficients to their exponential growth rate.
Throughout this subsection we use the convention \(\log 0=-\infty\).
For \(n\ge1\), set
\[
   F_n(\omega)
   :=
   \log\kappa_n(\omega).
\]
Since \(0\le\kappa_n(\omega)\le1\), each \(F_n\) takes values in \(\left[-\infty,0\right]\).

\begin{prop}
\label{prop:random-submultiplicativity}
    For all \(m,n\ge0\),
    \[
       \kappa_{n+m}(\omega)
       \le
       \kappa_n\left(\theta^m\omega\right)\kappa_m(\omega).
    \]
    Consequently, for \(m,n\ge1\),
    \[
       F_{n+m}(\omega)
       \le
       F_m(\omega)+F_n\left(\theta^m\omega\right).
    \]
\end{prop}

\begin{proof}
    If \(m=0\) or \(n=0\), the first estimate follows from \(\kappa_0=1\).
    Assume \(m,n\ge1\).
    The cocycle identity gives
    \[
       \Phi_\omega^{(n+m)}
       =
       \Phi_{\theta^m\omega}^{(n)}
       \circ
       \Phi_\omega^{(m)}.
    \]
    By \Cref{prop:submult},
    \[
       \kappa_{\rm tr}\left(\Phi_\omega^{(n+m)}\right)
       \le
       \kappa_{\rm tr}\left(\Phi_{\theta^m\omega}^{(n)}\right)
       \kappa_{\rm tr}\left(\Phi_\omega^{(m)}\right).
    \]
    This is the desired inequality.
    Taking logarithms gives the second estimate when both factors on the right are nonzero.
    If one factor on the right is zero, then the product estimate forces \(\kappa_{n+m}(\omega)=0\).
    Thus, the logarithmic inequality remains valid with the convention \(\log0=-\infty\).
\end{proof}

We can now apply Kingman's subadditive ergodic theorem \cite{Kingman1973}.

\begin{theorem}[Trace-Dobrushin Lyapunov exponent]
\label{thm:random-lyapunov}
    There exists a \(\theta\)-invariant random variable
    \[
       \lambda_{\rm tr}:\Omega\to\left[-\infty,0\right]
    \]
    such that
    \[
       \lambda_{\rm tr}(\omega)
       =
       \lim_{n\to\infty}
       {\dfrac 1n}\log\kappa_n(\omega)
    \]
    for \(\pr\)-a.e. \(\omega\).
    If \(\theta\) is ergodic, then \(\lambda_{\rm tr}\) is a.s. constant.
    
    In the ergodic case, this constant satisfies
    \[
       \lambda_{\rm tr}
       =
       \inf_{n\ge1}
       {\dfrac 1n}\,
       \mathbb E\left[\log\kappa_n\right].
    \]
\end{theorem}

\begin{proof}
    We already saw that each \(F_n\) is measurable.
    Moreover, \(F_n^+=0\in L^1(\pr)\) for every \(n\ge1\).
    By \Cref{prop:random-submultiplicativity}, the sequence \((F_n)_{n\ge1}\) is subadditive over \(\theta\):
    \[
       F_{n+m}(\omega)
       \le
       F_m(\omega)+F_n(\theta^m\omega).
    \]
    Kingman's subadditive ergodic theorem, in its extended-valued form, gives the a.s. existence of
    \[
       \lim_{n\to\infty}{\dfrac{1}{n}}F_n(\omega).
    \]
    We denote this limit by \(\lambda_{\rm tr}(\omega)\).
    The same theorem gives that \(\lambda_{\rm tr}\) is \(\theta\)-invariant.

    If \(\theta\) is ergodic, then every \(\theta\)-invariant random variable is a.s. constant.
    Thus \(\lambda_{\rm tr}\) is a.s. constant.
    For nonpositive extended random variables, we use the convention
    \[
       \mathbb E[F]
       :=
       -\mathbb E[-F]
       \in[-\infty,0].
    \]
    The ergodic form of Kingman's theorem gives
    \[
       \lambda_{\rm tr}
       =
       \inf_{n\ge1}
       {\dfrac{1}{n}}\,
       \mathbb E[F_n]
       =
       \inf_{n\ge1}
       {\dfrac{1}{n}}\,
       \mathbb E[\log\kappa_n].
    \]
\end{proof}

For integers \(s<t\), set
\[
   \Phi_{\omega;t:s}
   :=
   \Phi_{\theta^{t-1}\omega}
   \circ
   \Phi_{\theta^{t-2}\omega}
   \circ\cdots\circ
   \Phi_{\theta^s\omega},
   \qquad
   \kappa_{\omega;t:s}
   :=
   \kappa_{\rm tr}\left(\Phi_{\omega;t:s}\right).
\]
Thus \(\Phi_\omega^{(n)}=\Phi_{\omega;n:0}\) and \(\kappa_n(\omega)=\kappa_{\omega;n:0}\).
The chronological product over the past interval \(\left[-n,0\right]\) is
\[
   \Phi_{\omega;0:-n}
   =
   \Phi_{\theta^{-1}\omega}
   \circ
   \Phi_{\theta^{-2}\omega}
   \circ\cdots\circ
   \Phi_{\theta^{-n}\omega},
\]
and its coefficient is
\[
   \kappa_{\omega;0:-n}
   =
   \kappa_n\left(\theta^{-n}\omega\right).
\]
With these notations, we have the following result. 

\begin{prop}
\label{prop:pullback-same-exponent}
    For \(\pr\)-a.e. \(\omega\),
    \[
       \lim_{n\to\infty}
       {\dfrac{1}{n}}
       \log\kappa_{\omega;0:-n}
       =
       \lambda_{\rm tr}(\omega).
    \]
    Equivalently,
    \[
       \lim_{n\to\infty}
       {\dfrac{1}{n}}
       \log\kappa_n\left(\theta^{-n}\omega\right)
       =
       \lambda_{\rm tr}(\omega)
       \qquad
       \text{for }\pr\text{-a.e. }\omega.
    \]
\end{prop}

\begin{proof}
    Set \(F_n(\omega):=\log\kappa_n(\omega)\), with the convention \(\log0=-\infty\).
    Set
    \[
       G_n(\omega)
       :=
       \log\kappa_{\omega;0:-n}
       =
       F_n\left(\theta^{-n}\omega\right).
    \]
    
    We first check subadditivity over \(S:=\theta^{-1}\).
    By \Cref{prop:random-submultiplicativity}, applied at \(\theta^{-(n+m)}\omega\), one has
    \[
       \kappa_{n+m}\left(\theta^{-(n+m)}\omega\right)
       \le
       \kappa_n\left(\theta^{-n}\omega\right)
       \kappa_m\left(\theta^{-(n+m)}\omega\right).
    \]
    Equivalently,
    \[
       G_{n+m}(\omega)
       \le
       G_n(\omega)+G_m\left(\theta^{-n}\omega\right).
    \]
    Since \(S^n\omega=\theta^{-n}\omega\), the sequence \(\left(G_n\right)_{n\ge1}\) is subadditive over \(S\).
    
    Each \(G_n\) is measurable and takes values in \(\left[-\infty,0\right]\).
    Hence \(G_n^+=0\in L^1\left(\pr\right)\).
    Kingman's theorem applied to the probability-preserving system \(\left(\Omega,\mcF,\pr,\theta^{-1}\right)\) gives the a.s. existence of
    \[
       \lambda_{\rm tr}^{\rm past}(\omega)
       :=
       \lim_{n\to\infty}{\dfrac{1}{n}}G_n(\omega).
    \]
    By the conditional form of Kingman's theorem for extended nonpositive subadditive cocycles,
    \[
       \lambda_{\rm tr}^{\rm past}
       =
       \inf_{n\ge1}
       {\dfrac{1}{n}}
       \mathbb E\left[G_n\mid\mcI_{\theta^{-1}}\right]
       \qquad
       \text{a.s.}
    \]
    The invariant \(\sigma\)-algebras of \(\theta\) and \(\theta^{-1}\) coincide, so \(\mcI_{\theta^{-1}}=\mcI_\theta\).
    
    It remains to identify the conditional expectations.
    Since \(G_n=F_n\circ\theta^{-n}\), it is enough to prove
    \[
       \mathbb E\left[F_n\circ\theta^{-n}\mid\mcI_\theta\right]
       =
       \mathbb E\left[F_n\mid\mcI_\theta\right].
    \]
    We use the extended nonpositive convention.
    Equivalently, we prove the corresponding identity for the nonnegative random variable \(H_n:=-F_n\).
    
    Let \(A\in\mcI_\theta\).
    Since \(\theta\) is invertible and probability preserving,
    \[
       \int_A H_n\circ\theta^{-n}\,d\pr
       =
       \int_{\theta^{-n}A} H_n\,d\pr.
    \]
    Since \(A\in\mcI_\theta\), we have \(\theta^{-n}A=A\) modulo null sets.
    Therefore
    \[
       \int_A H_n\circ\theta^{-n}\,d\pr
       =
       \int_A H_n\,d\pr.
    \]
    Thus
    \[
       \mathbb E\left[H_n\circ\theta^{-n}\mid\mcI_\theta\right]
       =
       \mathbb E\left[H_n\mid\mcI_\theta\right].
    \]
    Multiplying by \(-1\) gives
    \[
       \mathbb E\left[G_n\mid\mcI_\theta\right]
       =
       \mathbb E\left[F_n\mid\mcI_\theta\right].
    \]
    Hence
    \[
       \lambda_{\rm tr}^{\rm past}
       =
       \inf_{n\ge1}
       {\dfrac{1}{n}}
       \mathbb E\left[G_n\mid\mcI_\theta\right]
       =
       \inf_{n\ge1}
       {\dfrac{1}{n}}
       \mathbb E\left[F_n\mid\mcI_\theta\right]
       =
       \lambda_{\rm tr}.
    \]
    This proves the claim.
\end{proof}

We conclude this part of the random theory by identifying when the trace-Dobrushin Lyapunov exponent is strictly negative.
This is the key condition behind the exponential form of quenched memory loss.
The first observation is that, in the stationary setting, almost-sure decay of the product coefficients already rules out the borderline case \(\lambda_{\rm tr}=0\).

\begin{prop}
\label{prop:memory-loss-negative-exponent-equivalence}
    The following are equivalent:
    \begin{enumerate}
        \item \(\lambda_{\rm tr}(\omega)<0\), for \(\pr\)-a.e. \(\omega\).
       \item \( \kappa_{\omega;n:0}\longrightarrow0\), for \(\pr\)-a.e. \(\omega\).
       \item \( \kappa_{\omega;0:-n}\longrightarrow0\), for \(\pr\)-a.e. \(\omega\).
    \end{enumerate}
\end{prop}

\begin{proof}
    Assume first that \(\lambda_{\rm tr}(\omega)<0\) for \(\pr\)-a.e. \(\omega\).
    By \Cref{thm:random-lyapunov},
    \[
       \lim_{n\to\infty}
       {\dfrac{1}{n}}\log\kappa_{\omega;n:0}
       =
       \lambda_{\rm tr}(\omega)
    \]
    for \(\pr\)-a.e. \(\omega\).
    Hence \(\kappa_{\omega;n:0}\to0\) for \(\pr\)-a.e. \(\omega\).
    
    By \Cref{prop:pullback-same-exponent},
    \[
       \lim_{n\to\infty}
       {\dfrac{1}{n}}\log\kappa_{\omega;0:-n}
       =
       \lambda_{\rm tr}(\omega)
    \]
    for \(\pr\)-a.e. \(\omega\).
    Thus \(\kappa_{\omega;0:-n}\to0\) for \(\pr\)-a.e. \(\omega\) as well.
    
    It remains to prove the converse.
    Assume that
    \[
       \kappa_{\omega;n:0}\longrightarrow0
       \qquad
       \text{for }\pr\text{-a.e. }\omega.
    \]
    Set \(F_n(\omega):=\log\kappa_{\omega;n:0}\), with the convention \(\log0=-\infty\).
    Let
    \[
       A
       :=
       \left\{
          \omega:\lambda_{\rm tr}(\omega)=0
       \right\}.
    \]
    Since \(\lambda_{\rm tr}\) is \(\theta\)-invariant, \(A\in\mcI_\theta\).
    By Kingman's conditional variational formula,
    \[
       \lambda_{\rm tr}
       =
       \inf_{n\ge1}
       {\dfrac{1}{n}}
       \mathbb E\left[F_n\mid\mcI_\theta\right]
       \qquad
       \text{a.s.}
    \]
    Since \(F_n\le0\), we have \(\mathbb E\left[F_n\mid\mcI_\theta\right]\le0\).
    On \(A\), for every \(n\ge1\),
    \[
       0
       =
       \lambda_{\rm tr}
       \le
       {\dfrac{1}{n}}
       \mathbb E\left[F_n\mid\mcI_\theta\right]
       \le
       0.
    \]
    Therefore
    \[
       \mathbb E\left[F_n\mid\mcI_\theta\right]
       =
       0
       \qquad
       \text{a.s. on }A.
    \]
    Equivalently, for \(H_n:=-F_n\ge0\),
    \[
       \mathbb E\left[H_n\mid\mcI_\theta\right]
       =
       0
       \qquad
       \text{a.s. on }A.
    \]
    Because \(A\in\mcI_\theta\),
    \[
       \mathbb E\left[\mathbf 1_A H_n\right]
       =
       \mathbb E\left[
          \mathbf 1_A
          \mathbb E\left[H_n\mid\mcI_\theta\right]
       \right]
       =
       0.
    \]
    Since \(\mathbf 1_AH_n\ge0\), this gives \(H_n=0\) for \(\pr\)-a.e. \(\omega\in A\).
    Thus \(F_n=0\), equivalently
    \[
       \kappa_{\omega;n:0}=1,
    \]
    for \(\pr\)-a.e. \(\omega\in A\), for each fixed \(n\ge1\).
    Taking a countable intersection over \(n\), we get
    \[
       \kappa_{\omega;n:0}=1
       \qquad
       \text{for every }n\ge1
    \]
    for \(\pr\)-a.e. \(\omega\in A\).
    
    This contradicts \(\kappa_{\omega;n:0}\to0\) on any subset of \(A\) of positive measure.
    Hence \(\pr\left(A\right)=0\).
    Since \(\lambda_{\rm tr}\in\left[-\infty,0\right]\), this proves \(\lambda_{\rm tr}<0\) a.s.
    
    Finally, assume instead that
    \[
       \kappa_{\omega;0:-n}\longrightarrow0
       \qquad
       \text{for }\pr\text{-a.e. }\omega.
    \]
    Set \(G_n(\omega):=\log\kappa_{\omega;0:-n}\).
    As in the proof of \Cref{prop:pullback-same-exponent}, the sequence \(\left(G_n\right)_{n\ge1}\) is subadditive over \(\theta^{-1}\).
    Applying the argument above to the probability-preserving system \(\left(\Omega,\mcF,\pr,\theta^{-1}\right)\) shows that the pullback Kingman exponent is negative a.s.
    By \Cref{prop:pullback-same-exponent}, this exponent agrees with \(\lambda_{\rm tr}\).
    Therefore \(\lambda_{\rm tr}<0\) a.s.
\end{proof}

%%%%%%%%%%%%%%%%%%%%%%%%%%%%%%%%%%%%%%%%%%%%%%%%%%%%%%%%%%%%%%%%%%%%%%%%%%%%
%%%%%%%%%%%%%%%%%%%%%%%%%%%%%%%%%%%%%%%%%%%%%%%%%%%%%%%%%%%%%%%%%%%%%%%%%%%%

\subsection{Sufficient criteria for negative exponent}
\label{sec:random-negative-criteria}

%%%%%%%%%%%%%%%%%%%%%%%%%%%%%%%%%%%%%%%%%%%%%%%%%%%%%%%%%%%%%%%%%%%%%%%%%%%%
%%%%%%%%%%%%%%%%%%%%%%%%%%%%%%%%%%%%%%%%%%%%%%%%%%%%%%%%%%%%%%%%%%%%%%%%%%%%

We now record practical criteria for checking \(\lambda_{\rm tr}<0\).
The sharp criterion is the existence of a finite block which is strictly contractive with positive probability.
Strict positivity and Doeblin minorization are useful sufficient mechanisms for producing such blocks.

The first criterion is intrinsic.
It says that, in the ergodic case, negativity of the Lyapunov exponent is equivalent to the existence of one fixed block length which contracts on a set of positive probability.

\begin{prop}
\label{prop:positive-probability-contracting-block}
    Assume that \(\theta\) is ergodic.
    Then \(\lambda_{\rm tr}<0\) if and only if there exists \(L\ge1\) such that
    \[
       \pr\left\{
          \omega:
          \kappa_{\omega;L:0}<1
       \right\}
       >
       0.
    \]
    Equivalently, \(\lambda_{\rm tr}<0\) if and only if there exists \(L\ge1\) such that
    \[
       \mathbb E\left[\log\kappa_{\omega;L:0}\right]<0,
    \]
    where the expectation is interpreted in the extended nonpositive sense.
\end{prop}

\begin{proof}
    Set \(F_L\left(\omega\right):=\log\kappa_{\omega;L:0}\), with the convention \(\log0=-\infty\).
    Since \(0\le\kappa_{\omega;L:0}\le1\), one has \(F_L\le0\).
    Assume first that \(\pr\left\{F_L<0\right\}>0\) for some \(L\ge1\).
    Then \(-F_L\) is a nonnegative extended random variable which is positive on a set of positive probability.
    Hence
    \[
       \mathbb E\left[F_L\right]
       =
       -\mathbb E\left[-F_L\right]
       <
       0.
    \]
    By the ergodic variational formula in \Cref{thm:random-lyapunov},
    \[
       \lambda_{\rm tr}
       =
       \inf_{n\ge1}
       {\dfrac{1}{n}}
       \mathbb E\left[\log\kappa_{\omega;n:0}\right].
    \]
    Therefore
    \[
       \lambda_{\rm tr}
       \le
       {\dfrac{1}{L}}\mathbb E\left[F_L\right]
       <
       0.
    \]
    Conversely, assume that \(\lambda_{\rm tr}<0\).
    The same variational formula gives
    \[
       \inf_{n\ge1}
       {\dfrac{1}{n}}
       \mathbb E\left[\log\kappa_{\omega;n:0}\right]
       <
       0.
    \]
    Thus there exists \(L\ge1\) such that \(\mathbb E\left[F_L\right]<0\).
    Since \(F_L\le0\), this is possible only if \(\pr\left\{F_L<0\right\}>0\).
    Equivalently,
    \[
       \pr\left\{
          \omega:
          \kappa_{\omega;L:0}<1
       \right\}
       >
       0.
    \]
\end{proof}

The preceding proposition reduces the problem to producing a block with \(\kappa_{\omega;L:0}<1\) on a set of positive probability.
A convenient way to certify this is through a Doeblin-type minorization.
The following elementary lemma records the contraction consequence of such a minorization.

\begin{lemma}
\label{lem:minorization-contracts}
    Let \(S:\matrices\to\matrices\) be CPTP.
    Assume that there are \(\varepsilon\in\left(0,1\right]\) and \(\tau\in\mcS_d\) such that
    \[
       S\left(X\right)
       \ge
       \varepsilon\,\tr{X}\tau
       \qquad
       \text{for all }X\ge0.
    \]
    Then
    \[
       \kappa_{\rm tr}\left(S\right)\le1-\varepsilon.
    \]
\end{lemma}

\begin{proof}
    Let \(R_\tau\left(X\right):=\tr{X}\tau\).
    The minorization assumption says that \(S-\varepsilon R_\tau\) is positive.
    Assume first that \(\varepsilon<1\).
    Then
    \[
       \widetilde S
       :=
       {\dfrac{S-\varepsilon R_\tau}{1-\varepsilon}}
    \]
    is positive and trace preserving.
    For every self-adjoint \(X\in\tracezero\), one has \(R_\tau\left(X\right)=0\), and therefore
    \[
       S\left(X\right)
       =
       \left(1-\varepsilon\right)\widetilde S\left(X\right).
    \]
    Since positive trace-preserving maps are trace-norm contractions on self-adjoint matrices,
    \[
       \norm{S\left(X\right)}_1
       \le
       \left(1-\varepsilon\right)\norm{X}_1.
    \]
    Taking the supremum over nonzero self-adjoint \(X\in\tracezero\) gives
    \[
       \kappa_{\rm tr}\left(S\right)\le1-\varepsilon.
    \]
        If \(\varepsilon=1\), then for every \(X\ge0\),
    \[
       S(X)-R_\tau(X)\ge0
       \qquad\text{and}\qquad
       \tr{S(X)-R_\tau(X)}=0.
    \]
    Hence \(S(X)=R_\tau(X)\) for every \(X\ge0\).
    By linearity, \(S=R_\tau\) on \(\matrices\), and therefore \(\kappa_{\rm tr}(S)=0\).
\end{proof}

Strict positivity is a familiar source of such a minorization.
Indeed, in finite dimension, a strictly positive channel sends the compact state space into the interior of the positive cone, and hence admits a uniform lower bound.
The next corollary applies this observation to random blocks.

\begin{cor}
    \label{cor:strict-positive-block-negative-exponent}
    Assume that \(\theta\) is ergodic.
    If there exists \(L\ge1\) such that
    \[
       \pr\left\{
          \omega:
          \Phi_{\omega;L:0}
          \text{ is strictly positive}
       \right\}
       >
       0,
    \]
    then \(\lambda_{\rm tr}<0\).
    In particular, if the cocycle is eventually strictly positive, meaning that for \(\pr\)-a.e. \(\omega\) there exists \(n=n\left(\omega\right)\ge1\) such that \(\Phi_{\omega;n:0}\) is strictly positive, then \(\lambda_{\rm tr}<0\).
\end{cor}

\begin{proof}
    Let
    \[
       E_L
       :=
       \left\{
          \omega:
          \Phi_{\omega;L:0}
          \text{ is strictly positive}
       \right\}.
    \]
    By assumption, \(\pr\left(E_L\right)>0\).
    This is a measurable set.
    Indeed, \(\omega\mapsto \Phi_{\omega;L:0}\) is measurable, and the map
    \[
       S\longmapsto
       \min_{\rho\in\mcS_d}
       \lambda_{\min}(S(\rho))
    \]
    is continuous in finite dimensions.
    Moreover, \(S\) is strictly positive if and only if this minimum is strictly positive.
    
    Fix \(\omega\in E_L\), and set \(S:=\Phi_{\omega;L:0}\).
    The map
    \[
       \rho\longmapsto\lambda_{\min}\left(S\left(\rho\right)\right)
    \]
    is continuous on the compact state space \(\mcS_d\).
    Since \(S\) is strictly positive, \(S\left(\rho\right)>0\) for every \(\rho\in\mcS_d\).
    Hence
    \[
       a_\omega
       :=
       \min_{\rho\in\mcS_d}
       \lambda_{\min}\left(S\left(\rho\right)\right)
       >
       0.
    \]
    Thus \(S\left(\rho\right)\ge a_\omega I\) for every \(\rho\in\mcS_d\).
    If \(X\ge0\) and \(X\ne0\), then \(X=\tr{X}\rho\) for some \(\rho\in\mcS_d\), so
    \[
       S\left(X\right)
       =
       \tr{X}S\left(\rho\right)
       \ge
       da_\omega\,\tr{X}{\dfrac Id}.
    \]
    The same inequality is trivial for \(X=0\).
    Set \(\varepsilon_\omega:=da_\omega\).
    Since \(S\) is trace preserving, \(0<\varepsilon_\omega\le1\).
    
    By \Cref{lem:minorization-contracts},
    \[
       \kappa_{\omega;L:0}
       =
       \kappa_{\rm tr}\left(S\right)
       \le
       1-\varepsilon_\omega
       <
       1.
    \]
    Therefore
    \[
       E_L
       \subset
       \left\{
          \omega:
          \kappa_{\omega;L:0}<1
       \right\}.
    \]
    Hence
    \[
       \pr\left\{
          \omega:
          \kappa_{\omega;L:0}<1
       \right\}
       >
       0.
    \]
    The conclusion follows from \Cref{prop:positive-probability-contracting-block}.
    
    If the cocycle is eventually strictly positive, then
    \[
       \pr\left(
          \bigcup_{L\ge1}
          \left\{
             \omega:
             \Phi_{\omega;L:0}
             \text{ is strictly positive}
          \right\}
       \right)
       =
       1.
    \]
    Since this is a countable union, there exists \(L\ge1\) such that
    \[
       \pr\left\{
          \omega:
          \Phi_{\omega;L:0}
          \text{ is strictly positive}
       \right\}
       >
       0.
    \]
    The first part applies.
\end{proof}

The strict-positivity assumption is stronger than necessary.
For the trace-Dobrushin exponent, it is enough to have a Doeblin minorization for a fixed random block, even if the block is not strictly positive.

\begin{cor}
\label{cor:random-doeblin-block}
    Assume that \(\theta\) is ergodic.
    Suppose that there exist \(L\ge1\), a measurable function \(\varepsilon:\Omega\to[0,1]\), and a measurable random state \(\tau:\Omega\to\mcS_d\) such that
    \[
       \Phi_{\omega;L:0}\left(X\right)
       \ge
       \varepsilon\left(\omega\right)\tr{X}\tau_\omega
       \qquad
       \text{for all }X\ge0
    \]
    for \(\pr\)-a.e. \(\omega\), and
    \[
       \mathbb E\left[
          \log\left(1-\varepsilon\right)
       \right]
       <
       0
    \]
    in the extended nonpositive sense.
    Then \(\lambda_{\rm tr}<0\).
    In particular, it is enough that \(\pr\left\{\varepsilon>0\right\}>0\).
\end{cor}

\begin{proof}
    By \Cref{lem:minorization-contracts},
    \[
       \kappa_{\omega;L:0}
       \le
       1-\varepsilon\left(\omega\right)
    \]
    for \(\pr\)-a.e. \(\omega\).
    Hence
    \[
       \mathbb E\left[
          \log\kappa_{\omega;L:0}
       \right]
       \le
       \mathbb E\left[
          \log\left(1-\varepsilon\right)
       \right]
       <
       0.
    \]
    The conclusion follows from \Cref{prop:positive-probability-contracting-block}.
    If \(\pr\left\{\varepsilon>0\right\}>0\), then \(\log\left(1-\varepsilon\right)<0\) on a set of positive probability and is never positive.
    Thus, its extended expectation is strictly negative.
\end{proof}

One common way to verify a Doeblin minorization is to find a common positive lower bound for the images of rank-one projections.
This is the block-level version of the quantum Markov--Dobrushin certificate.

\begin{cor}
\label{cor:random-md-block}
    Assume that \(\theta\) is ergodic.
    Suppose that there exist \(L\ge1\) and a measurable positive operator \(B_\omega\ge0\) such that
    \[
       B_\omega
       \le
       \Phi_{\omega;L:0}\left(P\right)
       \qquad
       \text{for every rank-one projection }P
    \]
    for \(\pr\)-a.e. \(\omega\), and
    \[
       \mathbb E\left[
          \log\left(1-\tr{B_\omega}\right)
       \right]
       <
       0.
    \]
    Then \(\lambda_{\rm tr}<0\).
\end{cor}

\begin{proof}
    Let \(\varepsilon\left(\omega\right):=\tr{B_\omega}\).
    If \(\varepsilon\left(\omega\right)>0\), set \(\tau_\omega:=B_\omega/\varepsilon\left(\omega\right)\), and choose \(\tau_\omega\) arbitrarily when \(\varepsilon\left(\omega\right)=0\).
    
    For every state \(\rho\), write a spectral decomposition \(\rho=\sum_i p_iP_i\) into rank-one projections.
    Then
    \[
       \Phi_{\omega;L:0}\left(\rho\right)
       =
       \sum_i p_i\Phi_{\omega;L:0}\left(P_i\right)
       \ge
       \sum_i p_iB_\omega
       =
       B_\omega.
    \]
    By homogeneity, for every \(X\ge0\),
    \[
       \Phi_{\omega;L:0}\left(X\right)
       \ge
       \tr{X}B_\omega
       =
       \varepsilon\left(\omega\right)\tr{X}\tau_\omega.
    \]
    Thus \Cref{cor:random-doeblin-block} applies.
\end{proof}

The preceding criteria show that eventual strict positivity is a sufficient mechanism for a negative trace-Dobrushin exponent.
It is not a necessary mechanism.
The next example shows that replacement mixing can occur even when no power of the channel is strictly positive.

\begin{example}
\label{ex:amplitude-random-strictness}
    Let the base space be trivial, and let the single channel be the qubit amplitude-damping channel \(T_\gamma\), where \(0<\gamma<1\).
    In Bloch coordinates, centered state differences transform as
    \[
       \left(x,y,z\right)
       \longmapsto
       \left(
          \sqrt{1-\gamma}\,x,
          \sqrt{1-\gamma}\,y,
          \left(1-\gamma\right)z
       \right).
    \]
    Since the trace norm of a centered qubit difference equals the Euclidean norm of the corresponding Bloch-vector difference,
    \[
       \kappa_{\rm tr}\left(T_\gamma\right)
       =
       \sqrt{1-\gamma}
       <
       1.
    \]
    The same computation gives
    \[
       \kappa_{\rm tr}\left(T_\gamma^n\right)
       =
       \left(\sqrt{1-\gamma}\right)^n.
    \]
    Thus
    \[
       \lambda_{\rm tr}
       =
       \log\sqrt{1-\gamma}
       <
       0.
    \]
    The channel therefore has exponential replacement mixing, and the attracting replacement state is the pure state \(\ket{0}\bra{0}\).
    
    However,
    \[
       T_\gamma^n\left(\ket{0}\bra{0}\right)
       =
       \ket{0}\bra{0}
       \qquad
       \text{for every }n\ge1.
    \]
    Thus no power of \(T_\gamma\) maps every nonzero positive operator to a strictly positive operator.
    Consequently, \(T_\gamma\) is not eventually strictly positive.
    This shows that, for CPTP cocycles, eventual strict positivity is only a sufficient subclass of the negative trace-Dobrushin exponent regime.
\end{example}

%%%%%%%%%%%%%%%%%%%%%%%%%%%%%%%%%%%%%%%%%%%%%%%%%%%%%%%%%%%%%%%%%%%%%%%%%%%%
%%%%%%%%%%%%%%%%%%%%%%%%%%%%%%%%%%%%%%%%%%%%%%%%%%%%%%%%%%%%%%%%%%%%%%%%%%%%

\paragraph{Removing ergodicity.}

    The ergodicity assumption in the preceding positive-probability criteria is used only to rule out invariant components on which no good block occurs.
    There is a direct non-ergodic replacement that does not require choosing ergodic components.
    
    For a measurable set \(E\subset\Omega\), define its two-sided saturation by
    \[
       \operatorname{Sat}(E)
       :=
       \bigcup_{j\in\mathbb Z}\theta^{-j}E .
    \]
    Thus, in the general probability-preserving case, the ergodic assumption
    \[
       \pr(E)>0
    \]
    should be replaced by
    \[
       \pr\!\left(\operatorname{Sat}(E)\right)=1.
    \]
    Equivalently, every invariant measurable set \(A\in\mcF\) with \(\pr(A)>0\) satisfies
    \[
       \pr(A\cap E)>0.
    \]
    Under this replacement, the preceding positive-probability sufficient criteria remain valid without assuming that \(\theta\) is ergodic.
    
    Indeed, let
    \[
       F_n(\omega):=\log\kappa_{\omega;n:0},
       \qquad F_n\le0,
    \]
    and let
    \[
       \lambda_{\rm tr}(\omega)
       =
       \lim_{n\to\infty}\frac1nF_n(\omega)
    \]
    be the Kingman exponent.  Suppose that \(E=\bigcup_{r\ge1}E_r\), where for each \(r\) there is
    a block length \(L_r\) such that
    \[
       E_r\subset
       \left\{
          \omega:\kappa_{\omega;L_r:0}<1
       \right\}.
    \]
    Assume
    \[
       \pr \!\left(\operatorname{Sat}(E)\right)=1.
    \]
    We claim that \(\lambda_{\rm tr}<0\) a.s.  Let
    \[
       Z:=\{\omega:\lambda_{\rm tr}(\omega)=0\}.
    \]
    Since \(\lambda_{\rm tr}\le0\) and is \(\theta\)-invariant, \(Z\) is invariant.  If
    \(\pr (Z)>0\), then the saturation assumption gives
    \[
       \pr (Z\cap E)>0.
    \]
    Since \(E=\bigcup_{r\ge1}E_r\) is a countable union, there exists \(r\) such that
    \[
       \pr (Z\cap E_r)>0.
    \]
    On \(E_r\) one has \(F_{L_r}<0\), while \(F_{L_r}\le0\) everywhere.  Hence
    \[
       \int_Z F_{L_r}\,d\pr <0.
    \]
    Applying Kingman's variational formula to the restricted probability-preserving system on
    \(Z\), normalized by \(\pr (Z)\), gives
    \[
       \int_Z \lambda_{\rm tr}\,d\pr 
       \le
       \frac1{L_r}\int_Z F_{L_r}\,d\pr 
       <0.
    \]
    But by definition \(\lambda_{\rm tr}=0\) on \(Z\), a contradiction.  Therefore
    \(\pr (Z)=0\), and hence \(\lambda_{\rm tr}<0\) a.s.
    
    For example, in \Cref{cor:random-doeblin-block}, if ergodicity is dropped, it is enough to assume
    \[
       \pr \!\left(
          \operatorname{Sat}\{\omega:\varepsilon(\omega)>0\}
       \right)=1.
    \]
    Indeed, the Doeblin minorization gives
    \[
       \kappa_{\omega;L:0}
       \le
       1-\varepsilon(\omega),
    \]
    so the event \(\{\varepsilon>0\}\) is contained in
    \[
       \left\{
          \omega:\kappa_{\omega;L:0}<1
       \right\}.
    \]
    The preceding saturation principle therefore gives \(\lambda_{\rm tr}<0\) a.s.

%%%%%%%%%%%%%%%%%%%%%%%%%%%%%%%%%%%%%%%%%%%%%%%%%%%%%%%%%%%%%%%%%%%%%%%%%%%%
%%%%%%%%%%%%%%%%%%%%%%%%%%%%%%%%%%%%%%%%%%%%%%%%%%%%%%%%%%%%%%%%%%%%%%%%%%%%

\subsection{Stationary random states and replacement estimates}
\label{sec:random-pullback-state}

%%%%%%%%%%%%%%%%%%%%%%%%%%%%%%%%%%%%%%%%%%%%%%%%%%%%%%%%%%%%%%%%%%%%%%%%%%%%
%%%%%%%%%%%%%%%%%%%%%%%%%%%%%%%%%%%%%%%%%%%%%%%%%%%%%%%%%%%%%%%%%%%%%%%%%%%%

The Lyapunov exponent controls the diameter of long channel products.
To obtain convergence to a rank-one channel, we must identify the moving center of the image sets.
This center is constructed by fixing the terminal time and sending the initial time to the remote past.

For integers \(s<t\), recall that
\[
   \Phi_{\omega;t:s}
   :=
   \Phi_{\theta^{t-1}\omega}
   \circ
   \Phi_{\theta^{t-2}\omega}
   \circ\cdots\circ
   \Phi_{\theta^s\omega},
   \qquad
   \kappa_{\omega;t:s}
   :=
   \kappa_{\rm tr}\left(\Phi_{\omega;t:s}\right).
\]
Thus the forward product is \(\Phi_\omega^{(n)}=\Phi_{\omega;n:0}\), and the pullback product ending at \(\omega\) is
\[
   \Phi_{\omega;0:-n}
   =
   \Phi_{\theta^{-1}\omega}
   \circ
   \Phi_{\theta^{-2}\omega}
   \circ\cdots\circ
   \Phi_{\theta^{-n}\omega}.
\]
The product \(\Phi_{\omega;0:-n}\) is the chronological product over the past interval from time \(-n\) to time \(0\).
It is not the reverse-ordered product on the same interval.

\begin{prop}
\label{prop:random-pullback-state}
    Assume that \(\kappa_{\omega;0:-n}\to0\) as \(n\to\infty\) for \(\pr\)-a.e. \(\omega\).
    Fix \(\tau\in\mcS_d\), and set \(\rho_{\omega,n}:=\Phi_{\omega;0:-n}\left(\tau\right)\).
    Then \(\rho_{\omega,n}\) converges in trace norm for \(\pr\)-a.e. \(\omega\), and the limit
    \(\rho_\omega:=\lim_{n\to\infty}\Phi_{\omega;0:-n}\left(\tau\right)\) is independent of the choice of \(\tau\).
    After modifying on a null set, the map \(\omega\mapsto\rho_\omega\) is measurable and satisfies
    \[
       \Phi_\omega\left(\rho_\omega\right)=\rho_{\theta\omega}
       \qquad
       \text{for }\pr\text{-a.e. }\omega.
    \]
    For the associated random replacement channel \(R_\omega\left(X\right):=\tr{X}\rho_\omega\), the following estimates hold for \(\pr\)-a.e. \(\omega\) and every \(n\ge1\):
    \[
    \begin{aligned}
       \sup_{\sigma\in\mcS_d}
       \norm{
          \Phi_{\omega;0:-n}\left(\sigma\right)-\rho_\omega
       }_1
       &\le
       2\kappa_{\omega;0:-n},
       \\
       \norm{
          \rho_{\omega,n}-\rho_\omega
       }_1
       &\le
       2\kappa_{\omega;0:-n},
       \\
       \norm{
          \Phi_{\omega;0:-n}-R_\omega
       }_{1\to1}
       &\le
       4\kappa_{\omega;0:-n}.
    \end{aligned}
    \]
    Moreover, if \(\eta:\Omega\to\mcS_d\) is any measurable random state, then
    \[
       \norm{
          \Phi_{\omega;0:-n}
          \left(\eta_{\theta^{-n}\omega}\right)
          -
          \rho_\omega
       }_1
       \le
       2\kappa_{\omega;0:-n}
    \]
    for \(\pr\)-a.e. \(\omega\) and every \(n\ge1\).
    Finally, \(\rho\) is the unique stationary random state, up to \(\pr\)-a.e. equality.
\end{prop}

\begin{proof}
    Let
    \[
       \Omega_0
       :=
       \left\{
          \omega:
          \kappa_{\omega;0:-n}\to0
       \right\}.
    \]
    Then \(\pr\left(\Omega_0\right)=1\).
    Fix \(\omega\in\Omega_0\).
    Let \(r>n\).
    The cocycle identity gives
    \[
       \rho_{\omega,r}
       =
       \Phi_{\omega;0:-n}
       \left(
          \Phi_{\omega;-n:-r}\left(\tau\right)
       \right),
       \qquad
       \rho_{\omega,n}
       =
       \Phi_{\omega;0:-n}\left(\tau\right).
    \]
    Both inputs of \(\Phi_{\omega;0:-n}\) are states.
    By \Cref{prop:dobr},
    \[
       \norm{
          \rho_{\omega,r}-\rho_{\omega,n}
       }_1
       \le
       2\kappa_{\rm tr}\left(\Phi_{\omega;0:-n}\right)
       =
       2\kappa_{\omega;0:-n}.
    \]
    Since \(\kappa_{\omega;0:-n}\to0\), the sequence \(\left(\rho_{\omega,n}\right)_{n\ge1}\) is Cauchy.
    The state space \(\mcS_d\) is complete in trace norm.
    Hence \(\rho_{\omega,n}\) has a trace-norm limit.
    
    The same estimate proves the independence of the reference state.
    Indeed, if \(\tau,\tau'\in\mcS_d\), then
    \[
       \norm{
          \Phi_{\omega;0:-n}\left(\tau\right)
          -
          \Phi_{\omega;0:-n}\left(\tau'\right)
       }_1
       \le
       2\kappa_{\omega;0:-n}
       \longrightarrow0.
    \]
    
    Measurability follows from pointwise convergence.
    For each \(n\), the map \(\omega\mapsto\rho_{\omega,n}\) is measurable because it is obtained from finitely many measurable compositions.
    The pointwise limit on \(\Omega_0\) is therefore measurable.
    Fix a state \(\tau_0\in\mcS_d\), and define \(\rho_\omega:=\tau_0\) on \(\Omega\setminus\Omega_0\).
    
    We now prove stationarity.
    Let
    \[
       \Omega_1:=\Omega_0\cap\theta^{-1}\Omega_0 .
    \]
    Then \(\pr(\Omega_1)=1\).
    For \(\omega\in\Omega_1\), the pullback limits exist at both \(\omega\) and \(\theta\omega\).
    Using continuity and the cocycle identity,
    \[
       \Phi_\omega(\rho_\omega)
       =
       \lim_{n\to\infty}
       \Phi_\omega
       \left(
          \Phi_{\omega;0:-n}(\tau)
       \right)
       =
       \lim_{n\to\infty}
       \Phi_{\omega;1:-n}(\tau).
    \]
    On the other hand,
    \[
       \rho_{\theta\omega}
       =
       \lim_{m\to\infty}
       \Phi_{\theta\omega;0:-m}(\tau).
    \]
    Taking \(m=n+1\), we have
    \[
       \Phi_{\theta\omega;0:-(n+1)}
       =
       \Phi_{\omega;1:-n}.
    \]
    Therefore
    \[
       \rho_{\theta\omega}
       =
       \lim_{n\to\infty}
       \Phi_{\omega;1:-n}(\tau).
    \]
    Thus
    \[
       \Phi_\omega(\rho_\omega)=\rho_{\theta\omega}
    \]
    for \(\pr\)-a.e. \(\omega\).
    
    Let
    \[
       S
       :=
       \{\omega:\Phi_\omega(\rho_\omega)=\rho_{\theta\omega}\}.
    \]
    Then \(\pr(S)=1\).
    Set
    \[
       \Omega_*
       :=
       \Omega_0
       \cap
       S
       \cap
       \bigcap_{j\ge1}
       \{\omega:\theta^{-j}\omega\in S\}.
    \]
    Since \(\theta\) is invertible and measure preserving, \(\pr(\Omega_*)=1\).
    For every \(\omega\in\Omega_*\), stationarity holds at \(\theta^{-j}\omega\) for every \(j\ge0\).
    Hence, for every \(n\ge1\),
    \[
       \rho_\omega
       =
       \Phi_{\omega;0:-n}
       \left(
          \rho_{\theta^{-n}\omega}
       \right).
    \]
    
    Applying \Cref{prop:dobr} to the two inputs \(\sigma\) and \(\rho_{\theta^{-n}\omega}\) gives
    \[
       \norm{
          \Phi_{\omega;0:-n}\left(\sigma\right)-\rho_\omega
       }_1
       \le
       2\kappa_{\omega;0:-n}
    \]
    for every \(\sigma\in\mcS_d\).
    Taking the supremum over \(\sigma\in\mcS_d\) gives the state-level estimate.
    Taking \(\sigma=\tau\) gives
    \[
       \norm{
          \rho_{\omega,n}-\rho_\omega
       }_1
       \le
       2\kappa_{\omega;0:-n}.
    \]
    
    The replacement-channel estimate follows from \Cref{prop:det-radius-replacement}.
    Indeed, for the deterministic product \(\Phi_{\omega;0:-n}\), use the reference state \(\rho_{\theta^{-n}\omega}\).
    The associated evolved reference state is
    \[
       \Phi_{\omega;0:-n}
       \left(
          \rho_{\theta^{-n}\omega}
       \right)
       =
       \rho_\omega.
    \]
    Therefore, the corresponding replacement channel is exactly \(R_\omega\).
    Hence
    \[
       \norm{
          \Phi_{\omega;0:-n}-R_\omega
       }_{1\to1}
       \le
       4\kappa_{\rm tr}
       \left(\Phi_{\omega;0:-n}\right)
       =
       4\kappa_{\omega;0:-n}.
    \]
    
    The random-boundary estimate follows by taking \(\sigma=\eta_{\theta^{-n}\omega}\) in the state-level estimate.
    This is valid because \(\eta_{\theta^{-n}\omega}\in\mcS_d\).
    
    It remains to prove uniqueness.
    Let \(\eta:\Omega\to\mcS_d\) be any stationary random state.
    Let
    \[
       S_\eta
       :=
       \left\{
          \omega:
          \Phi_\omega\left(\eta_\omega\right)=\eta_{\theta\omega}
       \right\}.
    \]
    Then \(\pr\left(S_\eta\right)=1\).
    Replacing \(\Omega_*\) by
    \[
       \Omega_*
       \cap
       \bigcap_{j\ge1}\theta^j S_\eta,
    \]
    we may assume that the stationarity identity for \(\eta\) can also be iterated backwards.
    Thus, for every \(\omega\in\Omega_*\) and every \(n\ge1\),
    \[
       \eta_\omega
       =
       \Phi_{\omega;0:-n}
       \left(
          \eta_{\theta^{-n}\omega}
       \right).
    \]
    By the random-boundary estimate applied to \(\eta\),
    \[
       \norm{
          \eta_\omega-\rho_\omega
       }_1
       \le
       2\kappa_{\omega;0:-n}.
    \]
    Letting \(n\to\infty\) gives \(\eta_\omega=\rho_\omega\) for \(\pr\)-a.e. \(\omega\).
    Thus, the stationary random state is unique up to \(\pr\)-a.e. equality.
\end{proof}

The pullback construction identifies the stationary random center.
Once this center is known, forward trace-memory loss gives convergence to the replacement channel at the future endpoint.

\begin{prop}
\label{prop:random-forward-replacement}
    Let \(\rho:\Omega\to\mcS_d\) be a stationary random state, and define \(R_\omega\left(X\right):=\tr{X}\rho_\omega\).
    Then, for \(\pr\)-a.e. \(\omega\) and every \(n\ge1\),
    \[
    \begin{aligned}
       \sup_{\sigma\in\mcS_d}
       \norm{
          \Phi_{\omega;n:0}\left(\sigma\right)-\rho_{\theta^n\omega}
       }_1
       &\le
       2\kappa_{\omega;n:0},
       \\
       \norm{
          \Phi_{\omega;n:0}-R_{\theta^n\omega}
       }_{1\to1}
       &\le
       4\kappa_{\omega;n:0}.
    \end{aligned}
    \]
    Consequently, if \(\kappa_{\omega;n:0}\to0\) for \(\pr\)-a.e. \(\omega\), then
    \[
       \norm{
          \Phi_{\omega;n:0}-R_{\theta^n\omega}
       }_{1\to1}
       \longrightarrow0
       \qquad
       \text{for }\pr\text{-a.e. }\omega.
    \]
\end{prop}
    
\begin{proof}
    Let
    \[
       S
       :=
       \left\{
          \omega:
          \Phi_\omega\left(\rho_\omega\right)=\rho_{\theta\omega}
       \right\}.
    \]
    Since \(\rho\) is stationary, \(\pr\left(S\right)=1\).
    Replacing \(S\) by
    \[
       \Omega_*
       :=
       \bigcap_{j\ge0}\theta^{-j}S,
    \]
    we obtain another full-measure set.
    For every \(\omega\in\Omega_*\), the stationarity identity holds at \(\omega,\theta\omega,\theta^2\omega,\ldots\).
    Therefore, for every \(n\ge1\),
    \[
       \Phi_{\omega;n:0}\left(\rho_\omega\right)
       =
       \rho_{\theta^n\omega}.
    \]
    
    Fix \(\omega\in\Omega_*\), \(n\ge1\), and \(\sigma\in\mcS_d\).
    By \Cref{prop:dobr},
    \[
       \norm{
          \Phi_{\omega;n:0}\left(\sigma\right)-\rho_{\theta^n\omega}
       }_1
       =
       \norm{
          \Phi_{\omega;n:0}\left(\sigma\right)
          -
          \Phi_{\omega;n:0}\left(\rho_\omega\right)
       }_1
       \le
       2\kappa_{\omega;n:0}.
    \]
    Taking the supremum over \(\sigma\in\mcS_d\) gives the state-level estimate.
    
    The replacement-channel estimate follows from \Cref{prop:det-radius-replacement}.
    Apply that proposition to the deterministic product \(\Phi_{\omega;n:0}\) with reference state \(\rho_\omega\).
    The evolved reference state is
    \[
       \Phi_{\omega;n:0}\left(\rho_\omega\right)
       =
       \rho_{\theta^n\omega}.
    \]
    Thus the associated replacement channel is \(R_{\theta^n\omega}\), and hence
    \[
       \norm{
          \Phi_{\omega;n:0}-R_{\theta^n\omega}
       }_{1\to1}
       \le
       4\kappa_{\rm tr}
       \left(
          \Phi_{\omega;n:0}
       \right)
       =
       4\kappa_{\omega;n:0}.
    \]
    
    If \(\kappa_{\omega;n:0}\to0\) on a full-measure set, then intersecting that set with \(\Omega_*\) gives the claimed convergence.
\end{proof}

By combining \Cref{prop:random-pullback-state} and \Cref{prop:random-forward-replacement} we obtain the following corollary. 

\begin{cor}
\label{cor:random-two-sided-replacement}
    Assume that
    \[
       \kappa_{\omega;0:-n}\to0
       \qquad
       \text{and}
       \qquad
       \kappa_{\omega;n:0}\to0
    \]
    as \(n\to\infty\), for \(\pr\)-a.e. \(\omega\).
    Then there is a unique stationary random state \(\rho:\Omega\to\mcS_d\), up to \(\pr\)-a.e. equality.
    With \(R_\omega\left(X\right):=\tr{X}\rho_\omega\), we have that as $n \to \infty$
    \[
       \norm{
          \Phi_{\omega;0:-n}-R_\omega
       }_{1\to1}
       \longrightarrow0
       \qquad
       \text{and}
       \qquad
       \norm{
          \Phi_{\omega;n:0}-R_{\theta^n\omega}
       }_{1\to1}
       \longrightarrow0
    \]
    for \(\pr\)-a.e. \(\omega\).
\end{cor}
\begin{proof}
    By the assumption \(\kappa_{\omega;0:-n}\to0\) for \(\pr\)-a.e. \(\omega\), \Cref{prop:random-pullback-state} gives a stationary random state
    \[
       \rho_\omega
       =
       \lim_{n\to\infty}
       \Phi_{\omega;0:-n}(\tau),
    \]
    independent of the reference state \(\tau\). 
    The same proposition gives uniqueness up to \(\pr\)-a.e. equality and the estimate
    \[
       \left\|
          \Phi_{\omega;0:-n}-R_\omega
       \right\|_{1\to1}
       \le
       4\kappa_{\omega;0:-n}
    \]
    on a full-measure set, for every \(n\ge1\). Hence
    \[
       \left\|
          \Phi_{\omega;0:-n}-R_\omega
       \right\|_{1\to1}
       \to0
    \]
    for \(\pr\)-a.e. \(\omega\).
    Since \(\rho\) is stationary, \Cref{prop:random-forward-replacement} applies. 
    Therefore, on another full-measure set,
    \[
       \left\|
          \Phi_{\omega;n:0}-R_{\theta^n\omega}
       \right\|_{1\to1}
       \le
       4\kappa_{\omega;n:0}
    \]
    for every \(n\ge1\). By the assumption
    \[
       \kappa_{\omega;n:0}\to0,
    \]
    we obtain
    \[
       \left\|
          \Phi_{\omega;n:0}-R_{\theta^n\omega}
       \right\|_{1\to1}
       \to0
    \]
    for \(\pr\)-a.e. \(\omega\).
\end{proof}

%%%%%%%%%%%%%%%%%%%%%%%%%%%%%%%%%%%%%%%%%%%%%%%%%%%%%%%%%%%%%%%%%%%%%%%%%%%%
%%%%%%%%%%%%%%%%%%%%%%%%%%%%%%%%%%%%%%%%%%%%%%%%%%%%%%%%%%%%%%%%%%%%%%%%%%%%

\subsection{Quenched rates of convergence}
\label{section:random_rates}

%%%%%%%%%%%%%%%%%%%%%%%%%%%%%%%%%%%%%%%%%%%%%%%%%%%%%%%%%%%%%%%%%%%%%%%%%%%%
%%%%%%%%%%%%%%%%%%%%%%%%%%%%%%%%%%%%%%%%%%%%%%%%%%%%%%%%%%%%%%%%%%%%%%%%%%%%

The Dobrushin-Lyapunov exponent gives the exponential rate at which the centered part of the random product contracts.
The pullback construction identifies the moving replacement center.
We now combine these two facts to obtain quenched replacement estimates in the induced trace norm.

\randomrates*

\begin{proof}
    By \Cref{prop:memory-loss-negative-exponent-equivalence}, the assumption that \(\lambda_{\rm tr}(\omega)<0\) $\pr$-a.s. implies both
    \[
       \kappa_{\omega;n:0}\to0
       \qquad
       \text{and}
       \qquad
       \kappa_{\omega;0:-n}\to0
    \]
    for \(\pr\)-a.e. \(\omega\). 
    Hence \Cref{cor:random-two-sided-replacement} gives a stationary random state   \(\rho:\Omega\to\mcS_d\), unique up to \(\pr\)-a.e. equality, and the qualitative forward and pullback replacement limits. 
    In particular,
    \[
       \Phi_\omega(\rho_\omega)=\rho_{\theta\omega}
       \qquad
       \text{for }\pr\text{-a.e. }\omega,
    \]
    so \(\rho\) is dynamically stationary in the sense of  \Cref{def:dynamically-stationary-random-state}.

    It remains to prove the exponential estimates.
    Let \(E_0\) be the full-measure set on which
    \[
       \lambda_{\rm tr}(\omega)<0,
       \qquad
       \lambda_{\rm tr}(\theta\omega)=\lambda_{\rm tr}(\omega),
    \]
    and on which the forward and pullback Lyapunov limits hold. Replacing \(E_0\) by
    \[
       E:=\bigcap_{k\in\mathbb Z}\theta^{-k}E_0,
    \]
    we may assume that \(E\) is \(\theta\)-invariant and still has full measure.
    Define
    \[
       \bar\lambda_{\rm tr}(\omega)
       :=
       \begin{cases}
          \lambda_{\rm tr}(\omega), & \omega\in E,\\
          -2, & \omega\notin E.
       \end{cases}
    \]
    Then \(\bar\lambda_{\rm tr}\) is measurable, \(\theta\)-invariant, agrees with \(\lambda_{\rm tr}\) almost surely, and satisfies
    \[
       \bar\lambda_{\rm tr}(\omega)<0
       \qquad
       \text{for every }\omega.
    \]
    Now set
    \[
       \beta(\omega)
       :=
       \begin{cases}
          \dfrac12\bar\lambda_{\rm tr}(\omega),
          & -\infty<\bar\lambda_{\rm tr}(\omega)<0,\\[0.8em]
          -1,
          & \bar\lambda_{\rm tr}(\omega)=-\infty.
       \end{cases}
    \]
    Then \(\beta:\Omega\to(-\infty,0)\) is measurable and \(\theta\)-invariant.
    Moreover,
    \[
       \lambda_{\rm tr}(\omega)<\beta(\omega)<0
       \qquad
       \text{for }\pr\text{-a.e. }\omega.
    \]
    Define
    \[
       \widetilde C_\beta^+(\omega)
       :=
       \max\left\{
       1 \, , \,
       \sup_{n\ge1}
       e^{-\beta(\omega)n}\kappa_{\omega;n:0}
       \right\}, 
    \]
    and
    \[
       \widetilde C_\beta^-(\omega)
       :=
       \max\left\{
       1 \, , \,
       \sup_{n\ge1}
       e^{-\beta(\omega)n}\kappa_{\omega;0:-n}\right\}. 
    \]
    These random variables are measurable because the suprema are countable.

    By \Cref{thm:random-lyapunov} and \Cref{prop:pullback-same-exponent},
    \[
       \lim_{n\to\infty}
       \frac1n\log\kappa_{\omega;n:0}
       =
       \lambda_{\rm tr}(\omega),
       \qquad
       \lim_{n\to\infty}
       \frac1n\log\kappa_{\omega;0:-n}
       =
       \lambda_{\rm tr}(\omega)
    \]
    for \(\pr\)-a.e. \(\omega\), with the convention \(\log0=-\infty\).
    Since
    \[
       \lambda_{\rm tr}(\omega)<\beta(\omega),
    \]
    it follows that, for \(\pr\)-a.e. \(\omega\), both suprema above are finite.
    Indeed, for such \(\omega\), there exists \(N(\omega)\) such that for all
    \(n\ge N(\omega)\),
    \[
       \kappa_{\omega;n:0}
       \le
       e^{\beta(\omega)n},
       \qquad
       \kappa_{\omega;0:-n}
       \le
       e^{\beta(\omega)n}.
    \]
    Hence, the tails of the two suprema are bounded by \(1\), and only finitely many initial terms can make the suprema larger than \(1\).
    Redefine \(\widetilde C_\beta^+\) and \(\widetilde C_\beta^-\) to be \(1\) on the null sets where they are not finite, and call the resulting random variables \(C_\beta^+\) and \(C_\beta^-\). Then \(C_\beta^+\) and \(C_\beta^-\) are measurable  and a.s. finite, and for \(\pr\)-a.e. \(\omega\) and every \(n\ge1\),
    \[
       \kappa_{\omega;n:0}
       \le
       C_\beta^+(\omega)e^{\beta(\omega)n},
       \qquad
       \kappa_{\omega;0:-n}
       \le
       C_\beta^-(\omega)e^{\beta(\omega)n}.
    \]

    Combining these estimates with \Cref{prop:random-forward-replacement} gives
    \[
       \norm{
          \Phi_{\omega;n:0}-R_{\theta^n\omega}
       }_{1\to1}
       \le
       4\kappa_{\omega;n:0}
       \le
       4C_\beta^+(\omega)e^{\beta(\omega)n}.
    \]
    Similarly, combining the pullback estimate from \Cref{prop:random-pullback-state} gives
    \[
       \norm{
          \Phi_{\omega;0:-n}-R_\omega
       }_{1\to1}
       \le
       4\kappa_{\omega;0:-n}
       \le
       4C_\beta^-(\omega)e^{\beta(\omega)n}.
    \]
    This proves the two claimed operator-norm estimates.
\end{proof}

The particular choice of \(\beta\) is not essential. 
The same proof works for any measurable \(\theta\)-invariant random variable
\[
   \beta:\Omega\to(-\infty,0)
\]
satisfying
\[
   \lambda_{\rm tr}(\omega)<\beta(\omega)<0
   \qquad
   \text{for }\pr\text{-a.e. }\omega.
\]

We conclude this section with an example where a unique stationary state, but $\lambda_{\rm tr}$ is not negative. 

\begin{example}
\label{ex:unique-stationary-not-negative-exponent}
    Take the deterministic one-point base
    \[
       \Omega=\{\omega_0\},
       \qquad
       \theta\omega_0=\omega_0,
       \qquad
       \pr(\{\omega_0\})=1.
    \]
    Let \(d=2\), and define a quantum channel \(\Phi:M_2\to M_2\) by
    \[
       \Phi(X)
       =
       K_0XK_0^\dagger+K_1XK_1^\dagger,
       \qquad
       K_0:=|1\rangle\langle0|,
       \qquad
       K_1:=|0\rangle\langle1|.
    \]
    Equivalently, for
    \[
       X=
       \begin{pmatrix}
          x_{00} & x_{01}\\
          x_{10} & x_{11}
       \end{pmatrix},
    \]
    one has
    \[
       \Phi(X)
       =
       \begin{pmatrix}
          x_{11} & 0\\
          0 & x_{00}
       \end{pmatrix}.
    \]
    Since
    \[
       K_0^\dagger K_0+K_1^\dagger K_1
       =
       |0\rangle\langle0|+|1\rangle\langle1|
       =
       I,
    \]
    the map \(\Phi\) is completely positive and trace-preserving.
    A dynamically stationary random state is, in this deterministic case, simply a
    fixed state of \(\Phi\).  Let
    \[
       \rho=
       \begin{pmatrix}
          a & z\\
          \overline z & 1-a
       \end{pmatrix}
       \in\mcS_2 .
    \]
    Then
    \[
       \Phi(\rho)
       =
       \begin{pmatrix}
          1-a & 0\\
          0 & a
       \end{pmatrix}.
    \]
    Hence \(\Phi(\rho)=\rho\) forces \(z=0\) and \(a=1/2\). 
    Therefore, the unique dynamically stationary state is
    \[
       \rho_*=\frac12 I.
    \]
    However, the trace-Dobrushin Lyapunov exponent is not negative.  Let
    \[
       X:=|0\rangle\langle0|-|1\rangle\langle1|
       =
       \begin{pmatrix}
          1 & 0\\
          0 & -1
       \end{pmatrix}.
    \]
    Then \(X=X^\dagger\), \(\tr X=0\), and \(\norm{X}_1=2\).  Moreover,
    \[
       \Phi(X)=-X,
    \]
    and hence
    \[
       \Phi^n(X)=(-1)^nX
       \qquad
       \forall n\ge1.
    \]
    Therefore
    \[
       \kappa_{\rm tr}(\Phi^n)
       \ge
       \frac{\norm{\Phi^n(X)}_1}{\norm{X}_1}
       =
       1.
    \]
    Since \(\kappa_{\rm tr}(\Phi^n)\le1\) for every CPTP map, we get
    \[
       \kappa_{\rm tr}(\Phi^n)=1
       \qquad
       \forall n\ge1.
    \]
    Consequently,
    \[
       \lambda_{\rm tr}
       =
       \lim_{n\to\infty}
       \frac1n\log\kappa_{\rm tr}(\Phi^n)
       =
       0.
    \]

    Thus, the uniqueness of the dynamically stationary random state does not imply
    \(\lambda_{\rm tr}<0\).  In particular, uniqueness alone does not imply replacement
    mixing.
\end{example}

%%%%%%%%%%%%%%%%%%%%%%%%%%%%%%%%%%%%%%%%%%%%%%%%%%%%%%%%%%%%%%%%%%%%%%%%%%%%
%%%%%%%%%%%%%%%%%%%%%%%%%%%%%%%%%%%%%%%%%%%%%%%%%%%%%%%%%%%%%%%%%%%%%%%%%%%%

\subsection{Annealed bounds under stochastic mixing}
\label{sec:annealed-mixing}

%%%%%%%%%%%%%%%%%%%%%%%%%%%%%%%%%%%%%%%%%%%%%%%%%%%%%%%%%%%%%%%%%%%%%%%%%%%%
%%%%%%%%%%%%%%%%%%%%%%%%%%%%%%%%%%%%%%%%%%%%%%%%%%%%%%%%%%%%%%%%%%%%%%%%%%%% 

We now impose a stochastic decorrelation condition on the random channel environment.  
The object to which the condition is applied is the stationary two-sided sequence of one-step channels
\[
    X_j(\omega):=\Phi_{\theta^j\omega},
    \qquad j\in\mbZ .
\]
The use of two-sided time is natural here because the base map \(\theta\) is invertible and the pullback construction of the stationary random state uses negative times.

Since the symbol \(\rho\) is already used for density matrices and for the stationary random state \(\rho_\omega\), we denote the maximal-correlation coefficient by \(\varrho_{\rm mc}\).  
For real-valued square-integrable random variables \(U,V\), write
\[
    \operatorname{Cov}(U,V)
    :=
    \mathbb E[UV]-\mathbb E[U]\mathbb E[V].
\]
For sub-\(\sigma\)-algebras \(\mcA,\mcB\subseteq\mcF\), define
\begin{equation}
\label{eq:maximal-correlation-sigma-fields}
    \varrho_{\rm mc}(\mcA,\mcB)
    :=
    \sup
    \left\{
       |\mathbb E[UV]|:
       \begin{array}{l}
       U\in L^2(\mcA),\ V\in L^2(\mcB),\\
       \mathbb E U=\mathbb E V=0,\quad
       \|U\|_{L^2}=\|V\|_{L^2}=1
       \end{array}
    \right\},
\end{equation}
with the convention that the supremum is \(0\) if the class is empty.  
We also use the standard strong-mixing coefficients
\begin{align}
\label{eq:psi-mixing-sigma-fields}
    \psi_{\rm mix}(\mcA,\mcB)
    &:=
    \sup
    \left\{
        \left|
        \frac{\pr(A\cap B)}{\pr(A)\pr(B)}-1
        \right|:
        A\in\mcA,\ B\in\mcB,\ \pr(A)\pr(B)>0
    \right\},\\
\label{eq:phi-mixing-sigma-fields}
    \varphi_{\rm mix}(\mcA,\mcB)
    &:=
    \sup
    \left\{
        |\pr(B\mid A)-\pr(B)|:
        A\in\mcA,\ B\in\mcB,\ \pr(A)>0
    \right\}.
\end{align}
Thus \(0\le\varrho_{\rm mc}\le1\), \(0\le\psi_{\rm mix}\le\infty\), and \(0\le\varphi_{\rm mix}\le1\).  
Vanishing of any one of these coefficients is equivalent to independence of \(\mcA\) and \(\mcB\). 
We use the standard background on these coefficients from \cite{bradley2005basic,bradley2007introduction}.

The estimates needed below are the following classical covariance inequalities.
For sub-\(\sigma\)-algebras \(\mcA,\mcB\subseteq\mcF\),
\begin{align}
\label{eq:mixing-cov-rho}
    |\operatorname{Cov}(U,V)|
    &\le
    \varrho_{\rm mc}(\mcA,\mcB)\,
    \|U\|_{L^2}\|V\|_{L^2},
    &&
    U\in L^2(\mcA),\ V\in L^2(\mcB),\\
\label{eq:mixing-cov-psi}
    |\operatorname{Cov}(U,V)|
    &\le
    \psi_{\rm mix}(\mcA,\mcB)\,
    \|U\|_{L^1}\|V\|_{L^1},
    &&
    U\in L^1(\mcA),\ V\in L^1(\mcB),\\
\label{eq:mixing-cov-varphi}
    |\operatorname{Cov}(U,V)|
    &\le
    2\varphi_{\rm mix}(\mcA,\mcB)\,
    \|U\|_{L^1}\|V\|_{L^\infty},
    &&
    U\in L^1(\mcA),\ V\in L^\infty(\mcB).
\end{align}
The first inequality is immediate from the definition of maximal correlation, after centering \(U\) and \(V\).  
The latter two are standard strong-mixing covariance bounds; see \cite[\S1.2, Theorem~3]{doukhan2012mixing} and \cite[Theorem~3.9]{bradley2007introduction}.

We now pass from pairs of \(\sigma\)-algebras to the channel process.  
For integers \(a<b\), set
\[
    \mcG_a^b:=\sigma(X_j:a\le j<b),
\]
and define the past and future \(\sigma\)-algebras
\[
    \mcG_{-\infty}^k:=\sigma(X_j:j<k),
    \qquad
    \mcG_k^\infty:=\sigma(X_j:j\ge k).
\]
The half-open convention matches the product convention
\[
    \Phi_{\omega;b:a}
    =
    \Phi_{\theta^{b-1}\omega}
    \circ\cdots\circ
    \Phi_{\theta^a\omega},
    \qquad a<b,
\]
because \(\Phi_{\omega;b:a}\) is measurable with respect to \(\mcG_a^b\).

\begin{dfn}
\label{dfn:channel-environment-mixing-coefficients}
    For \(m\in\mbN\), define
    \begin{align}
    \label{eq:varrho-mixing-coefficient}
        \varrho_m
        &:=
        \sup_{k\in\mbZ}
        \varrho_{\rm mc}
        \left(
            \mcG_{-\infty}^{k},
            \mcG_{k+m}^{\infty}
        \right),\\
    \label{eq:psi-channel-mixing-coefficient}
        \psi_m
        &:=
        \sup_{k\in\mbZ}
        \psi_{\rm mix}
        \left(
            \mcG_{-\infty}^{k},
            \mcG_{k+m}^{\infty}
        \right),\\
    \label{eq:varphi-channel-mixing-coefficient}
        \varphi_m
        &:=
        \sup_{k\in\mbZ}
        \varphi_{\rm mix}
        \left(
            \mcG_{-\infty}^{k},
            \mcG_{k+m}^{\infty}
        \right).
    \end{align}
    We say that the channel environment is \(\varrho\)-mixing if
    \(\varrho_m\to0\), \(\psi\)-mixing if \(\psi_m\to0\), and
    \(\varphi\)-mixing if \(\varphi_m\to0\).
\end{dfn}

The process \((X_j)_{j\in\mbZ}\) is strictly stationary, since \(X_j=X_0\circ\theta^j\) and \(\theta\) preserves \(\pr\).  
The coefficients in \Cref{dfn:channel-environment-mixing-coefficients} measure only the randomness
seen by the channel process; they need not detect unrelated invariant factors of the ambient probability space.  
The sequences \((\varrho_m)\), \((\psi_m)\), and \((\varphi_m)\) are nonincreasing in \(m\).
If the random variables \((X_j)_{j\in\mbZ}\) are jointly independent, then
\[
    \varrho_m=\psi_m=\varphi_m=0,
    \qquad m\ge1.
\]

The stronger coefficients imply the maximal-correlation condition used in the proofs.

\begin{prop}
\label{prop:mixing-coefficient-comparisons}
    For any sub-\(\sigma\)-algebras \(\mcA,\mcB\subseteq\mcF\),
    \[
        \varrho_{\rm mc}(\mcA,\mcB)
        \le
        \psi_{\rm mix}(\mcA,\mcB),
    \]
    and
    \[
        \varrho_{\rm mc}(\mcA,\mcB)
        \le
        2
        \varphi_{\rm mix}(\mcA,\mcB)^{1/2}
        \varphi_{\rm mix}(\mcB,\mcA)^{1/2}
        \le
        2\varphi_{\rm mix}(\mcA,\mcB)^{1/2}.
    \]
\end{prop}

\begin{proof}
    If \(\psi_{\rm mix}(\mcA,\mcB)=\infty\), the first estimate is trivial.
    Otherwise, take centered real random variables \(U\in L^2(\mcA)\), \(V\in L^2(\mcB)\) with
    \(\|U\|_{L^2}=\|V\|_{L^2}=1\). 
    By \eqref{eq:mixing-cov-psi},
    \[
        |\mathbb E[UV]|
        =
        |\operatorname{Cov}(U,V)|
        \le
        \psi_{\rm mix}(\mcA,\mcB)\,
        \|U\|_{L^1}\|V\|_{L^1}
        \le
        \psi_{\rm mix}(\mcA,\mcB).
    \]
    Taking the supremum over such \(U,V\) gives  \(\varrho_{\rm mc}(\mcA,\mcB)\le\psi_{\rm mix}(\mcA,\mcB)\).
    The comparison with \(\varphi_{\rm mix}\) is the classical maximal-correlation bound of Doob; see \cite[Lemma~7.1]{doob1942stochastic} and also
    \cite{bradley2005basic,peligrad83noteon}.
\end{proof}

Applying \Cref{prop:mixing-coefficient-comparisons} with
\[
    \mcA=\mcG_{-\infty}^{k},
    \qquad
    \mcB=\mcG_{k+m}^{\infty},
\]
and then taking the supremum over \(k\in\mbZ\), gives
\[
    \varrho_m\le \psi_m,
    \qquad
    \varrho_m\le 2\varphi_m^{1/2},
    \qquad m\ge1.
\]
Consequently, either \(\psi_m\to0\) or \(\varphi_m\to0\) implies \(\varrho_m\to0\).
We finally record the separation property that will be used repeatedly.  If
\(r,s,t\in\mbN\), then
\[
    \kappa_{\omega;r:0}
    =
    \kappa_{\rm tr}\!\left(\Phi_{\omega;r:0}\right)
\]
is \(\mcG_{-\infty}^{r}\)-measurable, while
\[
    \kappa_{\theta^{r+s}\omega;t:0}
    =
    \kappa_{\rm tr}\!\left(\Phi_{\omega;r+s+t:r+s}\right)
\]
is \(\mcG_{r+s}^{\infty}\)-measurable.  Therefore
\eqref{eq:mixing-cov-rho} gives
\[
\begin{aligned}
    \left|
    \operatorname{Cov}
    \left(
        \kappa_{\omega;r:0},
        \kappa_{\theta^{r+s}\omega;t:0}
    \right)
    \right|  
    &\qquad\le
    \varrho_s
    \left\|\kappa_{\omega;r:0}\right\|_{L^2}
    \left\|\kappa_{\theta^{r+s}\omega;t:0}\right\|_{L^2}.
\end{aligned}
\]
This covariance estimate is the probabilistic input behind the annealed contraction recursion.

\begin{lemma}
\label{lemma:kappa-annealed-recursion}
    Set
    \[
        a_n:=\mathbb E\!\left[\kappa_{\omega;n:0}\right],
        \qquad n\ge1.
    \]
    Then, for all \(r,s,t\in\mbN\),
    \[
       a_{r+s+t}
       \le
       a_r a_t+\varrho_s(a_r a_t)^{1/2}.
    \]
\end{lemma}

\begin{proof}
    We first separate the product into two distant blocks and one middle block.
    By the cocycle identity,
    \[
        \Phi_{\omega;r+s+t:0}
        =
        \Phi_{\omega;r+s+t:r+s}
        \circ
        \Phi_{\omega;r+s:r}
        \circ
        \Phi_{\omega;r:0}.
    \]
    Using the submultiplicativity of \(\kappa_{\rm tr}\) from
    \Cref{prop:submult}, and the bound \(\kappa_{\rm tr}(T)\le1\) for every
    CPTP map \(T\), we get
    \[
    \begin{aligned}
        \kappa_{\omega;r+s+t:0}
        &\le
        \kappa_{\omega;r+s+t:r+s}\,
        \kappa_{\omega;r+s:r}\,
        \kappa_{\omega;r:0}        \\
        &\le
        \kappa_{\omega;r+s+t:r+s}\,
        \kappa_{\omega;r:0}.
    \end{aligned}
    \]
    Since
    \[
        \Phi_{\omega;r+s+t:r+s}
        =
        \Phi_{\theta^{r+s}\omega;t:0},
    \]
    this becomes
    \[
        \kappa_{\omega;r+s+t:0}
        \le
        \kappa_{\omega;r:0}\,
        \kappa_{\theta^{r+s}\omega;t:0}.
    \]
    Taking expectations gives
    \begin{equation}
    \label{eq:kappa-annealed-product-start}
        a_{r+s+t}
        \le
        \mathbb E\!\left[
            \kappa_{\omega;r:0}\,
            \kappa_{\theta^{r+s}\omega;t:0}
        \right].
    \end{equation}

    Now set
    \[
        U(\omega):=\kappa_{\omega;r:0},
        \qquad
        V(\omega):=\kappa_{\theta^{r+s}\omega;t:0}.
    \]
    The variable \(U\) is measurable with respect to
    \(\mcG_0^r\), hence with respect to \(\mcG_{-\infty}^{r}\).
    The variable \(V\) is measurable with respect to
    \(\mcG_{r+s}^{r+s+t}\), hence with respect to
    \(\mcG_{r+s}^{\infty}\).
    Therefore, by the maximal-correlation covariance bound
    \eqref{eq:mixing-cov-rho} and by the definition of \(\varrho_s\),
    \[
        |\operatorname{Cov}(U,V)|
        \le
        \varrho_s\,
        \|U\|_{L^2}\,
        \|V\|_{L^2}.
    \]
    Hence
    \[
        \mathbb E[UV]
        \le
        \mathbb E[U]\mathbb E[V]
        +
        \varrho_s
        \left(\mathbb E[U^2]\right)^{1/2}
        \left(\mathbb E[V^2]\right)^{1/2}.
    \]

    By stationarity of the channel sequence,
    \[
        \mathbb E[V]
        =
        \mathbb E\!\left[\kappa_{\theta^{r+s}\omega;t:0}\right]
        =
        \mathbb E\!\left[\kappa_{\omega;t:0}\right]
        =
        a_t.
    \]
    Also \(\mathbb E[U]=a_r\). Since \(0\le \kappa_{\rm tr}(T)\le1\) for every
    CPTP map \(T\), we have \(U^2\le U\) and \(V^2\le V\). Therefore
    \[
        \left(\mathbb E[U^2]\right)^{1/2}
        \le
        a_r^{1/2},
        \qquad
        \left(\mathbb E[V^2]\right)^{1/2}
        \le
        a_t^{1/2}.
    \]
    Combining these estimates with \eqref{eq:kappa-annealed-product-start},
    we obtain
    \[
        a_{r+s+t}
        \le
        a_r a_t
        +
        \varrho_s(a_r a_t)^{1/2}.
    \]
\end{proof}
\begin{lemma}
\label{lemma:kappa-annealed-superpoly}
    Assume that
    \[
        \lambda_{\rm tr}(\omega)<0
        \qquad
        \text{for }\pr\text{-a.e. }\omega,
    \]
    and that the channel environment is \(\varrho\)-mixing, i.e.
    \(\varrho_m\to0\).  Then, for every \(p\in\mbN\), there exists
    \(A_p<\infty\) such that
    \[
       \mathbb E\!\left[\kappa_{\omega;n:0}\right]
       \le
       A_p n^{-p},
       \qquad n\ge1.
    \]
\end{lemma}

\begin{proof}
    Put
    \[
        a_n:=\mathbb E\!\left[\kappa_{\omega;n:0}\right],
        \qquad n\ge1.
    \]
    We first record two elementary facts about \((a_n)\).

    The sequence \((a_n)\) is nonincreasing.  
    Indeed, if \(1\le m<n\), then the cocycle identity gives
    \[
        \Phi_{\omega;n:0}
        =
        \Phi_{\omega;n:m}\circ \Phi_{\omega;m:0}.
    \]
    By submultiplicativity of \(\kappa_{\rm tr}\),
    \[
        \kappa_{\omega;n:0}
        \le
        \kappa_{\omega;n:m}\kappa_{\omega;m:0}
        \le
        \kappa_{\omega;m:0},
    \]
    because every CPTP map has trace-Dobrushin coefficient at most \(1\).
    Taking expectations gives \(a_n\le a_m\).

    Next, \(a_n\to0\).  Since
    \[
        \lambda_{\rm tr}(\omega)
        =
        \lim_{n\to\infty}
        \frac1n\log \kappa_{\omega;n:0}
        <0
    \]
    for \(\pr\)-a.e. \(\omega\), we have
    \[
        \kappa_{\omega;n:0}\longrightarrow0
        \qquad
        \text{for }\pr\text{-a.e. }\omega.
    \]
    Since \(0\le \kappa_{\omega;n:0}\le1\), dominated convergence gives
    \(a_n\to0\).

    Fix \(p\in\mbN\), and set
    \[
        q:=4^{-p}.
    \]
    Since \(\varrho_m\to0\) and \(a_n\to0\), we may choose
    \(M,N_0\in\mbN\), with \(M\le N_0\), such that
    \[
        \varrho_M+a_{N_0}\le q.
    \]
    Define
    \[
        N_i:=2^iN_0+(2^i-1)M,
        \qquad i\in\mbN_0.
    \]
    Then \(N_{i+1}=2N_i+M\).  Applying
    \Cref{lemma:kappa-annealed-recursion} with \(r=t=N_i\) and \(s=M\),
    we obtain
    \[
        a_{N_{i+1}}
        \le
        a_{N_i}^2+\varrho_M a_{N_i}
        =
        a_{N_i}\bigl(a_{N_i}+\varrho_M\bigr).
    \]
    Since \(N_i\ge N_0\) and \((a_n)\) is nonincreasing,
    \[
        a_{N_i}\le a_{N_0}.
    \]
    Hence
    \[
        a_{N_i}+\varrho_M
        \le
        a_{N_0}+\varrho_M
        \le q,
    \]
    and therefore
    \[
        a_{N_{i+1}}\le q a_{N_i}.
    \]
    Iterating gives
    \[
        a_{N_i}
        \le
        q^i a_{N_0}
        =
        4^{-pi}a_{N_0},
        \qquad i\in\mbN_0.
    \]

    Now let \(n\ge N_0\), and choose \(i\in\mbN_0\) such that
    \[
        N_i\le n<N_{i+1}.
    \]
    By monotonicity,
    \[
        a_n\le a_{N_i}\le 4^{-pi}a_{N_0}.
    \]
    Since \(M\le N_0\),
    \[
        N_{i+1}
        =
        2^{i+1}N_0+(2^{i+1}-1)M
        \le
        2^{i+2}N_0.
    \]
    For \(i\ge2\), this implies
    \[
        N_{i+1}\le 4^iN_0.
    \]
    Thus, whenever the chosen index satisfies \(i\ge2\),
    \[
        n<N_{i+1}\le4^iN_0,
    \]
    and hence
    \[
        4^{-pi}\le \left(\frac{N_0}{n}\right)^p.
    \]
    Therefore, for all \(n\ge N_2\),
    \[
        a_n
        \le
        a_{N_0}N_0^p n^{-p}.
    \]

    Finally, enlarge the constant to cover the finitely many values
    \(1\le n<N_2\).  For example, take
    \[
        A_p
        :=
        \max\left\{
            a_{N_0}N_0^p,\,
            \max_{1\le n<N_2} n^p a_n
        \right\}.
    \]
    Then
    \[
        a_n
        =
        \mathbb E\!\left[\kappa_{\omega;n:0}\right]
        \le
        A_p n^{-p}
        \qquad
        \text{for all }n\ge1.
    \]
\end{proof}

\begin{lemma}
\label{lemma:kappa-annealed-independent}
    Assume that
    \[
        \lambda_{\rm tr}(\omega)<0
        \qquad
        \text{for }\pr\text{-a.e. }\omega,
    \]
    and that the random variables \((X_j)_{j\in\mbZ}\) are jointly independent.
    Then there exist constants \(A<\infty\) and \(\eta>0\) such that
    \[
       \mathbb E\!\left[\kappa_{\omega;n:0}\right]
       \le
       Ae^{-\eta n},
       \qquad n\ge1.
    \]
\end{lemma}

\begin{proof}
    Put
    \[
        a_n:=\mathbb E\!\left[\kappa_{\omega;n:0}\right],
        \qquad n\ge1.
    \]
    Since \(\lambda_{\rm tr}(\omega)<0\) for \(\pr\)-a.e. \(\omega\), we have
    \[
        \kappa_{\omega;n:0}\longrightarrow0
        \qquad
        \text{for }\pr\text{-a.e. }\omega.
    \]
    As \(0\le \kappa_{\omega;n:0}\le1\), dominated convergence gives
    \[
        a_n\longrightarrow0.
    \]
    Choose \(N_0\in\mbN\) such that
    \[
        a_{N_0}\le \frac12.
    \]

    We first estimate the product at multiples of \(N_0\).  For \(q\in\mbN\), the cocycle identity gives
    \[
        \Phi_{\omega;qN_0:0}
        =
        \Phi_{\omega;qN_0:(q-1)N_0}
        \circ
        \cdots
        \circ
        \Phi_{\omega;2N_0:N_0}
        \circ
        \Phi_{\omega;N_0:0}.
    \]
    Hence, by submultiplicativity of \(\kappa_{\rm tr}\),
    \[
        \kappa_{\omega;qN_0:0}
        \le
        \prod_{j=0}^{q-1}
        \kappa_{\omega;(j+1)N_0:jN_0}.
    \]
    Equivalently,
    \[
        \kappa_{\omega;(j+1)N_0:jN_0}
        =
        \kappa_{\theta^{jN_0}\omega;N_0:0}.
    \]
    Each factor is a function of the block
    \[
        X_{jN_0},X_{jN_0+1},\ldots,X_{(j+1)N_0-1}.
    \]
    These blocks are disjoint, and the sequence \((X_j)_{j\in\mbZ}\) is jointly independent.  
    Therefore, the random variables
    \[
        \kappa_{\theta^{jN_0}\omega;N_0:0},
        \qquad j=0,\ldots,q-1,
    \]
    are independent.  
    By stationarity, they all have the same law as \(\kappa_{\omega;N_0:0}\).  
    Thus
    \[
    \begin{aligned}
        a_{qN_0}
        &=
        \mathbb E\!\left[\kappa_{\omega;qN_0:0}\right]  
        &\le
        \mathbb E\!\left[
        \prod_{j=0}^{q-1}
        \kappa_{\theta^{jN_0}\omega;N_0:0}
        \right] 
        &=
        \prod_{j=0}^{q-1}
        \mathbb E\!\left[
        \kappa_{\theta^{jN_0}\omega;N_0:0}
        \right] 
        &=
        a_{N_0}^q
        \le
        2^{-q}.
    \end{aligned}
    \]

    We next pass from multiples of \(N_0\) to arbitrary times.  
    The sequence \((a_n)\) is nonincreasing.  
    Indeed, if \(m<n\), then
    \[
        \Phi_{\omega;n:0}
        =
        \Phi_{\omega;n:m}\circ\Phi_{\omega;m:0},
    \]
    so
    \[
        \kappa_{\omega;n:0}
        \le
        \kappa_{\omega;n:m}\kappa_{\omega;m:0}
        \le
        \kappa_{\omega;m:0},
    \]
    because \(\kappa_{\rm tr}(T)\le1\) for every CPTP map \(T\). 
    Taking expectations gives \(a_n\le a_m\).

    Let \(n\ge N_0\), and set
    \[
        q:=\left\lfloor \frac{n}{N_0}\right\rfloor .
    \]
    Then \(q\ge1\) and \(qN_0\le n\).  
    Therefore, by monotonicity,
    \[
        a_n
        \le
        a_{qN_0}
        \le
        2^{-q}.
    \]
    Since
    \[
        q\ge \frac{n}{N_0}-1,
    \]
    we get
    \[
        2^{-q}
        \le
        2^{1-n/N_0}
        =
        2\exp\!\left(-\frac{\log 2}{N_0}n\right).
    \]
    Hence the estimate holds for all \(n\ge N_0\) with
    \[
        \eta:=\frac{\log 2}{N_0}.
    \]

    Finally, since \(0\le a_n\le1\), the same bound holds for the finitely many values \(1\le n<N_0\) after increasing the prefactor.  
    Thus there exist \(A<\infty\) and \(\eta>0\) such that
    \[
        a_n
        =
        \mathbb E\!\left[\kappa_{\omega;n:0}\right]
        \le
        Ae^{-\eta n},
        \qquad n\ge1.
    \]
\end{proof}

We are now ready to prove \Cref{thm:random-annealed-deviations}.

\randomannealeddeviations*

\begin{proof}
    Set
    \[
       K_n^+(\omega):=\kappa_{\omega;n:0},
       \qquad
       K_n^-(\omega):=\kappa_{\omega;0:-n}.
    \]
    By \Cref{prop:memory-loss-negative-exponent-equivalence}, the assumption
    \(\lambda_{\rm tr}<0\) a.s. implies
    \[
       \kappa_{\omega;n:0}\to0
       \qquad\text{and}\qquad
       \kappa_{\omega;0:-n}\to0
    \]
    for \(\pr\)-a.e. \(\omega\). Thus \Cref{cor:random-two-sided-replacement}
    applies and gives the unique stationary random state \(\rho\).
    
    We first record the pointwise deterministic comparison with the Dobrushin
    coefficients. By \Cref{prop:random-forward-replacement},
    \[
       \norm{
          \Phi_{\omega;n:0}-R_{\theta^n\omega}
       }_{1\to1}
       \le
       4\kappa_{\omega;n:0}
       =
       4K_n^+(\omega)
    \]
    for \(\pr\)-a.e. \(\omega\) and every \(n\ge1\). By
    \Cref{prop:random-pullback-state},
    \[
       \norm{
          \Phi_{\omega;0:-n}-R_{\omega}
       }_{1\to1}
       \le
       4\kappa_{\omega;0:-n}
       =
       4K_n^-(\omega)
    \]
    for \(\pr\)-a.e. \(\omega\) and every \(n\ge1\). Hence
    \[
       \Delta_n^{\rm op}(\omega)
       \le
       4\max\{K_n^+(\omega),K_n^-(\omega)\}
       \le
       4\bigl(K_n^+(\omega)+K_n^-(\omega)\bigr).
    \]
    Taking expectations gives
    \[
       \mathbb E[\Delta_n^{\rm op}]
       \le
       4\mathbb E[K_n^+]+4\mathbb E[K_n^-].
    \]
    But
    \[
       K_n^-(\omega)
       =
       \kappa_{\omega;0:-n}
       =
       \kappa_n(\theta^{-n}\omega),
    \]
    and \(\theta\) is probability preserving. Therefore
    \[
       \mathbb E[K_n^-]
       =
       \mathbb E[\kappa_n\circ\theta^{-n}]
       =
       \mathbb E[\kappa_n]
       =
       \mathbb E[K_n^+].
    \]
    Consequently,
    \[
       \mathbb E[\Delta_n^{\rm op}]
       \le
       8\,\mathbb E[\kappa_n].
    \]
    
    Now apply the annealed Dobrushin-coefficient estimate already proved for
    \(\kappa_n=\kappa_{\omega;n:0}\). If \(\varrho_m\to0\), then for every
    \(p\in\mathbb N\) there is \(A_p<\infty\) such that
    \[
       \mathbb E[\kappa_n]\le A_p n^{-p}.
    \]
    Thus
    \[
       \mathbb E[\Delta_n^{\rm op}]
       \le
       8A_p n^{-p}.
    \]
    Renaming \(8A_p\) as \(C_p\) proves the first claim.
    
    If the variables \((X_j)_{j\in\mathbb Z}\) are jointly independent, the same
    annealed coefficient estimate gives constants \(A<\infty\) and
    \(\gamma>0\) such that
    \[
       \mathbb E[\kappa_n]\le A e^{-\gamma n}.
    \]
    Therefore
    \[
       \mathbb E[\Delta_n^{\rm op}]
       \le
       8A e^{-\gamma n}.
    \]
    Renaming \(8A\) as \(C\) proves the exponential claim.
    
    It remains only to explain the state-level statement. For every state
    \(\sigma\in\mcS_d\), \(\tr\sigma=1\), and therefore
    \[
       R_{\theta^n\omega}(\sigma)
       =
       \tr(\sigma)\rho_{\theta^n\omega}
       =
       \rho_{\theta^n\omega}.
    \]
    Hence
    \[
       \norm{\Phi_{\omega;n:0}(\sigma)-\rho_{\theta^n\omega}}_1
       =
       \norm{
          \left(\Phi_{\omega;n:0}-R_{\theta^n\omega}\right)(\sigma)
       }_1
       \le
       \norm{\Phi_{\omega;n:0}-R_{\theta^n\omega}}_{1\to1}.
    \]
    Taking the supremum over \(\sigma\in\mcS_d\) gives the forward state-level
    bound. The pullback bound is identical:
    \[
       \sup_{\sigma\in\mcS_d}
       \norm{\Phi_{\omega;0:-n}(\sigma)-\rho_\omega}_1
       \le
       \norm{\Phi_{\omega;0:-n}-R_\omega}_{1\to1}.
    \]
    Thus
    \[
       \Delta_n^{\rm st}(\omega)\le \Delta_n^{\rm op}(\omega),
    \]
    and the already proved operator estimates imply the state estimates.
\end{proof}

%%%%%%%%%%%%%%%%%%%%%%%%%%%%%%%%%%%%%%%%%%%%%%%%%%%%%%%%%%%%%%%%%%%%%%%%%%%%
%%%%%%%%%%%%%%%%%%%%%%%%%%%%%%%%%%%%%%%%%%%%%%%%%%%%%%%%%%%%%%%%%%%%%%%%%%%%

\section{Inhomogeneous MPS}

%%%%%%%%%%%%%%%%%%%%%%%%%%%%%%%%%%%%%%%%%%%%%%%%%%%%%%%%%%%%%%%%%%%%%%%%%%%%
%%%%%%%%%%%%%%%%%%%%%%%%%%%%%%%%%%%%%%%%%%%%%%%%%%%%%%%%%%%%%%%%%%%%%%%%%%%%

\subsection{Deterministic MPS from Inhomogeneous CPTP Transfer Products}
\label{subsec:det-cptp-mps}

%%%%%%%%%%%%%%%%%%%%%%%%%%%%%%%%%%%%%%%%%%%%%%%%%%%%%%%%%%%%%%%%%%%%%%%%%%%%
%%%%%%%%%%%%%%%%%%%%%%%%%%%%%%%%%%%%%%%%%%%%%%%%%%%%%%%%%%%%%%%%%%%%%%%%%%%%

The notation in this subsection is the deterministic MPS notation introduced in \Cref{sec:intro-mps}.
Thus \(\mcK\) is the physical single-site Hilbert space, \(\mcH\) is the auxiliary or bond Hilbert space, \(K_i^{[n]}\) are left-canonical site-dependent tensors, \(\Phi_n\) are the associated CPTP auxiliary channels, and \(\Theta_{m,n}\) are the right-tail transfer products.
We also use the block products \(K_{\mathbf i}^{[a,b]}\), the inserted maps \(\widehat X_{[a,b]}\), the trace-closed vectors \(\ket{\Psi_n}\), and the trace-closed finite-volume states from \Cref{sec:intro-mps}.
We write
\[
	\mcS\left(\mcH\right)
	:=
	\left\{
		\rho\in B\left(\mcH\right):
		\rho\ge0,\,
		\tr{\rho}=1
	\right\},
	\qquad
	D_{\mcH}:=\dim\mcH.
\]

For completeness, we recall the two transfer identities used below.
For every \(m\ge0\), set \(\Theta_{m,m}:=\operatorname{id}_{B(\mcH)}\).
For \(0\le m\le n\), we use \(\Theta_{m,m}:=\operatorname{id}_{B(\mcH)}\), and for \(m<n\),
\[
	\Theta_{m,n}
	=
	\Phi_{m+1}\circ\Phi_{m+2}\circ\cdots\circ\Phi_n.
\]
Moreover, for \(m<k<n\),
\[
	\Theta_{m,n}
	=
	\Theta_{m,k}\circ\Theta_{k,n}.
\]

Define a reversed-time sequence by
\[
	\widetilde\Phi_{-r}:=\Phi_r,
	\qquad
	r\ge1,
\]
and, if needed, set \(\widetilde\Phi_t:=\operatorname{id}_{B(\mcH)}\) for \(t\ge0\).
With the chronological convention of \Cref{sec:det}, one has, for \(0\le m<n\),
\[
	\Theta_{m,n}
	=
	\widetilde\Phi_{-m:-n}.
\]
Thus the deterministic estimates of \Cref{sec:det} apply to the MPS right-tail products after reversing the time index.
In particular, the hypothesis \(\kappa_{\rm tr}(\Theta_{q,n})\to0\) for fixed \(q\) is the MPS form of pullback trace-memory loss.

\begin{lemma}
\label{lem:mps-inserted-transfer-estimates}
	Let \(1\le a\le b\), and let \(X\in B\left(\mcK_{[a,b]}\right)\).
	Then
	\[
		\norm{
			\widehat X_{[a,b]}
		}_{1\to1}
		\le
		\norm{X}_\infty.
	\]
	Moreover, for every \(Y\in B\left(\mcH\right)\),
	\[
		\left|
			\tr{
				\widehat X_{[a,b]}(Y)
			}
		\right|
		\le
		\norm{X}_\infty\norm{Y}_1.
	\]
\end{lemma}

\begin{proof}
	Define \(V_{[a,b]}:\mcH\to\mcK_{[a,b]}\otimes\mcH\) by
	\[
		V_{[a,b]}\xi
		:=
		\sum_{\mathbf i}
		\ket{\mathbf i}\otimes K_{\mathbf i}^{[a,b]}\xi,
		\qquad
		\xi\in\mcH.
	\]
	The left-canonical condition gives \(V_{[a,b]}^*V_{[a,b]}=I_{\mcH}\).
	Now,
    \[
    	V_{[a,b]}^*V_{[a,b]}
    	=
    	\sum_{\mathbf i}
    	\left(
    		K_{\mathbf i}^{[a,b]}
    	\right)^*
    	K_{\mathbf i}^{[a,b]}
    	=
    	I_{\mcH},
    \]
    where the last equality follows by summing successively over the indices \(i_a,\ldots,i_b\) and using the left-canonical identities.
    Thus \(V_{[a,b]}\) is an isometry.
	With respect to the trace pairing, the dual map of \(\widehat X_{[a,b]}\) is
	\[
		\widehat X_{[a,b]}^\sharp(A)
		=
		V_{[a,b]}^*
		\left(
			X\otimes A
		\right)
		V_{[a,b]},
		\qquad
		A\in B\left(\mcH\right).
	\]
	Indeed,
	\[
		\tr{
			A\widehat X_{[a,b]}(Y)
		}
		=
		\tr{
			\widehat X_{[a,b]}^\sharp(A)Y
		},
		\qquad
		A,Y\in B\left(\mcH\right).
	\]
	Since \(V_{[a,b]}\) is an isometry,
	\[
		\norm{
			\widehat X_{[a,b]}^\sharp(A)
		}_\infty
		\le
		\norm{X}_\infty\norm{A}_\infty.
	\]
	By trace duality, this proves \(\norm{\widehat X_{[a,b]}}_{1\to1}\le\norm{X}_\infty\).

	Taking \(A=I_{\mcH}\), set
	\[
		F_X^{[a,b]}
		:=
		V_{[a,b]}^*
		\left(
			X\otimes I_{\mcH}
		\right)
		V_{[a,b]}.
	\]
	Then \(\tr{\widehat X_{[a,b]}(Y)}=\tr{F_X^{[a,b]}Y}\), and \(\norm{F_X^{[a,b]}}_\infty\le\norm{X}_\infty\).
	Hence
	\[
		\left|
			\tr{
				\widehat X_{[a,b]}(Y)
			}
		\right|
		\le
		\norm{X}_\infty\norm{Y}_1.
	\]
\end{proof}

For a linear map \(L:B\left(\mcH\right)\to B\left(\mcH\right)\), define its superoperator trace by
\[
	\operatorname{Tr}_{\rm sup}(L)
	:=
	\sum_{\alpha,\beta=1}^{D_{\mcH}}
	\tr{
		E_{\alpha\beta}^*
		L(E_{\alpha\beta})
	},
	\qquad
	E_{\alpha\beta}:=\ket{e_\alpha}\bra{e_\beta}.
\]
This is the ordinary trace of \(L\) as a linear operator on \(B\left(\mcH\right)\) equipped with the Hilbert--Schmidt inner product.
In particular, it is independent of the chosen orthonormal basis \(\{e_\alpha\}_{\alpha=1}^{D_{\mcH}}\).
We collect a few lemmas that are required for the subsequent analysis. 

\begin{lemma}
\label{lem:mps-superoperator-trace-estimates}
	For every linear map \(L:B\left(\mcH\right)\to B\left(\mcH\right)\),
	\[
		\left|
			\operatorname{Tr}_{\rm sup}(L)
		\right|
		\le
		D_{\mcH}^2\norm{L}_{1\to1}.
	\]
	Consequently, for every \(X\in B\left(\mcK_{[a,b]}\right)\),
	\[
		\left|
			\operatorname{Tr}_{\rm sup}
			\left(
				\widehat X_{[a,b]}\circ L
			\right)
		\right|
		\le
		D_{\mcH}^2
		\norm{X}_\infty
		\norm{L}_{1\to1}.
	\]
	If \(\rho\in\mcS\left(\mcH\right)\) and \(R_\rho(Y):=\tr{Y}\rho\), then
	\[
		\operatorname{Tr}_{\rm sup}
		\left(
			L\circ R_\rho
		\right)
		=
		\tr{L(\rho)}.
	\]
\end{lemma}

\begin{proof}
	The first estimate follows from the definition of \(\operatorname{Tr}_{\rm sup}\), since \(\norm{E_{\alpha\beta}}_1=1\) and \(\norm{E_{\alpha\beta}^*}_\infty=1\).
	The second estimate follows from the first estimate and \Cref{lem:mps-inserted-transfer-estimates}.
	For the last identity, note that \(R_\rho(E_{\alpha\beta})=\delta_{\alpha\beta}\rho\).
	Therefore
	\[
	\begin{aligned}
		\operatorname{Tr}_{\rm sup}
		\left(
			L\circ R_\rho
		\right)
		&=
		\sum_{\alpha,\beta}
		\tr{
			E_{\alpha\beta}^*
			L(R_\rho(E_{\alpha\beta}))
		} \\
		&=
		\sum_{\alpha}
		\tr{
			E_{\alpha\alpha}L(\rho)
		} \\
		&=
		\tr{L(\rho)}.
	\end{aligned}
	\]
\end{proof}

\begin{lemma}
\label{lem:mps-finite-transfer-formula}
	Let \(1\le m\le n\), and let \(X\in B\left(\mcK_{[1,m]}\right)\).
	Then
	\[
		\bra{\Psi_n}
		\left(
			X\otimes I_{\mcK}^{\otimes(n-m)}
		\right)
		\ket{\Psi_n}
		=
		\operatorname{Tr}_{\rm sup}
		\left(
			\widehat X_{[1,m]}\circ\Theta_{m,n}
		\right).
	\]
	In particular,
	\[
		\inner{\Psi_n}{\Psi_n}
		=
		\operatorname{Tr}_{\rm sup}
		\left(
			\Theta_{0,n}
		\right)
		=
		\operatorname{Tr}_{\rm sup}
		\left(
			\widehat I_{[1,m]}\circ\Theta_{m,n}
		\right).
	\]
	Hence, whenever \(\inner{\Psi_n}{\Psi_n}\ne0\),
	\[
		\varphi_n(X)
		=
		\frac{
			\operatorname{Tr}_{\rm sup}
			\left(
				\widehat X_{[1,m]}\circ\Theta_{m,n}
			\right)
		}{
			\operatorname{Tr}_{\rm sup}
			\left(
				\widehat I_{[1,m]}\circ\Theta_{m,n}
			\right)
		}.
	\]
\end{lemma}

\begin{proof}
	We use the elementary identity
	\[
		\operatorname{Tr}_{\rm sup}
		\left(
			Y\mapsto AYB^*
		\right)
		=
		\tr{A}\tr{B^*},
		\qquad
		A,B\in B\left(\mcH\right).
	\]
	This identity follows immediately from the definition of \(\operatorname{Tr}_{\rm sup}\) and the matrix units \(E_{\alpha\beta}\).

	Let \(\mathbf i,\mathbf j\) run over the block \([1,m]\), and let \(\boldsymbol\ell\) run over the tail \([m+1,n]\).
	If \(n=m\), the tail multi-index is empty and \(K_{\boldsymbol\ell}^{[m+1,m]}:=I_{\mcH}\).
	Expanding the finite-volume vector gives
	\[
	\begin{aligned}
		&\bra{\Psi_n}
		\left(
			X\otimes I_{\mcK}^{\otimes(n-m)}
		\right)
		\ket{\Psi_n} \\
		&\qquad =
		\sum_{\mathbf i,\mathbf j}
		\sum_{\boldsymbol\ell}
		\bra{\mathbf i}X\ket{\mathbf j}
		\overline{
			\tr{
				K_{\mathbf i}^{[1,m]}
				K_{\boldsymbol\ell}^{[m+1,n]}
			}
		}
		\tr{
			K_{\mathbf j}^{[1,m]}
			K_{\boldsymbol\ell}^{[m+1,n]}
		}.
	\end{aligned}
	\]
	On the other hand,
	\[
	\begin{aligned}
		\widehat X_{[1,m]}
		\circ
		\Theta_{m,n}
		=
		\sum_{\mathbf i,\mathbf j}
		\sum_{\boldsymbol\ell}
		\bra{\mathbf i}X\ket{\mathbf j}
		\left(
			Y\mapsto
			K_{\mathbf j}^{[1,m]}
			K_{\boldsymbol\ell}^{[m+1,n]}
			Y
			\left(
				K_{\mathbf i}^{[1,m]}
				K_{\boldsymbol\ell}^{[m+1,n]}
			\right)^*
		\right).
	\end{aligned}
	\]
	Applying the elementary superoperator-trace identity to each summand gives exactly the preceding expansion.
	This proves the first formula.

	Taking \(X=I_{\mcK_{[1,m]}}\) gives
	\[
		\inner{\Psi_n}{\Psi_n}
		=
		\operatorname{Tr}_{\rm sup}
		\left(
			\widehat I_{[1,m]}\circ\Theta_{m,n}
		\right).
	\]
	Since \(\widehat I_{[1,m]}\circ\Theta_{m,n}=\Theta_{0,n}\), this also gives
	\[
		\inner{\Psi_n}{\Psi_n}
		=
		\operatorname{Tr}_{\rm sup}
		\left(
			\Theta_{0,n}
		\right).
	\]
	The normalized formula follows by dividing by \(\inner{\Psi_n}{\Psi_n}\).
\end{proof}

Recall that the one-sided quasi-local spin algebra is
\[
	\mcA_{\mbN}
	:=
	\overline{
		\bigcup_{m\ge1}
		B\left(\mcK_{[1,m]}\right)
	}^{\norm{\cdot}_\infty},
\]
where \(B\left(\mcK_{[1,m]}\right)\) is identified with a subalgebra of \(B\left(\mcK_{[1,n]}\right)\), for \(m\le n\), through the embedding \(X\mapsto X\otimes I_{\mcK}^{\otimes(n-m)}\).
This is the standard \(C^*\)-inductive-limit construction of the quasi-local algebra of a one-sided quantum spin chain \cite{BratteliRobinson1987,BratteliRobinson1997,FannesNachtergaeleWerner1992}.
A family of states
\[
	\psi_m:B\left(\mcK_{[1,m]}\right)\to\mbC,
	\qquad
	m\ge1,
\]
is called compatible if
\[
	\psi_{m+1}\left(X\otimes I_{\mcK}\right)
	=
	\psi_m\left(X\right),
	\qquad
	X\in B\left(\mcK_{[1,m]}\right).
\]
By the universal property of the inductive limit, every compatible family of local states defines a unique state on \(\mcA_{\mbN}\).

We are now ready to prove \Cref{thm:intro-mps-thermodynamic-limit}. 

\introMpsThermodynamicLimit*

\begin{proof}
	We first convert the right-tail MPS products into the pullback products of \Cref{sec:det}.
	Define a two-sided sequence of CPTP maps on \(B(\mcH)\) by
	\[
		\widetilde\Phi_{-r}
		:=
		\Phi_r,
		\qquad
		r\ge1,
	\]
	and by
	\[
		\widetilde\Phi_t
		:=
		\operatorname{id}_{B(\mcH)},
		\qquad
		t\ge0.
	\]
	With the chronological convention of \Cref{sec:det}, one has, for \(0\le q<n\),
	\[
		\widetilde\Phi_{-q:-n}
		=
		\Theta_{q,n}.
	\]
	Indeed,
	\[
		\widetilde\Phi_{-q:-n}
		=
		\widetilde\Phi_{-q-1}
		\circ
		\widetilde\Phi_{-q-2}
		\circ\cdots\circ
		\widetilde\Phi_{-n}
		=
		\Phi_{q+1}\circ\Phi_{q+2}\circ\cdots\circ\Phi_n.
	\]

	We check the pullback hypothesis for the reversed sequence.
	If \(t=-q\le0\), then the assumption gives
	\[
		\kappa_{\rm tr}
		\left(
			\widetilde\Phi_{t:-n}
		\right)
		=
		\kappa_{\rm tr}
		\left(
			\Theta_{q,n}
		\right)
		\longrightarrow0.
	\]
	If \(t>0\), then
	\[
		\widetilde\Phi_{t:-n}
		=
		\widetilde\Phi_{t:0}
		\circ
		\widetilde\Phi_{0:-n}
		=
		\Theta_{0,n},
	\]
	because the maps \(\widetilde\Phi_0,\ldots,\widetilde\Phi_{t-1}\) are identities.
	Hence
	\[
		\kappa_{\rm tr}
		\left(
			\widetilde\Phi_{t:-n}
		\right)
		\le
		\kappa_{\rm tr}
		\left(
			\Theta_{0,n}
		\right)
		\longrightarrow0.
	\]
	Thus, the reversed two-sided sequence satisfies pullback trace-memory loss.

	By \Cref{thm:pullback-canonical-replacement}, there is a unique pullback boundary family
	\[
		\widetilde\rho_t\in\mcS(\mcH),
		\qquad
		t\in\mbZ,
	\]
	satisfying
	\[
		\widetilde\rho_{t+1}
		=
		\widetilde\Phi_t(\widetilde\rho_t),
		\qquad
		t\in\mbZ.
	\]
	Define, for \(r\ge1\),
	\[
		\rho_r
		:=
		\widetilde\rho_{-r+1}.
	\]
	Then
	\[
		\rho_r
		=
		\widetilde\rho_{-r+1}
		=
		\widetilde\Phi_{-r}(\widetilde\rho_{-r})
		=
		\Phi_r(\rho_{r+1}),
	\]
	so \((\rho_r)_{r\ge1}\) is a right boundary sequence for the MPS transfer products.

	We next prove the uniqueness of this right boundary sequence.
	Suppose that \((\eta_r)_{r\ge1}\subset\mcS(\mcH)\) is another sequence satisfying
	\[
		\eta_r
		=
		\Phi_r(\eta_{r+1}),
		\qquad
		r\ge1.
	\]
	Extend it to a two-sided family by setting
	\[
		\widetilde\eta_{-r+1}
		:=
		\eta_r,
		\qquad
		r\ge1,
	\]
	and
	\[
		\widetilde\eta_t
		:=
		\eta_1,
		\qquad
		t\ge0.
	\]
	Since \(\widetilde\Phi_t=\operatorname{id}\) for \(t\ge0\), this family satisfies
	\[
		\widetilde\eta_{t+1}
		=
		\widetilde\Phi_t(\widetilde\eta_t),
		\qquad
		t\in\mbZ.
	\]
	Thus, for every \(s<t\),
	\[
		\widetilde\eta_t
		=
		\widetilde\Phi_{t:s}(\widetilde\eta_s).
	\]
	By the quantitative estimate in \Cref{thm:det-replacement-approximation}, applied to the product \(\widetilde\Phi_{t:s}\) with reference state \(\widetilde\eta_s\), the replacement channel \(Y\mapsto\tr{Y}\widetilde\eta_t\) satisfies
	\[
		\norm{
			\widetilde\Phi_{t:s}
			-
			\left(
				Y\mapsto\tr{Y}\widetilde\eta_t
			\right)
		}_{1\to1}
		\le
		4
		\kappa_{\rm tr}
		\left(
			\widetilde\Phi_{t:s}
		\right).
	\]
	Letting \(s\to-\infty\), the right-hand side tends to zero for every fixed \(t\).
	Therefore \((\widetilde\eta_t)_{t\in\mbZ}\) is also a family of pullback replacement centers.
	The uniqueness statement in \Cref{thm:pullback-canonical-replacement} gives
	\[
		\widetilde\eta_t
		=
		\widetilde\rho_t,
		\qquad
		t\in\mbZ.
	\]
	In particular, \(\eta_r=\rho_r\) for every \(r\ge1\).
	This proves the uniqueness of the right boundary sequence.

	We now construct the infinite-volume physical state.
	For \(m\ge1\), define
	\[
		\varphi_\infty^{(m)}(X)
		:=
		\tr{
			\widehat X_{[1,m]}(\rho_{m+1})
		},
		\qquad
		X\in\mcA_{[1,m]}.
	\]
	We first show that \(\varphi_\infty^{(m)}\) is a state.
	Let
	\[
		V_m:\mcH\to\mcK_{[1,m]}\otimes\mcH
	\]
	be the isometry
	\[
		V_m\xi
		:=
		\sum_{\mathbf i}
		\ket{\mathbf i}\otimes K_{\mathbf i}^{[1,m]}\xi.
	\]
	For \(X\in\mcA_{[1,m]}\), the scalar functional obtained by tracing the inserted transfer map is
	\[
		\tr{
			\widehat X_{[1,m]}(Y)
		}
		=
		\tr{F_XY},
		\qquad
		Y\in B(\mcH),
	\]
	where
	\[
		F_X
		=
		V_m^*
		\left(
			X\otimes I_{\mcH}
		\right)
		V_m.
	\]
	The map \(X\mapsto F_X\) is positive and unital.
	Since \(\rho_{m+1}\) is a state, \(X\mapsto\tr{F_X\rho_{m+1}}\) is a state on \(\mcA_{[1,m]}\).
	Hence each \(\varphi_\infty^{(m)}\) is a state.

	We next prove compatibility.
	For \(X\in\mcA_{[1,m]}\) and \(Y\in B(\mcH)\), inserting the identity at site \(m+1\) gives
	\[
		\widehat{X\otimes I_{\mcK}}_{[1,m+1]}(Y)
		=
		\widehat X_{[1,m]}
		\left(
			\Phi_{m+1}(Y)
		\right).
	\]
	Using \(\rho_{m+1}=\Phi_{m+1}(\rho_{m+2})\), we obtain
	\[
	\begin{aligned}
		\varphi_\infty^{(m+1)}
		\left(
			X\otimes I_{\mcK}
		\right)
		&=
		\tr{
			\widehat{X\otimes I_{\mcK}}_{[1,m+1]}
			(\rho_{m+2})
		} \\
		&=
		\tr{
			\widehat X_{[1,m]}
			\left(
				\Phi_{m+1}(\rho_{m+2})
			\right)
		} \\
		&=
		\tr{
			\widehat X_{[1,m]}(\rho_{m+1})
		} \\
		&=
		\varphi_\infty^{(m)}(X).
	\end{aligned}
	\]
	Thus, the local states are compatible.

    The compatible family defines a linear functional on the algebraic local algebra
    \[
    	\mcA_{\rm loc}
    	:=
    	\bigcup_{m\ge1}\mcA_{[1,m]}.
    \]
    Indeed, if an observable is represented in two different local algebras, compatibility shows that the two values coincide.
    We denote this functional by \(\varphi_{\rm loc}\).
    For \(X\in\mcA_{[1,m]}\), we have
    \[
    	\left|
    		\varphi_{\rm loc}(X)
    	\right|
    	=
    	\left|
    		\varphi_\infty^{(m)}(X)
    	\right|
    	\le
    	\norm{X}_\infty,
    \]
    because \(\varphi_\infty^{(m)}\) is a state on \(\mcA_{[1,m]}\).
    Also \(\varphi_{\rm loc}(I)=1\), and therefore \(\norm{\varphi_{\rm loc}}=1\).
    Hence, if \((X_j)_j\subset\mcA_{\rm loc}\) is Cauchy in \(\norm{\cdot}_\infty\), then \((\varphi_{\rm loc}(X_j))_j\) is Cauchy in \(\mbC\).
    It follows that \(\varphi_{\rm loc}\) extends uniquely by continuity to a bounded linear functional \(\varphi_\infty\) on
    \[
    	\mcA_{\mbN}
    	=
    	\overline{
    		\mcA_{\rm loc}
    	}^{\norm{\cdot}_\infty}.
    \]
    The extension is positive.
    Indeed, let \(A\in\mcA_{\mbN}\) satisfy \(A\ge0\).
    Choose \(B_j\in\mcA_{\rm loc}\) such that \(B_j\to A^{1/2}\) in norm.
    Then \(B_j^*B_j\in\mcA_{\rm loc}\), \(B_j^*B_j\ge0\), and \(B_j^*B_j\to A\) in norm.
    Therefore
    \[
    	\varphi_\infty(A)
    	=
    	\lim_{j\to\infty}
    	\varphi_{\rm loc}(B_j^*B_j)
    	\ge0.
    \]
    Since \(\varphi_\infty(I)=1\), the functional \(\varphi_\infty\) is a state.
    The displayed formula for \(\varphi_\infty(X)\) holds on every local algebra by construction.
    Uniqueness follows because \(\mcA_{\rm loc}\) is dense in \(\mcA_{\mbN}\).

	It remains to prove the convergence of the normalized trace-closed MPS states.
	Fix \(m\ge1\) and \(X\in\mcA_{[1,m]}\).
	For \(n\ge m\), set
	\[
		N_n(X)
		:=
		\operatorname{Tr}_{\rm sup}
		\left(
			\widehat X_{[1,m]}\circ\Theta_{m,n}
		\right),
	\]
	and
	\[
		Z_n
		:=
		\operatorname{Tr}_{\rm sup}
		\left(
			\widehat I_{[1,m]}\circ\Theta_{m,n}
		\right).
	\]
	By \Cref{lem:mps-finite-transfer-formula}, one has
	\[
		Z_n
		=
		\inner{\Psi_n}{\Psi_n},
	\]
	and, whenever \(Z_n\ne0\),
	\[
		\varphi_n(X)
		=
		\frac{N_n(X)}{Z_n}.
	\]

	Define the replacement channel
	\[
		R_{m+1}(Y)
		:=
		\tr{Y}\rho_{m+1},
		\qquad
		Y\in B(\mcH).
	\]
	We claim that
	\[
		\norm{
			\Theta_{m,n}
			-
			R_{m+1}
		}_{1\to1}
		\le
		4
		\kappa_{\rm tr}
		\left(
			\Theta_{m,n}
		\right).
	\]
	To see this, use the reversed product identity
	\[
		\Theta_{m,n}
		=
		\widetilde\Phi_{-m:-n}.
	\]
	Also,
	\[
		\rho_{m+1}
		=
		\widetilde\rho_{-m}
		=
		\widetilde\Phi_{-m:-n}(\widetilde\rho_{-n})
		=
		\Theta_{m,n}(\rho_{n+1}).
	\]
	Therefore \(R_{m+1}\) is exactly the reference-state replacement channel for the product \(\widetilde\Phi_{-m:-n}\) with reference state \(\widetilde\rho_{-n}=\rho_{n+1}\).
	The quantitative estimate in \Cref{thm:det-replacement-approximation} gives the claimed bound.

	We now compare the trace-closed numerator to its limit.
	For any linear map \(L:B(\mcH)\to B(\mcH)\), the elementary superoperator-trace estimate gives
	\[
		\left|
			\operatorname{Tr}_{\rm sup}(L)
		\right|
		\le
		D_{\mcH}^2
		\norm{L}_{1\to1}.
	\]
	The inserted-transfer estimate gives
	\[
		\norm{
			\widehat X_{[1,m]}
		}_{1\to1}
		\le
		\norm{X}_\infty.
	\]
	Hence
	\[
	\begin{aligned}
		\left|
			N_n(X)
			-
			\operatorname{Tr}_{\rm sup}
			\left(
				\widehat X_{[1,m]}\circ R_{m+1}
			\right)
		\right|
		&\le
		D_{\mcH}^2
		\norm{
			\widehat X_{[1,m]}
			\circ
			\left(
				\Theta_{m,n}
				-
				R_{m+1}
			\right)
		}_{1\to1} \\
		&\le
		4D_{\mcH}^2
		\norm{X}_\infty
		\kappa_{\rm tr}
		\left(
			\Theta_{m,n}
		\right).
	\end{aligned}
	\]

	We identify the limiting superoperator trace.
	If \(R_\rho(Y):=\tr{Y}\rho\), then for every linear map \(L:B(\mcH)\to B(\mcH)\),
	\[
		\operatorname{Tr}_{\rm sup}(L\circ R_\rho)
		=
		\tr{L(\rho)}.
	\]
	Indeed, \(R_\rho(E_{\alpha\beta})=\delta_{\alpha\beta}\rho\), and summing over the matrix units gives the identity.
	Therefore
	\[
		\operatorname{Tr}_{\rm sup}
		\left(
			\widehat X_{[1,m]}\circ R_{m+1}
		\right)
		=
		\tr{
			\widehat X_{[1,m]}(\rho_{m+1})
		}
		=
		\varphi_\infty(X).
	\]
	Consequently,
	\[
		\left|
			N_n(X)
			-
			\varphi_\infty(X)
		\right|
		\le
		4D_{\mcH}^2
		\norm{X}_\infty
		\kappa_{\rm tr}
		\left(
			\Theta_{m,n}
		\right).
	\]

	Applying the same estimate with \(X=I_{\mcK_{[1,m]}}\) gives the normalization estimate.
	Since \(\widehat I_{[1,m]}\) is CPTP and \(\tr{\rho_{m+1}}=1\),
	\[
		\operatorname{Tr}_{\rm sup}
		\left(
			\widehat I_{[1,m]}\circ R_{m+1}
		\right)
		=
		\tr{
			\widehat I_{[1,m]}(\rho_{m+1})
		}
		=
		1.
	\]
	Thus
	\[
		\left|
			Z_n-1
		\right|
		\le
		4D_{\mcH}^2
		\kappa_{\rm tr}
		\left(
			\Theta_{m,n}
		\right).
	\]

	Since \(\kappa_{\rm tr}(\Theta_{m,n})\to0\) for fixed \(m\), there is \(n_0\ge m+1\) such that, for all \(n\ge n_0\),
	\[
		4D_{\mcH}^2
		\kappa_{\rm tr}
		\left(
			\Theta_{m,n}
		\right)
		\le
		\frac12.
	\]
	For such \(n\), the estimate \(|Z_n-1|\le1/2\) implies that \(Z_n\ne0\).
	Hence \(\inner{\Psi_n}{\Psi_n}=Z_n\ne0\), and the normalized finite-volume state \(\varphi_n\) is defined.
	Moreover, \(|Z_n|^{-1}\le2\).

	Finally, using \(|\varphi_\infty(X)|\le\norm{X}_\infty\), we obtain
	\[
	\begin{aligned}
		\left|
			\varphi_n(X)
			-
			\varphi_\infty(X)
		\right|
		&=
		\left|
			\frac{N_n(X)}{Z_n}
			-
			\varphi_\infty(X)
		\right| \\
		&\le
		2
		\left|
			N_n(X)
			-
			\varphi_\infty(X)
		\right|
		+
		2
		\norm{X}_\infty
		\left|
			Z_n-1
		\right| \\
		&\le
		16D_{\mcH}^2
		\norm{X}_\infty
		\kappa_{\rm tr}
		\left(
			\Theta_{m,n}
		\right).
	\end{aligned}
	\]
	The right-hand side tends to zero as \(n\to\infty\).
	Therefore
	\[
		\varphi_n(X)
		\longrightarrow
		\varphi_\infty(X).
	\]
	This proves the theorem.
\end{proof}

	Under the hypotheses of \Cref{thm:intro-mps-thermodynamic-limit}, the proof is quantitative.
	Fix \(m\ge1\), and set
	\[
		K_{m,n}
		:=
		\kappa_{\rm tr}\left(\Theta_{m,n}\right),
		\qquad
		n\ge m.
	\]
	If
	\[
		4D_{\mcH}^2K_{m,n}
		\le
		\frac12,
		\qquad
		D_{\mcH}:=\dim\mcH,
	\]
	then \(\inner{\Psi_n}{\Psi_n}\ne0\).
	Moreover, for every \(X\in\mcA_{[1,m]}\),
	\[
		\left|
			\varphi_n(X)
			-
			\varphi_\infty(X)
		\right|
		\le
		16D_{\mcH}^2
		\norm{X}_\infty
		K_{m,n}.
	\]
	Thus, every upper bound on \(K_{m,n}\) gives the corresponding thermodynamic-limit rate for the trace-closed MPS expectations.
	The deterministic clocks of \Cref{sec:det-rates} apply to \(K_{m,n}\) after reversing the time index.
	Indeed, if \(\widetilde\Phi_{-j}:=\Phi_j\) for \(j\ge1\), then, for \(m<n\),
	\[
		\Theta_{m,n}
		=
		\widetilde\Phi_{-m:-n}.
	\]
	In particular, if \(a_j\in[0,1]\) and
	\[
		\kappa_{\rm tr}\left(\Phi_j\right)
		\le
		1-a_j,
		\qquad
		j=m+1,\ldots,n,
	\]
	then \Cref{prop:det-contraction-clocks} gives
	\[
		K_{m,n}
		\le
		\prod_{j=m+1}^{n}
		\left(1-a_j\right)
		\le
		\exp\left(
			-\sum_{j=m+1}^{n}a_j
		\right).
	\]
	Also, for every \(r\in(0,1)\), if
	\[
		G_r(m,n)
		:=
		\#\left\{
			j\in\left\{m+1,\ldots,n\right\}:
			a_j\ge r
		\right\},
	\]
	then
	\[
		K_{m,n}
		\le
		\left(1-r\right)^{G_r(m,n)}.
	\]

	Likewise, the good-block clocks of \Cref{prop:uniform-good-block} translate directly to right-tail products.
	For example, suppose that there exist integers \(\ell\ge1\), \(M\ge\ell\), and \(q\in(0,1)\) such that every window of \(M\) consecutive auxiliary maps contains a consecutive length-\(\ell\) subblock whose trace-Dobrushin coefficient is at most \(q\).
	That is, for every admissible \(p\), there exists \(u\in\{p,\ldots,p+M-\ell\}\) such that
	\[
		\kappa_{\rm tr}
		\left(
			\Theta_{u,u+\ell}
		\right)
		\le q.
	\]
	Then
	\[
		K_{m,n}
		\le
		q^{
			\left\lfloor
				\frac{n-m}{M}
			\right\rfloor
		}.
	\]
	Hence, no additional MPS-specific rate argument is needed once the right-tail trace-Dobrushin coefficient has been estimated.

    \medskip

We now prove \Cref{thm:intro-mps-correlation-bound}. 

\introMpsCorrelationBound*

\begin{proof}
	We use the infinite-volume state and the right boundary sequence constructed in \Cref{thm:intro-mps-thermodynamic-limit}.
	Write the observables as elements of the left-anchored algebras by setting
	\[
		A^{[1,q]}
		:=
		I_{\mcK}^{\otimes(p-1)}
		\otimes A
		\in\mcA_{[1,q]}.
	\]
	Set
	\[
		\mcT_A
		:=
		\widehat{A^{[1,q]}}_{[1,q]},
		\qquad
		\mcT_B
		:=
		\widehat B_{[r,s]}.
	\]
	Since \(q+1<r\), the gap transfer map is
	\[
		T_{\rm gap}
		:=
		\widehat I_{[q+1,r-1]}
		=
		\Theta_{q,r-1}.
	\]
    By the block factorization of the inserted transfer maps,
    \[
    	\widehat{
    		A^{[1,q]}
    		\otimes
    		I_{\mcK}^{\otimes(r-q-1)}
    		\otimes
    		B
    	}_{[1,s]}
    	=
    	\mcT_A
    	\circ
    	T_{\rm gap}
    	\circ
    	\mcT_B.
    \]
    Similarly,
    \[
    	\widehat{
    		I_{\mcK}^{\otimes(r-1)}
    		\otimes B
    	}_{[1,s]}
    	=
    	\widehat I_{[1,r-1]}
    	\circ
    	\mcT_B.
    \]
	The product \(AB\) denotes the product of the canonically embedded observables on \(\mcA_{\mbN}\).
	Since the supports are disjoint, this is the same as the tensor observable
	\[
		A^{[1,q]}
		\otimes
		I_{\mcK}^{\otimes(r-q-1)}
		\otimes
		B
		\in
		\mcA_{[1,s]}.
	\]

	Define
	\[
		Y_B
		:=
		\mcT_B(\rho_{s+1}),
		\qquad
		z_B
		:=
		\tr{Y_B}.
	\]
	We first identify \(z_B\) with the expectation of \(B\).
	As an observable on \([1,s]\), \(B\) is \(I_{\mcK}^{\otimes(r-1)}\otimes B\).
	Therefore
	\[
		\varphi_\infty(B)
		=
		\tr{
			\widehat I_{[1,r-1]}
			\left(
				\mcT_B(\rho_{s+1})
			\right)
		}.
	\]
	The map \(\widehat I_{[1,r-1]}\) is trace-preserving.
	Hence
	\[
		\varphi_\infty(B)
		=
		\tr{Y_B}
		=
		z_B.
	\]

	Next, the boundary recursion gives
	\[
		\rho_r
		=
		\widehat I_{[r,s]}(\rho_{s+1}).
	\]
	It also gives
	\[
		\rho_{q+1}
		=
		T_{\rm gap}(\rho_r).
	\]
	The block factorization of the inserted transfer maps yields
	\[
		\varphi_\infty(AB)
		=
		\tr{
			\mcT_A
			\left(
				T_{\rm gap}
				\left(
					Y_B
				\right)
			\right)
		}.
	\]
	The same formula with only \(A\) inserted gives
	\[
		\varphi_\infty(A)
		=
		\tr{
			\mcT_A
			\left(
				T_{\rm gap}
				\left(
					\rho_r
				\right)
			\right)
		}.
	\]
	Using \(z_B=\varphi_\infty(B)\), we obtain
	\[
		\varphi_\infty(AB)
		-
		\varphi_\infty(A)\varphi_\infty(B)
		=
		\tr{
			\mcT_A
			\left(
				T_{\rm gap}
				\left(
					Y_B-z_B\rho_r
				\right)
			\right)
		}.
	\]

	Set
	\[
		Z_B
		:=
		Y_B-z_B\rho_r.
	\]
	Then
	\[
		\tr{Z_B}
		=
		0.
	\]
	Moreover,
	\[
		\norm{Y_B}_1
		\le
		\norm{\mcT_B}_{1\to1}
		\norm{\rho_{s+1}}_1
		=
		\norm{\mcT_B}_{1\to1}.
	\]
	Thus
	\[
		|z_B|
		=
		|\tr{Y_B}|
		\le
		\norm{Y_B}_1
		\le
		\norm{\mcT_B}_{1\to1}.
	\]
	Since \(\norm{\rho_r}_1=1\), it follows that
	\[
		\norm{Z_B}_1
		\le
		2\norm{\mcT_B}_{1\to1}.
	\]

	By \Cref{rem:complex-tracezero-control}, applied to the CPTP map \(T_{\rm gap}\), we have
	\[
		\norm{
			T_{\rm gap}(Z_B)
		}_1
		\le
		2
		\kappa_{\rm tr}(T_{\rm gap})
		\norm{Z_B}_1.
	\]
	Therefore
	\[
		\norm{
			T_{\rm gap}(Z_B)
		}_1
		\le
		4
		\norm{\mcT_B}_{1\to1}
		\kappa_{\rm tr}(T_{\rm gap}).
	\]

	We now test this trace-zero error against the left observable \(A\).
	The inserted-transfer compression estimate gives
	\[
		\left|
			\tr{
				\mcT_A(W)
			}
		\right|
		\le
		\norm{A^{[1,q]}}_\infty
		\norm{W}_1
		=
		\norm{A}_\infty
		\norm{W}_1
	\]
	for every \(W\in B(\mcH)\).
	Applying this estimate with \(W=T_{\rm gap}(Z_B)\) gives
	\[
	\begin{aligned}
		\left|
			\varphi_\infty(AB)
			-
			\varphi_\infty(A)\varphi_\infty(B)
		\right|
		&\le
		\norm{A}_\infty
		\norm{
			T_{\rm gap}(Z_B)
		}_1 \\
		&\le
		4
		\norm{A}_\infty
		\norm{\mcT_B}_{1\to1}
		\kappa_{\rm tr}(T_{\rm gap}).
	\end{aligned}
	\]
	Finally, the inserted-transfer norm estimate gives
	\[
		\norm{\mcT_B}_{1\to1}
		=
		\norm{
			\widehat B_{[r,s]}
		}_{1\to1}
		\le
		\norm{B}_\infty.
	\]
	Since \(T_{\rm gap}=\Theta_{q,r-1}\), we conclude that
	\[
		\left|
			\varphi_\infty(AB)
			-
			\varphi_\infty(A)\varphi_\infty(B)
		\right|
		\le
		4
		\norm{A}_\infty
		\norm{B}_\infty
		\kappa_{\rm tr}\left(\Theta_{q,r-1}\right).
	\]
	This proves the stated correlation bound.
\end{proof}

%%%%%%%%%%%%%%%%%%%%%%%%%%%%%%%%%%%%%%%%%%%%%%%%%%%%%%%%%%%%%%%%%%%%%%%%%%%%
%%%%%%%%%%%%%%%%%%%%%%%%%%%%%%%%%%%%%%%%%%%%%%%%%%%%%%%%%%%%%%%%%%%%%%%%%%%%

\subsection{Random MPS from Stationary CPTP Transfer Cocycles}
\label{subsec:random-cptp-mps}

%%%%%%%%%%%%%%%%%%%%%%%%%%%%%%%%%%%%%%%%%%%%%%%%%%%%%%%%%%%%%%%%%%%%%%%%%%%%
%%%%%%%%%%%%%%%%%%%%%%%%%%%%%%%%%%%%%%%%%%%%%%%%%%%%%%%%%%%%%%%%%%%%%%%%%%%%

The notation in this subsection is the random MPS notation introduced in \Cref{sec:intro-random-mps}.
Thus \((\Omega,\mcF,\pr,\theta)\) is an invertible probability-preserving system, \(K_i(\omega)\) is a measurable left-canonical tensor field, and \(\Phi_\omega\) is the associated CPTP transfer cocycle.
The tensor placed at site \(n\) is \(K_i^{[n]}(\omega)=K_i(\theta^{-n}\omega)\), so the spatial direction of the MPS follows the inverse base transformation.
Since \(\theta\) is invertible and \(\pr\)-preserving, the map \(\tau:=\theta^{-1}\) is also \(\pr\)-preserving.
Equivalently, \(K_i^{[n]}(\omega)=K_i(\tau^n\omega)\), so the tensor sequence is stationary over the base \((\Omega,\mcF,\pr,\tau)\).
We keep the random-channel notation of \Cref{sec:random}, which is written over the original base transformation \(\theta\).
With this convention, the MPS right-tail products are pullback products for the \(\theta\)-cocycle.
For \(0\le m<n\), one has
\begin{equation}
\label{eq:random-mps-right-tail-pullback}
    \Theta_{\omega;m,n}
    =
    \Phi_{\theta^{-(m+1)}\omega}
    \circ
    \cdots
    \circ
    \Phi_{\theta^{-n}\omega}
    =
    \Phi_{\theta^{-m}\omega;0:-(n-m)}.
\end{equation}
This identity is the bridge between the random replacement theory of \Cref{thm:negative-exponent-replacement} and the deterministic MPS estimates of \Cref{subsec:det-cptp-mps}.
It is also the reason that the pullback constants \(C_\beta^-\), rather than the forward constants \(C_\beta^+\), appear in the random MPS bounds.
The random boundary state used below is still the \(\theta\)-stationary pullback state, namely \(\Phi_\eta(\rho_\eta)=\rho_{\theta\eta}\).
This stationarity relation gives the deterministic right-boundary sequence \(\rho_r^\omega=\rho_{\theta^{-(r-1)}\omega}\) for each fixed realization \(\omega\).
Thus, the proofs below are obtained by applying the deterministic thermodynamic limit and clustering results fiberwise, with the pullback contraction supplied by \Cref{thm:negative-exponent-replacement}.

\introRandomMpsThermodynamicLimit*

\begin{proof}

    Let \(\rho:\Omega\to\mcS(\mcH)\), \(\beta\), \(C_\beta^-\), and \(\Omega_*\) be given by \Cref{thm:negative-exponent-replacement}.
    We use \(C_\beta^-\) as constructed in the proof of \Cref{thm:negative-exponent-replacement}.
    In that proof, \(C_\beta^-\) is chosen so that there is a full-measure set \(E_{\rm coef}^-\) on which, for every \(\eta\in E_{\rm coef}^-\) and every \(L\ge1\),
    \[
        \kappa_{\rm tr}
        \left(
            \Phi_{\eta;0:-L}
        \right)
        \le
        C_\beta^-(\eta)e^{\beta(\eta)L}.
    \]
    Let \(E_{\rm stat}\) be the full-measure set on which
    \[
        \Phi_\eta(\rho_\eta)=\rho_{\theta\eta}
    \]
    holds.
    Let \(E_{\rm can}\) be the full-measure set on which the left-canonical identities hold.
    Let \(E_{\rm fin}\) be the full-measure set on which \(C_\beta^-(\eta)<\infty\).
    Let \(E_\beta\) be the full-measure set on which \(\beta(\theta\eta)=\beta(\eta)\) and \(\beta(\eta)<0\).
    After replacing \(\Omega_*\) by
    \[
        \Omega_*
        \cap
        \bigcap_{j\in\mbZ}
        \theta^{-j}
        \left(
            E_{\rm coef}^-
            \cap
            E_{\rm stat}
            \cap
            E_{\rm can}
            \cap
            E_{\rm fin}
            \cap
            E_\beta
        \right),
    \]
    we may assume that \(\Omega_*\) is still \(\theta\)-invariant and has full measure, and that all the above properties hold for every \(\eta=\theta^j\omega\), \(j\in\mbZ\), whenever \(\omega\in\Omega_*\).

	Fix \(\omega\in\Omega_*\).
	For this realization, the tensors
	\[
		K_i^{[n]}(\omega)
		=
		K_i(\theta^{-n}\omega)
	\]
	define a deterministic inhomogeneous left-canonical MPS.
	The associated site-\(n\) auxiliary channel is
	\[
		\Phi_n^\omega
		=
		\Phi_{\theta^{-n}\omega}.
	\]
	For \(0\le m<n\), the deterministic right-tail product is
	\[
		\Theta_{\omega;m,n}
		=
		\Phi_{\theta^{-(m+1)}\omega}
		\circ
		\Phi_{\theta^{-(m+2)}\omega}
		\circ\cdots\circ
		\Phi_{\theta^{-n}\omega}.
	\]
	Equivalently,
	\[
		\Theta_{\omega;m,n}
		=
		\Phi_{\theta^{-m}\omega;0:-(n-m)}.
	\]

	We verify the deterministic memory-loss hypothesis for this realization.
	Fix \(q\ge0\).
	For \(n>q\), set \(L:=n-q\).
	Then
	\[
		\kappa_{\rm tr}
		\left(
			\Theta_{\omega;q,n}
		\right)
		=
		\kappa_{\rm tr}
		\left(
			\Phi_{\theta^{-q}\omega;0:-L}
		\right).
	\]
	By the pullback coefficient estimate,
	\[
		\kappa_{\rm tr}
		\left(
			\Theta_{\omega;q,n}
		\right)
		\le
		C_\beta^-(\theta^{-q}\omega)
		e^{\beta(\theta^{-q}\omega)(n-q)}.
	\]
	Since \(\beta\) is \(\theta\)-invariant,
	\[
		\beta(\theta^{-q}\omega)=\beta(\omega).
	\]
	Thus
	\[
		\kappa_{\rm tr}
		\left(
			\Theta_{\omega;q,n}
		\right)
		\le
		C_\beta^-(\theta^{-q}\omega)
		e^{\beta(\omega)(n-q)}
		\longrightarrow0.
	\]
	Hence, the deterministic hypothesis of \Cref{thm:intro-mps-thermodynamic-limit} holds for the realization \(\omega\).

    By the deterministic memory-loss estimate just proved, \Cref{thm:intro-mps-thermodynamic-limit} applies to the fixed realization \(\omega\).
    Let
    \[
        \left(
            \widetilde\rho_r^\omega
        \right)_{r\ge1}
    \]
    and \(\widetilde\varphi_{\infty,\omega}\) denote the right boundary sequence and infinite-volume state produced by that fiberwise application.
    We now identify this fiberwise boundary sequence with the random stationary boundary supplied by \Cref{thm:negative-exponent-replacement}.
    For \(r\ge1\), define
    \[
        \rho_r^\omega
        :=
        \rho_{\theta^{-(r-1)}\omega}.
    \]
    Then
    \[
        \Phi_r^\omega(\rho_{r+1}^\omega)
        =
        \Phi_{\theta^{-r}\omega}
        \left(
            \rho_{\theta^{-r}\omega}
        \right)
        =
        \rho_{\theta^{-(r-1)}\omega}
        =
        \rho_r^\omega.
    \]
    Thus \((\rho_r^\omega)_{r\ge1}\) is a right boundary sequence for the fixed realization \(\omega\).
    By uniqueness in the fiberwise application of \Cref{thm:intro-mps-thermodynamic-limit},
    \[
        \widetilde\rho_r^\omega
        =
        \rho_r^\omega
        =
        \rho_{\theta^{-(r-1)}\omega}
        \qquad
        r\ge1.
    \]
    In particular,
    \[
        \widetilde\rho_{m+1}^\omega
        =
        \rho_{\theta^{-m}\omega}.
    \]
    Therefore, for every \(m\ge1\) and every \(X\in\mcA_{[1,m]}\),
    \[
        \widetilde\varphi_{\infty,\omega}(X)
        =
        \tr{
            \widehat X_{\omega;[1,m]}
            \left(
                \rho_{\theta^{-m}\omega}
            \right)
        }.
    \]
    For \(\omega\in\Omega_*\), set
    \[
        \varphi_{\infty,\omega}
        :=
        \widetilde\varphi_{\infty,\omega}.
    \]
    This proves the claimed local formula on \(\Omega_*\).

	It remains to prove the quantitative trace-closed convergence estimate.
	Fix \(m\ge1\) and \(X\in\mcA_{[1,m]}\).
	The quantitative estimate in the proof of \Cref{thm:intro-mps-thermodynamic-limit} gives, for all sufficiently large \(n\),
	\[
		\left|
			\varphi_{n,\omega}(X)
			-
			\varphi_{\infty,\omega}(X)
		\right|
		\le
		16D_{\mcH}^2
		\norm{X}_\infty
		\kappa_{\rm tr}
		\left(
			\Theta_{\omega;m,n}
		\right).
	\]
	Using
	\[
		\Theta_{\omega;m,n}
		=
		\Phi_{\theta^{-m}\omega;0:-(n-m)},
	\]
	we obtain
	\[
		\kappa_{\rm tr}
		\left(
			\Theta_{\omega;m,n}
		\right)
		\le
		C_\beta^-(\theta^{-m}\omega)
		e^{\beta(\theta^{-m}\omega)(n-m)}.
	\]
	Again \(\beta(\theta^{-m}\omega)=\beta(\omega)\).
	Therefore
	\[
		\left|
			\varphi_{n,\omega}(X)
			-
			\varphi_{\infty,\omega}(X)
		\right|
		\le
		16D_{\mcH}^2
		\norm{X}_\infty
		C_\beta^-(\theta^{-m}\omega)
		e^{\beta(\omega)(n-m)}
	\]
	for all sufficiently large \(n\).
	Since \(\beta(\omega)<0\), the right-hand side tends to zero.
	This proves the local convergence of the normalized trace-closed expectations.
	The same fiberwise deterministic result also gives \(\inner{\Psi_n^\omega}{\Psi_n^\omega}\ne0\) for all sufficiently large \(n\), so the normalized trace-closed expectations are well-defined in this range.

	We now prove the measurability of the random state field.
    Fix a reference state \(\psi_0\) on \(\mcA_{\mbN}\).
    For \(\omega\notin\Omega_*\), define
    \[
        \varphi_{\infty,\omega}
        :=
        \psi_0.
    \]
    First fix a local observable \(X\in\mcA_{[1,m]}\).
    On \(\Omega_*\), the local formula gives
    \[
        \varphi_{\infty,\omega}(X)
        =
        \tr{
            \widehat X_{\omega;[1,m]}
            \left(
                \rho_{\theta^{-m}\omega}
            \right)
        }.
    \]
    The map
    \[
        \omega\mapsto
        \widehat X_{\omega;[1,m]}
    \]
    is measurable because it is a finite sum of finite products of the measurable matrices \(K_i(\theta^{-r}\omega)\).
    The map
    \[
        \omega\mapsto
        \rho_{\theta^{-m}\omega}
    \]
    is measurable because \(\rho\) and \(\theta^{-m}\) are measurable.
    Since \(\Omega_*\) is measurable and the field is equal to the fixed state \(\psi_0\) on \(\Omega\setminus\Omega_*\), the map
    \[
        \omega\mapsto
        \varphi_{\infty,\omega}(X)
    \]
    is measurable for every local \(X\).
    Choose a countable norm-dense set \(\mcD\) in the algebraic local algebra \(\bigcup_{m\ge1}\mcA_{[1,m]}\), for instance, the local observables with rational real and imaginary matrix coefficients in the fixed product bases.
    For every \(X\in\mcD\), the coordinate map
    \[
        \omega\mapsto
        \varphi_{\infty,\omega}(X)
    \]
    is measurable.
    If \(Z\in\mcA_{\mbN}\), choose \(X_j\in\mcD\) with \(\norm{X_j-Z}_\infty\to0\).
    Since every \(\varphi_{\infty,\omega}\) is a state,
    \[
        \left|
            \varphi_{\infty,\omega}(X_j)
            -
            \varphi_{\infty,\omega}(Z)
        \right|
        \le
        \norm{X_j-Z}_\infty
    \]
    uniformly in \(\omega\).
    Hence \(\omega\mapsto\varphi_{\infty,\omega}(Z)\) is measurable for every \(Z\in\mcA_{\mbN}\).
    Thus \(\omega\mapsto\varphi_{\infty,\omega}\) is weak-\(*\) measurable.

    It remains to prove uniqueness up to \(\pr\)-a.e. equality.
    Let \(\chi_\infty:\Omega\to\mcS(\mcA_{\mbN})\) be another weak-\(*\) measurable state field satisfying the same local formula on a full-measure set.
    For each \(X\in\mcD\), the coordinate maps
    \[
        \omega\mapsto
        \chi_{\infty,\omega}(X)
        \qquad
        \text{and}
        \qquad
        \omega\mapsto
        \varphi_{\infty,\omega}(X)
    \]
    agree on a full-measure set.
    Since \(\mcD\) is countable, there is one full-measure set on which they agree for every \(X\in\mcD\).
    By norm density of \(\mcD\) and continuity of states, one has
    \[
        \chi_{\infty,\omega}
        =
        \varphi_{\infty,\omega}
    \]
    on that full-measure set.
    Therefore, the random infinite-volume state field is unique up to \(\pr\)-a.e. equality.
\end{proof}

We now prove \Cref{thm:intro-random-mps-correlation-bound}.

\introRandomMpsCorrelationBound*

\begin{proof}
    Fix \(\omega\in\Omega_*\).
    By the proof of \Cref{thm:intro-random-mps-thermodynamic-limit}, the deterministic realization
    \[
        K_i^{[n]}(\omega)
        =
        K_i(\theta^{-n}\omega)
    \]
    satisfies the hypotheses of \Cref{thm:intro-mps-thermodynamic-limit}.
    Moreover, the state \(\varphi_{\infty,\omega}\) is the fiberwise deterministic infinite-volume state for this realization.
    Therefore the deterministic correlation bound of \Cref{thm:intro-mps-correlation-bound} applies to \(\varphi_{\infty,\omega}\).

    Let
    \[
        L:=r-q-1.
    \]
    Since \(q+1<r\), one has \(L\ge1\).
    The identity-transfer product across the gap is
    \[
        \Theta_{\omega;q,r-1}
        =
        \Phi_{\theta^{-(q+1)}\omega}
        \circ
        \Phi_{\theta^{-(q+2)}\omega}
        \circ
        \cdots
        \circ
        \Phi_{\theta^{-(r-1)}\omega}.
    \]
    Equivalently,
    \[
        \Theta_{\omega;q,r-1}
        =
        \Phi_{\theta^{-q}\omega;0:-L}.
    \]

    By \Cref{thm:intro-mps-correlation-bound},
    \[
        \left|
            \varphi_{\infty,\omega}(AB)
            -
            \varphi_{\infty,\omega}(A)
            \varphi_{\infty,\omega}(B)
        \right|
        \le
        4
        \norm{A}_\infty
        \norm{B}_\infty
        \kappa_{\rm tr}
        \left(
            \Theta_{\omega;q,r-1}
        \right).
    \]
    By the pullback coefficient estimate used in the proof of \Cref{thm:intro-random-mps-thermodynamic-limit},
    \[
        \kappa_{\rm tr}
        \left(
            \Theta_{\omega;q,r-1}
        \right)
        =
        \kappa_{\rm tr}
        \left(
            \Phi_{\theta^{-q}\omega;0:-L}
        \right)
        \le
        C_\beta^-(\theta^{-q}\omega)
        e^{\beta(\theta^{-q}\omega)L}.
    \]
    Since \(\beta\) is \(\theta\)-invariant on \(\Omega_*\),
    \[
        \beta(\theta^{-q}\omega)
        =
        \beta(\omega).
    \]
    Combining the last two estimates gives
    \[
        \left|
            \varphi_{\infty,\omega}(AB)
            -
            \varphi_{\infty,\omega}(A)
            \varphi_{\infty,\omega}(B)
        \right|
        \le
        4
        \norm{A}_\infty
        \norm{B}_\infty
        C_\beta^-(\theta^{-q}\omega)
        e^{\beta(\omega)L}.
    \]
    This proves the quenched correlation estimate.

    For fixed \(A\) and \(B\), the left-hand side is measurable in \(\omega\) because \(\omega\mapsto\varphi_{\infty,\omega}\) is weak-\(*\) measurable.
\end{proof}

\begin{cor}
\label{cor:random-mps-high-probability-clustering}
    Assume the hypotheses of \Cref{thm:intro-random-mps-correlation-bound}.
    Let \(A\in\mcA_{[p,q]}\) and \(B\in\mcA_{[r,s]}\), where \(1\le p\le q\) and \(q+1<r\le s\).
    Set
    \[
        L:=r-q-1.
    \]
    Define the connected correlation random variable
    \[
        \Gamma_{A,B}(\omega)
        :=
        \left|
            \varphi_{\infty,\omega}(AB)
            -
            \varphi_{\infty,\omega}(A)
            \varphi_{\infty,\omega}(B)
        \right|.
    \]
    If the channel environment is \(\varrho\)-mixing, then for every \(u\in\mbN\) there exists \(C_u<\infty\) such that
    \[
        \pr
        \left\{
            \Gamma_{A,B}
            \le
            C_u
            \norm{A}_\infty
            \norm{B}_\infty
            L^{-u}
        \right\}
        \ge
        1-L^{-u}.
    \]
    If the random variables \((X_j)_{j\in\mbZ}\) are jointly independent, then there exist constants \(C<\infty\) and \(\gamma>0\) such that
    \[
        \pr
        \left\{
            \Gamma_{A,B}
            \le
            C
            \norm{A}_\infty
            \norm{B}_\infty
            e^{-\gamma L}
        \right\}
        \ge
        1-e^{-\gamma L}.
    \]
\end{cor}

\begin{proof}
    By \Cref{thm:intro-random-mps-correlation-bound}, for every \(\omega\in\Omega_*\),
    \[
        \Gamma_{A,B}(\omega)
        \le
        4
        \norm{A}_\infty
        \norm{B}_\infty
        \kappa_{\rm tr}
        \left(
            \Theta_{\omega;q,r-1}
        \right).
    \]
    The gap product is
    \[
        \Theta_{\omega;q,r-1}
        =
        \Phi_{\theta^{-q}\omega;0:-L}.
    \]
    Hence
    \[
        \kappa_{\rm tr}
        \left(
            \Theta_{\omega;q,r-1}
        \right)
        =
        \kappa_{\theta^{-(r-1)}\omega;L:0}.
    \]
    Since \(\theta\) preserves \(\pr\), the random variable \(\kappa_{\theta^{-(r-1)}\omega;L:0}\) has the same law as \(\kappa_{\omega;L:0}\).
    Suppose first that the channel environment is \(\varrho\)-mixing.
    By \Cref{lemma:kappa-annealed-superpoly}, applied with exponent \(2u\), there exists \(A_{2u}<\infty\) such that
    \[
        \mbE
        \left[
            \kappa_{\rm tr}
            \left(
                \Theta_{\omega;q,r-1}
            \right)
        \right]
        =
        \mbE
        \left[
            \kappa_{\omega;L:0}
        \right]
        \le
        A_{2u}L^{-2u}.
    \]
    Set
    \[
        B_u:=\max\{A_{2u},1\}.
    \]
    Markov's inequality gives
    \[
    \begin{aligned}
        \pr
        \left\{
            \kappa_{\rm tr}
            \left(
                \Theta_{\omega;q,r-1}
            \right)
            >
            B_uL^{-u}
        \right\}
        &\le
        \frac{
            A_{2u}L^{-2u}
        }{
            B_uL^{-u}
        }
        \\
        &\le
        L^{-u}.
    \end{aligned}
    \]
    Therefore, with probability at least \(1-L^{-u}\),
    \[
        \Gamma_{A,B}(\omega)
        \le
        4B_u
        \norm{A}_\infty
        \norm{B}_\infty
        L^{-u}.
    \]
    This proves the \(\varrho\)-mixing claim with \(C_u:=4B_u\).

    Now suppose that the random variables \((X_j)_{j\in\mbZ}\) are jointly independent.
    By \Cref{lemma:kappa-annealed-independent}, there exist \(A<\infty\) and \(\eta>0\) such that
    \[
        \mbE
        \left[
            \kappa_{\rm tr}
            \left(
                \Theta_{\omega;q,r-1}
            \right)
        \right]
        =
        \mbE
        \left[
            \kappa_{\omega;L:0}
        \right]
        \le
        Ae^{-\eta L}.
    \]
    Set
    \[
        B:=\max\{A,1\},
        \qquad
        \gamma:=\frac{\eta}{2}.
    \]
    Markov's inequality gives
    \[
    \begin{aligned}
        \pr
        \left\{
            \kappa_{\rm tr}
            \left(
                \Theta_{\omega;q,r-1}
            \right)
            >
            Be^{-\gamma L}
        \right\}
        &\le
        \frac{
            Ae^{-\eta L}
        }{
            Be^{-\gamma L}
        }
        \\
        &\le
        e^{-\gamma L}.
    \end{aligned}
    \]
    Therefore, with probability at least \(1-e^{-\gamma L}\),
    \[
        \Gamma_{A,B}(\omega)
        \le
        4B
        \norm{A}_\infty
        \norm{B}_\infty
        e^{-\gamma L}.
    \]
    This proves the independent case with \(C:=4B\).
\end{proof}
%%%%%%%%%%%%%%%%%%%%%%%%%%%%%%%%%%%%%%%%%%%%%%%%%%%%%%%%%%%%%%%%%%%%%%%%%%%%
%%%%%%%%%%%%%%%%%%%%%%%%%%%%%%%%%%%%%%%%%%%%%%%%%%%%%%%%%%%%%%%%%%%%%%%%%%%%

\section*{Acknowledgments}
\addcontentsline{toc}{section}{Acknowledgments}

%%%%%%%%%%%%%%%%%%%%%%%%%%%%%%%%%%%%%%%%%%%%%%%%%%%%%%%%%%%%%%%%%%%%%%%%%%%%
%%%%%%%%%%%%%%%%%%%%%%%%%%%%%%%%%%%%%%%%%%%%%%%%%%%%%%%%%%%%%%%%%%%%%%%%%%%%

    The authors acknowledge support by the Danish e-infrastructure Consortium (DeiC) grant 5260-00014. 
    Parts of this work were also supported by Villum Fonden Grant No. 25452 and Grant No. 60842, as well as QMATH Center of Excellence Grant No. 10059. 

%%%%%%%%%%%%%%%%%%%%%%%%%%%%%%%%%%%%%%%%%%%%%%%%%%%%%%%%%%%%%%%%%%%%%%%%%%%%
%%%%%%%%%%%%%%%%%%%%%%%%%%%%%%%%%%%%%%%%%%%%%%%%%%%%%%%%%%%%%%%%%%%%%%%%%%%%

\begin{appendix}

%%%%%%%%%%%%%%%%%%%%%%%%%%%%%%%%%%%%%%%%%%%%%%%%%%%%%%%%%%%%%%%%%%%%%%%%%%%%
%%%%%%%%%%%%%%%%%%%%%%%%%%%%%%%%%%%%%%%%%%%%%%%%%%%%%%%%%%%%%%%%%%%%%%%%%%%%

\section{Additional comparison with contraction coefficients}
\label{app:comparison-coefficients}

%%%%%%%%%%%%%%%%%%%%%%%%%%%%%%%%%%%%%%%%%%%%%%%%%%%%%%%%%%%%%%%%%%%%%%%%%%%%
%%%%%%%%%%%%%%%%%%%%%%%%%%%%%%%%%%%%%%%%%%%%%%%%%%%%%%%%%%%%%%%%%%%%%%%%%%%%

The main text uses only the state-level Markov--Dobrushin comparison in \Cref{subsec:comparison-one-step-criteria}.
For completeness, we record here how the trace-Dobrushin product coefficient compares with several other contraction mechanisms.
These coefficients answer related but distinct questions.
Some are computable one-step upper bounds.
Some are adapted to projective geometry or faithful-state geometry.
Some are designed for spatial many-body dynamics or for completely bounded memory loss.
There is therefore no single linear hierarchy containing all of them.

%%%%%%%%%%%%%%%%%%%%%%%%%%%%%%%%%%%%%%%%%%%%%%%%%%%%%%%%%%%%%%%%%%%%%%%%%%%%
%%%%%%%%%%%%%%%%%%%%%%%%%%%%%%%%%%%%%%%%%%%%%%%%%%%%%%%%%%%%%%%%%%%%%%%%%%%%

\subsection*{CP-order quantum Doeblin coefficients}
\label{app:comparison-doeblin}

%%%%%%%%%%%%%%%%%%%%%%%%%%%%%%%%%%%%%%%%%%%%%%%%%%%%%%%%%%%%%%%%%%%%%%%%%%%%
%%%%%%%%%%%%%%%%%%%%%%%%%%%%%%%%%%%%%%%%%%%%%%%%%%%%%%%%%%%%%%%%%%%%%%%%%%%%

The state-level Markov--Dobrushin coefficient only asks for a common lower bound at the level of output states.
A stronger certificate asks for a lower bound in the completely positive order.
For \(\tau\in\mcS_d\), write \(R_\tau(X):=\tr{X}\tau\) for the replacement channel with output \(\tau\).

\begin{definition}[CP-order quantum Doeblin coefficient]
\label{def:cp-order-doeblin}
Let \(\Phi:\matrices\to\matrices\) be CPTP.
The CP-order quantum Doeblin coefficient is
\[
   \alpha_{\rm Doeb}(\Phi)
   :=
   \max
   \left\{
        c\in[0,1]:
        \exists\tau\in\mcS_d
        \text{ such that }
        \Phi-cR_\tau \text{ is completely positive}
   \right\}.
\]
\end{definition}

Equivalently, this coefficient can be characterized through degradability from the erasure channel:
\[
   \alpha_{\rm Doeb}(\Phi)
   =
   \sup\{\varepsilon:E_\varepsilon\succeq_{\rm deg}\Phi\},
\]
where \(E_\varepsilon\) is the erasure channel and \(\succeq_{\rm deg}\) denotes the degradability order; see \cite{Hirche2024}.
We use the normalized Choi convention
\[
   J(\Phi)
   :=
   (\Phi\otimes\operatorname{id})(|\Omega\rangle\langle\Omega|),
   \qquad
   |\Omega\rangle
   :=
   d^{-1/2}\sum_{i=1}^d |i\rangle\otimes |i\rangle .
\]
With this convention, \(J(R_\tau)=\tau\otimes I/d\), and the coefficient has the semidefinite-program form
\[
   \alpha_{\rm Doeb}(\Phi)
   =
   \max_{\widehat\tau\ge0}
   \left\{
      \tr{\widehat\tau}:
      \widehat\tau\otimes I/d\le J(\Phi)
   \right\}.
\]
With the opposite input-output Choi convention, the tensor factors are interchanged.

The coefficient gives the trace-distance contraction bound
\[
   \kappa_{\rm tr}(\Phi)
   =
   \eta^{\rm Tr}(\Phi)
   \le
   1-\alpha_{\rm Doeb}(\Phi).
\]
Indeed, if \(\Phi-cR_\tau\) is completely positive and \(c<1\), then
\[
   \Phi
   =
   cR_\tau+(1-c)\Psi
\]
for a CPTP map \(\Psi\).
Therefore, for all states \(\rho,\sigma\in\mcS_d\),
\[
   \Phi(\rho)-\Phi(\sigma)
   =
   (1-c)\bigl(\Psi(\rho)-\Psi(\sigma)\bigr),
\]
and trace-norm contractivity gives
\[
   \norm{\Phi(\rho)-\Phi(\sigma)}_1
   \le
   (1-c)\norm{\rho-\sigma}_1.
\]
Optimizing over \(c\) gives the displayed bound.
The case \(c=1\) is the replacement-channel case and gives zero trace-distance contraction.

The paper \cite{Hirche2024} also introduces transpose and Hermitian-relaxed variants.
The transpose variant is
\[
   \alpha_T(\Phi)
   :=
   \alpha_{\rm Doeb}(\mathsf T\circ\Phi)
\]
whenever \(\mathsf T\circ\Phi\) is completely positive, where \(\mathsf T\) denotes matrix transposition.
If \(\mathsf T\circ\Phi\) is not completely positive, we omit this term, or equivalently set \(\alpha_T(\Phi)=-\infty\) in the maximum below.
The Hermitian relaxation is
\[
   \alpha_H(\Phi)
   :=
   \max
   \left\{
        c\in[0,1]:
        \exists X=X^*,\ \tr{X}=1,\ 
        \Phi-cR_X \text{ is completely positive}
   \right\},
\]
where \(R_X(Y):=\tr{Y}X\).
In the definition of \(\alpha_H\), the operator \(X\) is not required to be positive, so \(R_X\) need not be a quantum channel.
It is only a Hermitian relaxation used to bound the trace-distance contraction coefficient.
These variants give an improved estimate
\[
   \eta^{\rm Tr}(\Phi)
   \le
   1-\max\{\alpha_H(\Phi),\alpha_T(\Phi)\}
\]
when the corresponding quantities are applicable.

The CP-order Doeblin condition implies state-level Markov--Dobrushin minorization.
Indeed, if \(\Phi-cR_\tau\) is completely positive, then for every \(P\in\mcP_1\),
\[
   c\tau
   =
   cR_\tau(P)
   \le
   \Phi(P).
\]
Therefore
\[
   \alpha_{\rm Doeb}(\Phi)
   \le
   \alpha_{\rm MD}(\Phi),
   \qquad
   1-\alpha_{\rm MD}(\Phi)
   \le
   1-\alpha_{\rm Doeb}(\Phi).
\]
Thus, CP-order Doeblin minorization is a stronger certificate than state-level Markov--Dobrushin minorization.
Its advantage is that it is algebraic and SDP-computable.
Its disadvantage is that it can be too rigid.

\begin{example}
\label{ex:cp-vs-state-doeblin}
    Let \(d\ge2\), and consider
    \[
       \Phi(X)
       =
       \frac{1}{d+1}\bigl(\tr{X}I+X^T\bigr).
    \]
    This map is CPTP.
    Indeed, with the normalized Choi convention fixed above,
    \[
       J(\Phi)
       =
       \frac{1}{d(d+1)}(I+F),
    \]
    where \(F\) is the flip operator on \(\mbC^d\otimes\mbC^d\).
    Thus \(J(\Phi)\ge0\), because \(I+F\) is twice the projection onto the symmetric subspace.
    Trace preservation is immediate from the formula, since
    \[
       \tr{\Phi(X)}
       =
       \frac{1}{d+1}\bigl(d\,\tr{X}+\tr{X}\bigr)
       =
       \tr{X}.
    \]
    
    For every \(P\in\mcP_1\),
    \[
       \Phi(P)
       =
       \frac{1}{d+1}(I+P^T)
       \ge
       \frac{1}{d+1}I.
    \]
    Hence
    \[
       \alpha_{\rm MD}(\Phi)
       \ge
       \frac{d}{d+1}
       >
       0.
    \]
    
    On the other hand, \(J(\Phi)\) is supported on the symmetric subspace \(\operatorname{Sym}^2(\mbC^d)\).
    Suppose that \(\Phi-cR_\tau\) were completely positive for some \(c>0\) and some \(\tau\in\mcS_d\).
    Then
    \[
       0
       \le
       cJ(R_\tau)
       \le
       J(\Phi).
    \]
    For positive operators \(0\le A\le B\), one has \(\operatorname{supp}(A)\subseteq\operatorname{supp}(B)\).
    Therefore \(J(R_\tau)\) would have to be supported on \(\operatorname{Sym}^2(\mbC^d)\).
    But
    \[
       J(R_\tau)
       =
       \tau\otimes I/d,
    \]
    and hence
    \[
       \operatorname{supp}J(R_\tau)
       =
       \operatorname{supp}(\tau)\otimes\mbC^d.
    \]
    Choose \(0\ne u\in\operatorname{supp}(\tau)\).
    Since \(d\ge2\), choose \(v\in\mbC^d\) with \(v\perp u\).
    Then \(u\otimes v\in\operatorname{supp}J(R_\tau)\), but \(u\otimes v\) is not symmetric, because \(F(u\otimes v)=v\otimes u\ne u\otimes v\).
    This contradicts the required support inclusion.
    Therefore no such \(c>0\) exists, and
    \[
       \alpha_{\rm Doeb}(\Phi)=0.
    \]
    Thus, CP-order Doeblin minorization is strictly stronger than state-level Markov--Dobrushin minorization.
\end{example}

%%%%%%%%%%%%%%%%%%%%%%%%%%%%%%%%%%%%%%%%%%%%%%%%%%%%%%%%%%%%%%%%%%%%%%%%%%%%
%%%%%%%%%%%%%%%%%%%%%%%%%%%%%%%%%%%%%%%%%%%%%%%%%%%%%%%%%%%%%%%%%%%%%%%%%%%%

\subsection*{Hilbert--Birkhoff projective coefficients}
\label{app:comparison-hilbert}

%%%%%%%%%%%%%%%%%%%%%%%%%%%%%%%%%%%%%%%%%%%%%%%%%%%%%%%%%%%%%%%%%%%%%%%%%%%%
%%%%%%%%%%%%%%%%%%%%%%%%%%%%%%%%%%%%%%%%%%%%%%%%%%%%%%%%%%%%%%%%%%%%%%%%%%%%

Projective methods belong to a different geometric framework.
They are designed for positive cone-preserving maps and are especially natural when the map is not trace-preserving.

\begin{definition}[Hilbert--Birkhoff projective coefficient]
    Let \(A,B\ge0\) be nonzero.
    Define
    \[
       h(A,B)
       :=
       \log\big(\sup(A/B)\sup(B/A)\big),
       \qquad
       \sup(A/B):=\inf\{\lambda>0:A\le\lambda B\},
    \]
    with the convention that \(\sup(A/B)=\infty\) if no such \(\lambda\) exists.
    If \(T\) is positive and cone preserving, its projective diameter is
    \[
       \Delta(T):=\sup_{A,B\ge0,\ A,B\ne0}h(TA,TB),
    \]
    and its Hilbert--Birkhoff projective contraction coefficient is
    \[
       \eta_{\rm HB}(T):=\tanh(\Delta(T)/4).
    \]
\end{definition}

The coefficient \(\eta_{\rm HB}\) is submultiplicative and gives Perron--Frobenius type control for positive maps; see \cite{Birkhoff1957,Bushell1973,ReebKastoryanoWolf2011}.
For trace-preserving positive maps on the positive semidefinite cone, finite projective diameter implies base-norm, hence trace-norm, contraction.
In particular, when \(\Delta(T)<\infty\),
\[
   \kappa_{\rm tr}(T)
   \le
   \tanh(\Delta(T)/4)
   <
   1.
\]
Thus finite projective diameter is a sufficient condition for strict trace contraction.

However, this condition is not necessary.
The finite-diameter projective criterion may fail to detect trace-norm contraction toward a boundary state.

\begin{example}
\label{ex:amplitude}
    For \(0<\gamma<1\), define the qubit amplitude-damping channel
    \[
       T_\gamma(X)=K_0XK_0^*+K_1XK_1^*,
    \]
    where
    \[
       K_0=
       \begin{pmatrix}
       1&0\\
       0&\sqrt{1-\gamma}
       \end{pmatrix},
       \qquad
       K_1=
       \begin{pmatrix}
       0&\sqrt{\gamma}\\
       0&0
       \end{pmatrix}.
    \]
    On centered Bloch vectors, the linear part is
    \[
       (x,y,z)\mapsto
       \big(\sqrt{1-\gamma}\,x,\sqrt{1-\gamma}\,y,(1-\gamma)z\big).
    \]
    Therefore
    \[
       \kappa_{\rm tr}(T_\gamma)=\sqrt{1-\gamma}<1.
    \]
    In particular, \(T_\gamma^n(\rho)\to |0\rangle\langle0|\) exponentially for every state \(\rho\).
    
    Nevertheless,
    \[
       T_\gamma^n(|0\rangle\langle0|)
       =
       |0\rangle\langle0|
       \qquad\text{for all }n.
    \]
    Moreover,
    \[
       T_\gamma^n(|1\rangle\langle1|)
       =
       \bigl(1-(1-\gamma)^n\bigr)|0\rangle\langle0|
       +
       (1-\gamma)^n|1\rangle\langle1|,
    \]
    which has full support for every finite \(n\).
    Thus the two outputs \(T_\gamma^n(|0\rangle\langle0|)\) and \(T_\gamma^n(|1\rangle\langle1|)\) have different supports.
    Their Hilbert projective distance is infinite, and hence
    \[
       \Delta(T_\gamma^n)=\infty
       \qquad\text{for every }n\ge1.
    \]
    Thus strict trace contraction can hold even when all finite-step projective diameters are infinite.
    
    The same example also shows that one-step minorization is not necessary for trace contraction.
    Suppose \(B\le T_\gamma(P)\) for all \(P\in\mcP_1\).
    Since \(T_\gamma(|0\rangle\langle0|)=|0\rangle\langle0|\), the operator \(B\) must be of the form \(B=b|0\rangle\langle0|\) with \(b\ge0\).
    Now take
    \[
       |\psi_\varepsilon\rangle
       =
       \sqrt{1-\varepsilon}|0\rangle+\sqrt{\varepsilon}|1\rangle,
       \qquad
       P_\varepsilon=|\psi_\varepsilon\rangle\langle\psi_\varepsilon|.
    \]
    A direct computation gives
    \[
       T_\gamma(P_\varepsilon)-b|0\rangle\langle0|\ge0
       \quad\Longrightarrow\quad
       b\le \gamma\varepsilon.
    \]
    Since this must hold for all \(\varepsilon>0\), one gets \(b=0\).
    Hence
    \[
       \alpha_{\rm MD}(T_\gamma)=0.
    \]
    Since CP-order Doeblin minorization implies state-level Markov--Dobrushin minorization, it follows also that
    \[
       \alpha_{\rm Doeb}(T_\gamma)=0.
    \]
    Thus
    \[
       \alpha_{\rm MD}(T_\gamma)=0,
       \qquad
       \alpha_{\rm Doeb}(T_\gamma)=0,
       \qquad
       \kappa_{\rm tr}(T_\gamma)<1.
    \]
    This example shows that strict trace contraction is strictly more general than strict positivity, finite projective diameter, CP-order Doeblin minorization, and one-step Markov--Dobrushin minorization.
\end{example}

%%%%%%%%%%%%%%%%%%%%%%%%%%%%%%%%%%%%%%%%%%%%%%%%%%%%%%%%%%%%%%%%%%%%%%%%%%%%
%%%%%%%%%%%%%%%%%%%%%%%%%%%%%%%%%%%%%%%%%%%%%%%%%%%%%%%%%%%%%%%%%%%%%%%%%%%%

\subsection*{Projective contraction coefficients}
\label{app:comparison-cnum}

%%%%%%%%%%%%%%%%%%%%%%%%%%%%%%%%%%%%%%%%%%%%%%%%%%%%%%%%%%%%%%%%%%%%%%%%%%%%
%%%%%%%%%%%%%%%%%%%%%%%%%%%%%%%%%%%%%%%%%%%%%%%%%%%%%%%%%%%%%%%%%%%%%%%%%%%%

The random positive-cocycle literature also uses a normalized projective coefficient which is closely related to, but not identical with, the Hilbert--Birkhoff contraction coefficient.

\begin{definition}[Projective contraction coefficient]
    For states \(A,B\in\mcS_d\), set
    \[
       m(A,B)
       :=
       \sup\{\lambda\ge0:\lambda B\le A\},
    \]
    and define
    \[
       d_{\rm proj}(A,B)
       :=
       \frac{1-m(A,B)m(B,A)}{1+m(A,B)m(B,A)}.
    \]
    Let \(\varphi:\matrices\to\matrices\) be a positive map such that \(\tr{\varphi(A)}>0\) for every \(A\in\mcS_d\).
    Define its normalized projective action by
    \[
       \varphi\cdot A
       :=
       \frac{\varphi(A)}{\tr{\varphi(A)}}.
    \]
    The projective contraction coefficient is
    \[
       c(\varphi)
       :=
       \sup_{A,B\in\mcS_d}
       d_{\rm proj}(\varphi\cdot A,\varphi\cdot B).
    \]
\end{definition}

This coefficient is useful for positive maps and random positive cocycles, where the maps need not be trace preserving; see \cite{MovassaghSchenker2021,MovassaghSchenker2022}.
It measures the diameter of the normalized image of the state space in a bounded projective metric.

Let \(h\) denote Hilbert's projective metric on the positive semidefinite cone.
For comparable nonzero positive operators \(A,B\), one has
\[
   h(A,B)
   =
   -\log\bigl(m(A,B)m(B,A)\bigr),
\]
and therefore
\[
   d_{\rm proj}(A,B)
   =
   \tanh\!\left(\frac{h(A,B)}{2}\right).
\]
Thus \(d_{\rm proj}\) is a bounded transform of Hilbert's projective metric.
However, the coefficient \(c(\varphi)\) is not the same as the Hilbert--Birkhoff contraction coefficient
\[
   \eta_{\rm HB}(\varphi)
   =
   \tanh\!\left(\frac{\Delta(\varphi)}{4}\right),
\]
where \(\Delta(\varphi)\) is the Hilbert projective diameter.
The former records the bounded projective diameter of the normalized image, while the latter is the Birkhoff--Hopf Lipschitz contraction ratio for Hilbert's projective metric.

In the CPTP case, the normalization is trivial, but \(c\) remains a projective coefficient rather than a trace-norm memory coefficient.
It is well adapted to strict positivity and projective convergence.
By contrast, \(\kappa_{\rm tr}\) is adapted to trace-norm replacement mixing, including convergence toward nonfaithful boundary states.

For instance, in the amplitude-damping example of \Cref{ex:amplitude}, one has
\[
   \kappa_{\rm tr}(T_\gamma^n)
   =
   (1-\gamma)^{n/2}
   \longrightarrow 0.
\]
Nevertheless,
\[
   T_\gamma^n(|0\rangle\langle0|)
   =
   |0\rangle\langle0|
\]
and
\[
   T_\gamma^n(|1\rangle\langle1|)
   =
   \bigl(1-(1-\gamma)^n\bigr)|0\rangle\langle0|
   +
   (1-\gamma)^n|1\rangle\langle1|.
\]
The first output is rank one, whereas the second has full support for every finite \(n\).
Hence \[m(T_\gamma^n(|0\rangle\langle0|),T_\gamma^n(|1\rangle\langle1|))=0,\] and therefore
\[
   c(T_\gamma^n)=1
   \qquad
   \text{for every }n\ge1.
\]
Thus, the projective contraction coefficient does not detect the trace-norm replacement convergence of amplitude damping toward a boundary state.
This is exactly the kind of boundary-attractor behavior for which the centered trace-Dobrushin coefficient is the more appropriate CPTP product coefficient.

%%%%%%%%%%%%%%%%%%%%%%%%%%%%%%%%%%%%%%%%%%%%%%%%%%%%%%%%%%%%%%%%%%%%%%%%%%%%
%%%%%%%%%%%%%%%%%%%%%%%%%%%%%%%%%%%%%%%%%%%%%%%%%%%%%%%%%%%%%%%%%%%%%%%%%%%%

\subsection*{Entropy, Riemannian, and \texorpdfstring{\(\chi^2\)}{chi-squared} coefficients}
\label{app:comparison-entropy-chi2}

%%%%%%%%%%%%%%%%%%%%%%%%%%%%%%%%%%%%%%%%%%%%%%%%%%%%%%%%%%%%%%%%%%%%%%%%%%%%
%%%%%%%%%%%%%%%%%%%%%%%%%%%%%%%%%%%%%%%%%%%%%%%%%%%%%%%%%%%%%%%%%%%%%%%%%%%%

Entropy and \(L^2\)-type contraction coefficients are adapted to divergence and faithful-state geometries.
They are extremely useful for entropy production, functional inequalities, and local spectral analysis.
They are not the same object as the centered trace-Dobrushin coefficient \(\kappa_{\rm tr}\).

\begin{definition}[Relative entropy contraction coefficient]
    For a CPTP map \(\Phi\), define
    \[
       \eta_{\rm D}(\Phi)
       :=
       \sup_{\rho,\sigma}
       \frac{D(\Phi(\rho)\|\Phi(\sigma))}{D(\rho\|\sigma)},
    \]
    where
    \[
       D(\rho\|\sigma)
       :=
       \tr{\rho(\log\rho-\log\sigma)}.
    \]
    The supremum is taken over pairs of states with \(0<D(\rho\|\sigma)<\infty\).
\end{definition}

\begin{definition}[Riemannian contraction coefficient]
    Fix a monotone Riemannian metric \(g^{\mathfrak m}\) on the faithful state space.
    Assume that \(\Phi(\rho)>0\) whenever \(\rho>0\).
    The associated Riemannian contraction coefficient is
    \[
       \eta_{\rm Riem,\mathfrak m}(\Phi)
       :=
       \sup_{\rho>0}
       \sup_{\substack{X=X^*\\ \tr{X}=0\\ X\ne0}}
       \frac{
          g^{\mathfrak m}_{\Phi(\rho)}(\Phi(X),\Phi(X))
       }{
          g^{\mathfrak m}_{\rho}(X,X)
       }.
    \]
\end{definition}

\begin{definition}[\(\chi^2\) contraction coefficient]
    Let \(\omega>0\) be a faithful state with \(\Phi(\omega)=\omega\), and let \(\chi^2_\omega\) be a chosen quantum \(\chi^2\)-divergence associated with a monotone metric.
    The corresponding \(\chi^2\) contraction coefficient is
    \[
       \eta_{\chi^2,\omega}(\Phi)
       :=
       \sup_{\rho\ne\omega}
       \frac{\chi^2_\omega(\Phi(\rho))}{\chi^2_\omega(\rho)}.
    \]
\end{definition}

When the Riemannian coefficient is defined, the comparison results for monotone metrics imply inequalities of the form
\[
   \eta^{\rm Tr}(\Phi)
   \le
   \sqrt{\eta_{\rm Riem,\mathfrak m}(\Phi)}
\]
in the appropriate setting; see \cite{HiaiRuskai2016,Hirche2024}.
Thus, strict Riemannian contraction implies strict trace-distance contraction.

For a unital qubit channel with Bloch matrix \(A_\Phi\), the usual comparison takes the simple form
\[
   \eta_{\rm Riem,\mathfrak m}(\Phi)
   =
   \eta_{\rm D}(\Phi)
   =
   \|A_\Phi\|_\infty^2,
   \qquad
   \eta^{\rm Tr}(\Phi)
   =
   \|A_\Phi\|_\infty .
\]
Thus, in this case, the trace-distance coefficient is the square root of the corresponding Riemannian and relative-entropy contraction coefficient.

These comparison inequalities are useful, but they do not replace \(\kappa_{\rm tr}\).
The entropy, Riemannian, and \(\chi^2\) coefficients are tied to divergence, infinitesimal, or faithful-reference-state geometries.
By contrast, \(\kappa_{\rm tr}\) is defined directly on the self-adjoint trace-zero space and remains well behaved for boundary-attracting channels.
For example, amplitude damping has a nonfaithful attracting state, so the centered trace-Dobrushin coefficient still captures trace-norm replacement mixing, while faithful-state Riemannian coefficients are not the natural tool for that boundary limit.

%%%%%%%%%%%%%%%%%%%%%%%%%%%%%%%%%%%%%%%%%%%%%%%%%%%%%%%%%%%%%%%%%%%%%%%%%%%%
%%%%%%%%%%%%%%%%%%%%%%%%%%%%%%%%%%%%%%%%%%%%%%%%%%%%%%%%%%%%%%%%%%%%%%%%%%%%

\subsection*{Quantum Wasserstein--Dobrushin conditions}
\label{app:comparison-wasserstein}

%%%%%%%%%%%%%%%%%%%%%%%%%%%%%%%%%%%%%%%%%%%%%%%%%%%%%%%%%%%%%%%%%%%%%%%%%%%%
%%%%%%%%%%%%%%%%%%%%%%%%%%%%%%%%%%%%%%%%%%%%%%%%%%%%%%%%%%%%%%%%%%%%%%%%%%%%

The many-body quantum Dobrushin framework uses a spatial influence matrix rather than a scalar auxiliary-channel coefficient.
It is designed to control how local deviations propagate through a spatial quantum system.

A typical setup starts from a many-body update \(\Phi\) decomposed into local pieces \(\Phi_i\), where \(\Phi_i\) is the part of the dynamics associated with updating site \(i\).
The relevant metric is a quantum Wasserstein-1 norm \(\norm{\cdot}_{W_1}\) on traceless Hermitian perturbations.
In the formulation of \cite{BakshiLiuMoitraTang2025}, the influence of a deviation at site \(j\) on the update at site \(i\) is measured by a quantity of the form
\[
   D^{(\Phi)}_{ij}
   :=
   \sup
   \left\{
      \norm{\Phi_i(X)}_{W_1}:
      X=X^*,\ \tr{X}=0,\ \operatorname{tr}_j X=0,\ \norm{X}_{W_1}=1
   \right\},
\]
up to the normalization convention used in the decomposition of \(\Phi\).
Equivalently, one may take the supremum over pairs of \(j\)-neighboring quantum states.
A quantum Wasserstein--Dobrushin condition is then a contraction condition on this influence matrix, for instance, a bound on its column sums or on its \(1\to1\) norm strictly below one.

This is a spatial path-coupling criterion.
It controls the propagation of local deviations and yields rapid mixing for suitable many-body dynamics.
In the high-temperature Gibbs-state setting, the same framework can also be used to prove structural consequences such as decay of conditional mutual information.
This is a different problem from the scalar trace-memory-loss problem studied in the present paper.
Here the deviation is the whole auxiliary trace-zero component of a finite-dimensional channel product.
In the Wasserstein--Dobrushin setting, the deviation is spatially localized and the coefficient records how that localization propagates across sites.
Thus the two theories use Dobrushin-type language for different geometries: \(\kappa_{\rm tr}\) is a scalar trace-norm memory coefficient for auxiliary channel products, while \(D^{(\Phi)}\) is a spatial influence matrix for many-body quantum dynamics.

%%%%%%%%%%%%%%%%%%%%%%%%%%%%%%%%%%%%%%%%%%%%%%%%%%%%%%%%%%%%%%%%%%%%%%%%%%%%
%%%%%%%%%%%%%%%%%%%%%%%%%%%%%%%%%%%%%%%%%%%%%%%%%%%%%%%%%%%%%%%%%%%%%%%%%%%%

\subsection*{Diamond-norm variants}
\label{app:comparison-diamond}

%%%%%%%%%%%%%%%%%%%%%%%%%%%%%%%%%%%%%%%%%%%%%%%%%%%%%%%%%%%%%%%%%%%%%%%%%%%%
%%%%%%%%%%%%%%%%%%%%%%%%%%%%%%%%%%%%%%%%%%%%%%%%%%%%%%%%%%%%%%%%%%%%%%%%%%%%

One can define completely bounded versions of centered memory loss, but the most naive unconstrained version is not useful for ordinary state-memory loss.

\begin{definition}[Unconstrained diamond centered coefficient]
    For a CPTP map \(\Phi:\matrices\to\matrices\), define
    \[
       \kappa_\diamond(\Phi)
       :=
       \sup_{r\ge1}
       \sup_{\substack{X=X^*\in \matrices\otimes M_r\\ \tr{X}=0,\ X\ne0}}
       \frac{
          \norm{(\Phi\otimes\operatorname{id}_r)(X)}_1
       }{
          \norm{X}_1
       }.
    \]
    Here \(\tr{X}=0\) denotes the global trace on \(\matrices\otimes M_r\).
\end{definition}

For every channel \(\Phi\), one has
\[
   \kappa_\diamond(\Phi)=1.
\]
Indeed, since \(\Phi\otimes\operatorname{id}_r\) is positive and trace preserving, trace-norm contractivity on self-adjoint matrices gives \(\kappa_\diamond(\Phi)\le1\).
Conversely, choose \(r\ge2\), let \(Y=Y^*\in M_r\) be nonzero with \(\tr{Y}=0\), and let \(\rho_*\in\mcS_d\) be any state.
Then \(X:=\rho_*\otimes Y\) is self-adjoint and globally trace zero, and
\[
   (\Phi\otimes\operatorname{id}_r)(X)
   =
   \Phi(\rho_*)\otimes Y.
\]
Since \(\Phi(\rho_*)\) is a state, \(\norm{\Phi(\rho_*)}_1=1\), and hence
\[
   \norm{(\Phi\otimes\operatorname{id}_r)(X)}_1
   =
   \norm{\Phi(\rho_*)}_1\norm{Y}_1
   =
   \norm{Y}_1
   =
   \norm{X}_1.
\]
Thus, the supremum is at least one, and so \(\kappa_\diamond(\Phi)=1\).

The obstruction is that the global trace-zero condition permits deviations that are entirely stored in the reference system.
These deviations are invisible to the system channel and therefore survive unchanged.
Consequently, the unconstrained diamond-centered coefficient cannot detect ordinary state-memory loss.

Meaningful completely bounded variants must remove this reference-only obstruction.
For example, one may restrict to perturbations with fixed reference marginal, equivalently to self-adjoint \(X\in\matrices\otimes M_r\) satisfying
\[
   \operatorname{Tr}_{\matrices}(X)=0,
\]
or one may compare the amplified channel directly with an amplified replacement channel,
\[
   \Phi\otimes\operatorname{id}_r
   \quad\text{versus}\quad
   R_{\rho_*}\otimes\operatorname{id}_r.
\]
Such variants measure entanglement-assisted memory loss.
They are relevant if one wants stability against arbitrary reference systems, but they are not needed for the MPS boundary-stability problems considered in this paper.

%%%%%%%%%%%%%%%%%%%%%%%%%%%%%%%%%%%%%%%%%%%%%%%%%%%%%%%%%%%%%%%%%%%%%%%%%%%%
%%%%%%%%%%%%%%%%%%%%%%%%%%%%%%%%%%%%%%%%%%%%%%%%%%%%%%%%%%%%%%%%%%%%%%%%%%%%

\subsection*{Summary of implications}
\label{app:comparison-summary}

%%%%%%%%%%%%%%%%%%%%%%%%%%%%%%%%%%%%%%%%%%%%%%%%%%%%%%%%%%%%%%%%%%%%%%%%%%%
%%%%%%%%%%%%%%%%%%%%%%%%%%%%%%%%%%%%%%%%%%%%%%%%%%%%%%%%%%%%%%%%%%%%%%%%%%%%

The directly comparable one-step certificates satisfy the quantitative chain
\[
   \alpha_{\rm Doeb}(\Phi)
   \le
   \alpha_{\rm MD}(\Phi),
   \qquad
   \kappa_{\rm tr}(\Phi)
   \le
   \eta_{\rm MD}(\Phi)
   =
   1-\alpha_{\rm MD}(\Phi)
   \le
   1-\alpha_{\rm Doeb}(\Phi).
\]
Equivalently,
\[
   \alpha_{\rm Doeb}(\Phi)>0
   \Longrightarrow
   \alpha_{\rm MD}(\Phi)>0
   \Longrightarrow
   \kappa_{\rm tr}(\Phi)<1.
\]
Both implications are proper.
\Cref{ex:cp-vs-state-doeblin} shows that state-level Markov--Dobrushin minorization can hold when CP-order Doeblin minorization fails.
\Cref{ex:amplitude} shows that strict trace contraction can hold when both one-step minorization criteria fail.

Projective and faithful-state routes give additional sufficient conditions.
For a positive cone-preserving map with finite Hilbert projective diameter,
\[
   \Delta(\Phi)<\infty
   \Longrightarrow
   \kappa_{\rm tr}(\Phi)
   \le
   \eta_{\rm HB}(\Phi)
   =
   \tanh\!\left(\Delta(\Phi)/4\right)
   <
   1.
\]
In the faithful Riemannian setting where the comparison theorem applies,
\[
   \eta^{\rm Tr}(\Phi)
   \le
   \sqrt{\eta_{\rm Riem,\mathfrak m}(\Phi)}.
\]
Hence, by \(\eta^{\rm Tr}(\Phi)=\kappa_{\rm tr}(\Phi)\),
\[
   \eta_{\rm Riem,\mathfrak m}(\Phi)<1
   \Longrightarrow
   \kappa_{\rm tr}(\Phi)<1.
\]
These routes are not part of the same linear inclusion chain.
Projective coefficients are adapted to cone-preserving positive maps, especially non-trace-preserving ones.
Entropy, Riemannian, and \(\chi^2\) coefficients are adapted to faithful-state geometries.
Quantum Wasserstein--Dobrushin coefficients are adapted to spatial many-body dynamics.
Diamond variants are useful only after specifying how reference systems are constrained.

For products, the intrinsic object is the block coefficient
\[
   \kappa_{t:s}
   :=
   \kappa_{\rm tr}(\Phi_{t:s}).
\]
By \Cref{prop:dobr},
\[
   \kappa_{t:s}<1
   \Longleftrightarrow
   \sup_{\substack{\rho,\sigma\in\mcS_d\\ \rho\ne\sigma}}
   \frac{
      \norm{\Phi_{t:s}(\rho)-\Phi_{t:s}(\sigma)}_1
   }{
      \norm{\rho-\sigma}_1
   }
   <1.
\]
Thus \(\kappa_{t:s}<1\) is exactly strict trace contraction on state differences for the whole block.
This block-level statement is the sharp trace-memory-loss criterion used in the rest of the paper.

The one-step part of the hierarchy is
\[
   \left\{\Phi:\alpha_{\rm Doeb}(\Phi)>0\right\}
   \subsetneq
   \left\{\Phi:\alpha_{\rm MD}(\Phi)>0\right\}
   \subsetneq
   \left\{\Phi:\kappa_{\rm tr}(\Phi)<1\right\}.
\]
At the product level, this is further enlarged:
\[
   \left\{\text{one-step strict trace contraction}\right\}
   \subsetneq
   \left\{\text{finite-block strict trace contraction}\right\}.
\]
The first line compares one-step certificates for a single channel.
The second line is a product-level statement.
It allows loss of memory to arise only after composing several channels.
\Cref{ex:alternating-dephasing} gives a two-step block with
\[
   \kappa_{\rm tr}(D_Z)=\kappa_{\rm tr}(D_X)=1,
   \qquad
   \kappa_{\rm tr}(D_X\circ D_Z)=0.
\]

%%%%%%%%%%%%%%%%%%%%%%%%%%%%%%%%%%%%%%%%%%%%%%%%%%%%%%%%%%%%%%%%%%%%%%%%%%%%
%%%%%%%%%%%%%%%%%%%%%%%%%%%%%%%%%%%%%%%%%%%%%%%%%%%%%%%%%%%%%%%%%%%%%%%%%%%%

\section{Bistochastic Hilbert--Schmidt contraction}
\label{subsec:bistochastic-hs-gaps}

%%%%%%%%%%%%%%%%%%%%%%%%%%%%%%%%%%%%%%%%%%%%%%%%%%%%%%%%%%%%%%%%%%%%%%%%%%%%
%%%%%%%%%%%%%%%%%%%%%%%%%%%%%%%%%%%%%%%%%%%%%%%%%%%%%%%%%%%%%%%%%%%%%%%%%%%%

We record a Hilbert--Schmidt sufficient criterion for a negative trace-Dobrushin exponent in the bistochastic case.
The criterion is based on contraction of the channel on the complex trace-zero Hilbert--Schmidt subspace.
Since every bistochastic channel fixes the maximally mixed state \(I/d\), such contraction forces random products to approach the completely depolarizing replacement channel
\[
   \Omega\left(X\right)
   :=
   \tr{X}\frac{I}{d}.
\]
Thus, the conclusion is depolarizing replacement toward \(I/d\), not merely dephasing.

This condition is the fixed-dimensional analogue of the second-singular-value condition that appears in the quantum-expander literature.
In that literature, one often studies asymptotic families of channels on growing matrix algebras, together with Kraus-rank sparsity, Hilbert--Schmidt contraction, and fixed-state entropy conditions.
Here, the matrix algebra is fixed, and no Kraus-rank sparsity assumption is imposed.
We only use the trace-zero Hilbert--Schmidt contraction as a sufficient condition for trace-norm replacement mixing of random bistochastic products.
Averages of random unitaries and unitary-design constructions provide important sources of unital Hilbert--Schmidt contraction; see \cite{Hastings2007,Lancien2024}.
For the broader quantum-expander framework and recent random-channel constructions, see \cite{LancienYoussef2023}.
Related transfer-operator contraction estimates also appear in random MPS and PEPS correlation-length estimates; see \cite{LancienPerezGarcia2022}.

Throughout this subsection, write
\[
   \norm{X}_2
   :=
   \tr{X^*X}^{1/2}
\]
for the Hilbert--Schmidt norm.
We also introduce the complex trace-zero space
\[
   \mathsf H_0^{\mbC}
   :=
   \left\{
      X\in\matrices:
      \tr{X}=0
   \right\}.
\]
The symbol \(\tracezero\) continues to denote the real self-adjoint trace-zero space.

\begin{definition}[Bistochastic Hilbert--Schmidt contraction coefficient]
    Let \(T:\matrices\to\matrices\) be bistochastic, meaning CPTP and unital.
    Define
    \[
       s_0\left(T\right)
       :=
       \sup_{X\in\mathsf H_0^{\mbC}\setminus\left\{0\right\}}
       \frac{\norm{T\left(X\right)}_2}{\norm{X}_2}.
    \]
    We call \(s_0\left(T\right)\) the trace-zero Hilbert--Schmidt contraction coefficient of \(T\).
    Equivalently, \(s_0\left(T\right)\) is the largest singular value of the restriction \(T|_{\mathsf H_0^{\mbC}}\) with respect to the Hilbert--Schmidt inner product.
    We say that \(T\) is Hilbert--Schmidt contractive on the trace-zero subspace if
    \[
       s_0\left(T\right)<1.
    \]
    In this case, \(1-s_0\left(T\right)\) is the associated trace-zero Hilbert--Schmidt singular-value gap.
\end{definition}

For every bistochastic channel \(T\), one has \(0\le s_0\left(T\right)\le1\).
Indeed, complete positivity and unitality give the Schwarz inequality
\[
   T\left(X\right)^*T\left(X\right)
   \le
   T\left(X^*X\right),
\]
and trace preservation gives
\[
   \norm{T\left(X\right)}_2^2
   =
   \tr{T\left(X\right)^*T\left(X\right)}
   \le
   \tr{T\left(X^*X\right)}
   =
   \tr{X^*X}
   =
   \norm{X}_2^2.
\]
Thus \(T\) is Hilbert--Schmidt nonexpansive, and in particular \(s_0\left(T\right)\le1\).
Moreover, \(T\mapsto s_0\left(T\right)\) is continuous on the set of bistochastic maps.
This follows because \(T\mapsto T|_{\mathsf H_0^{\mbC}}\) is a continuous restriction map and the Hilbert--Schmidt operator norm is continuous in finite dimension.
Hence, if \(\omega\mapsto\Phi_\omega\) is a measurable assignment of bistochastic channels, then
\[
   \omega\mapsto s_0\left(\Phi_\omega\right)
\]
is measurable.
We use the convention \(\log0=-\infty\).

\begin{prop}
\label{prop:hs-contraction-controls-kappa}
    Assume that \(\Phi_\omega\) is bistochastic almost surely.
    Then the constant random state
    \(
       \rho_\omega
       =
       I/d
    \)
    is dynamically stationary in the sense of \Cref{def:dynamically-stationary-random-state}.
    Moreover, for every \(n\ge1\),
    \[
       \kappa_{\omega;n:0}
       \le
       \sqrt d
       \prod_{j=0}^{n-1}
       s_0\left(\Phi_{\theta^j\omega}\right).
    \]
    Consequently,
    \[
       \lambda_{\rm tr}\left(\omega\right)
       \le
       \limsup_{n\to\infty}
       \frac{1}{n}
       \sum_{j=0}^{n-1}
       \log s_0\left(\Phi_{\theta^j\omega}\right)
    \]
    whenever the right-hand side is defined in \(\left[-\infty,0\right]\).
\end{prop}

\begin{proof}
    Since each \(\Phi_\omega\) is unital and trace preserving,
    \[
       \Phi_\omega\left(I/d\right)
       =
       I/d.
    \]
    Hence, the constant random state \(\rho_\omega=I/d\) is stationary.
    Let \(X\in\tracezero\setminus\left\{0\right\}\).
    Then \(X\in\mathsf H_0^{\mbC}\).
    By the comparison between trace norm and Hilbert--Schmidt norm,
    \[
       \norm{\Phi_{\omega;n:0}\left(X\right)}_1
       \le
       \sqrt d\,
       \norm{\Phi_{\omega;n:0}\left(X\right)}_2.
    \]
    The complex trace-zero space \(\mathsf H_0^{\mbC}\) is invariant under every trace-preserving map.
    Therefore,
    \[
       \norm{\Phi_{\omega;n:0}\left(X\right)}_2
       \le
       \left(
          \prod_{j=0}^{n-1}
          s_0\left(\Phi_{\theta^j\omega}\right)
       \right)
       \norm{X}_2.
    \]
    Since \(\norm{X}_2\le\norm{X}_1\), we obtain
    \[
       \norm{\Phi_{\omega;n:0}\left(X\right)}_1
       \le
       \sqrt d
       \left(
          \prod_{j=0}^{n-1}
          s_0\left(\Phi_{\theta^j\omega}\right)
       \right)
       \norm{X}_1.
    \]
    Taking the supremum over \(X\in\tracezero\setminus\left\{0\right\}\) gives the displayed bound on \(\kappa_{\omega;n:0}\).
    With the convention \(\log0=-\infty\), this bound implies
    \[
       \frac{1}{n}\log\kappa_{\omega;n:0}
       \le
       \frac{1}{n}\log\sqrt d
       +
       \frac{1}{n}
       \sum_{j=0}^{n-1}
       \log s_0\left(\Phi_{\theta^j\omega}\right).
    \]
    Letting \(n\to\infty\) gives the exponent estimate, because \(n^{-1}\log\sqrt d\to0\).
\end{proof}

The preceding proposition turns Hilbert--Schmidt contraction along the random product into an upper bound on the trace-Dobrushin exponent.

\begin{theorem}
\label{thm:bistochastic-hs-contraction-criterion}
    Assume that \(\theta\) is ergodic and that \(\Phi_\omega\) is bistochastic almost surely.
    If \(\mbE\left[
          \log s_0\left(\Phi_\omega\right)
       \right]
       <
       0
    \)
    in the extended nonpositive sense, then
    \[
       \lambda_{\rm tr}
       \le
       \mbE\left[
          \log s_0\left(\Phi_\omega\right)
       \right]
       <
       0.
    \]
    Consequently, the cocycle has quenched exponential replacement mixing with the constant stationary random state \(\rho_\omega
       =
       I/d\).
    Equivalently, the forward products converge exponentially to the completely depolarizing replacement channel
    \[
       R_{\rm unif}\left(X\right)
       :=
       \tr{X}\frac{I}{d}.
    \]
    In particular, for every deterministic \(\beta<0\) satisfying \(\lambda_{\rm tr}<\beta\), there is an a.s. finite measurable random variable \(C_\beta\left(\omega\right)\) such that, for every \(n\ge1\),
    \[
       \sup_{\sigma\in\mcS_d}
       \norm{
          \Phi_{\omega;n:0}\left(\sigma\right)-I/d
       }_1
       \le
       C_\beta\left(\omega\right)e^{\beta n},
    \]
    and
    \[
       \norm{
          \Phi_{\omega;n:0}-R_{\rm unif}
       }_{1\to1}
       \le
       C_\beta\left(\omega\right)e^{\beta n}.
    \]
\end{theorem}

\begin{proof}
    Set
    \[
       g\left(\omega\right)
       :=
       \log s_0\left(\Phi_\omega\right)
       \in
       \left[-\infty,0\right].
    \]
    For \(M\ge1\), define
    \[
       g_M\left(\omega\right)
       :=
       \max\left\{
          g\left(\omega\right),
          -M
       \right\}.
    \]
    Then \(g_M\in L^1\left(\pr\right)\), \(g\le g_M\le0\), and \(g_M\downarrow g\) pointwise as \(M\to\infty\).
    By \Cref{prop:hs-contraction-controls-kappa},
    \[
       \log\kappa_{\omega;n:0}
       \le
       \log\sqrt d
       +
       \sum_{j=0}^{n-1}
       g\left(\theta^j\omega\right)
       \le
       \log\sqrt d
       +
       \sum_{j=0}^{n-1}
       g_M\left(\theta^j\omega\right).
    \]
    Taking expectations gives
    \[
       \mbE\left[
          \log\kappa_{\omega;n:0}
       \right]
       \le
       \log\sqrt d
       +
       n\mbE\left[g_M\right].
    \]
    Since \(\theta\) is ergodic, Kingman's variational formula in \Cref{thm:random-lyapunov} gives
    \[
       \lambda_{\rm tr}
       =
       \inf_{n\ge1}
       \frac{1}{n}
       \mbE\left[
          \log\kappa_{\omega;n:0}
       \right].
    \]
    Therefore
    \[
       \lambda_{\rm tr}
       \le
       \inf_{n\ge1}
       \left(
          \frac{1}{n}\log\sqrt d
          +
          \mbE\left[g_M\right]
       \right)
       =
       \mbE\left[g_M\right].
    \]
    Letting \(M\to\infty\) and using monotone convergence for \(-g_M\uparrow -g\), we obtain
    \[
       \lambda_{\rm tr}
       \le
       \mbE\left[g\right]
       =
       \mbE\left[
          \log s_0\left(\Phi_\omega\right)
       \right].
    \]
    The assumption \(\mbE\left[g\right]<0\) gives \(\lambda_{\rm tr}<0\).

    The constant random state \(\rho_\omega=I/d\) is stationary because every \(\Phi_\omega\) is unital and trace preserving.
    The uniqueness part of \Cref{thm:negative-exponent-replacement} identifies this constant random state with the stationary random state produced by the negative-exponent theorem.
    The state-level exponential estimate follows from \Cref{thm:negative-exponent-replacement}.
    The \(1\to1\) replacement-channel estimate follows from the operator-norm part of \Cref{thm:negative-exponent-replacement}, after increasing \(C_\beta\left(\omega\right)\) if necessary.
\end{proof}

The preceding argument only uses Hilbert--Schmidt contraction of the relevant product.
Therefore, the same criterion applies when such contraction appears first after forming fixed-length blocks.

\begin{cor}
\label{cor:block-hs-contraction-criterion}
    Assume that \(\theta\) is ergodic and that \(\Phi_\omega\) is bistochastic almost surely.
    For \(L\ge1\), define \(s_{0,L}\left(\omega\right)
       :=
       s_0\left(\Phi_{\omega;L:0}\right).
    \)
    If
    \[
       \mbE\left[
          \log s_{0,L}
       \right]
       <
       0
    \]
    for some \(L\ge1\), in the extended nonpositive sense, then \(\lambda_{\rm tr}< 0\).
    More precisely,
    \[
       \lambda_{\rm tr}
       \le
       \frac{1}{L}
       \mbE\left[
          \log s_{0,L}
       \right]
       <
       0.
    \]
    Consequently, the cocycle has quenched exponential replacement mixing with the constant stationary random state
    \[
       \rho_\omega
       =
       I/d.
    \]
    Equivalently, the forward products converge exponentially to the completely depolarizing replacement channel
    \[
       R_{\rm unif}\left(X\right)
       :=
       \tr{X}\frac{I}{d}.
    \]
\end{cor}

\begin{proof}
    Set
    \[
       h\left(\omega\right)
       :=
       \log s_{0,L}\left(\omega\right)
       \in
       \left[-\infty,0\right].
    \]
    For \(M\ge1\), define
    \[
       h_M\left(\omega\right)
       :=
       \max\left\{
          h\left(\omega\right),
          -M
       \right\}.
    \]
    Then \(h_M\in L^1\left(\pr\right)\), \(h\le h_M\le0\), and \(h_M\downarrow h\) pointwise as \(M\to\infty\).

    Fix \(q\ge1\).
    The product \(\Phi_{\omega;qL:0}\) decomposes into the \(q\) bistochastic blocks
    \[
       \Phi_{\omega;qL:0}
       =
       \Phi_{\theta^{\left(q-1\right)L}\omega;L:0}
       \circ
       \cdots
       \circ
       \Phi_{\theta^L\omega;L:0}
       \circ
       \Phi_{\omega;L:0}.
    \]
    Repeating the Hilbert--Schmidt estimate from \Cref{prop:hs-contraction-controls-kappa} at the block level gives
    \[
       \kappa_{\omega;qL:0}
       \le
       \sqrt d
       \prod_{r=0}^{q-1}
       s_{0,L}\left(\theta^{rL}\omega\right).
    \]
    Hence, with the convention \(\log0=-\infty\),
    \[
       \log\kappa_{\omega;qL:0}
       \le
       \log\sqrt d
       +
       \sum_{r=0}^{q-1}
       h\left(\theta^{rL}\omega\right)
       \le
       \log\sqrt d
       +
       \sum_{r=0}^{q-1}
       h_M\left(\theta^{rL}\omega\right).
    \]
    Taking expectations and using stationarity of \(\theta\) gives
    \[
       \mbE\left[
          \log\kappa_{\omega;qL:0}
       \right]
       \le
       \log\sqrt d
       +
       q\mbE\left[h_M\right].
    \]
    No ergodicity of \(\theta^L\) is needed for this estimate.
    By Kingman's variational formula in \Cref{thm:random-lyapunov},
    \[
       \lambda_{\rm tr}
       \le
       \frac{1}{qL}
       \mbE\left[
          \log\kappa_{\omega;qL:0}
       \right].
    \]
    Therefore
    \[
       \lambda_{\rm tr}
       \le
       \frac{1}{qL}\log\sqrt d
       +
       \frac{1}{L}
       \mbE\left[h_M\right].
    \]
    Letting \(q\to\infty\) gives
    \[
       \lambda_{\rm tr}
       \le
       \frac{1}{L}
       \mbE\left[h_M\right].
    \]
    Letting \(M\to\infty\) and using monotone convergence for \(-h_M\uparrow -h\), we obtain
    \[
       \lambda_{\rm tr}
       \le
       \frac{1}{L}
       \mbE\left[h\right]
       =
       \frac{1}{L}
       \mbE\left[
          \log s_{0,L}
       \right].
    \]
    The assumed negativity of this expectation gives \(\lambda_{\rm tr}<0\).
    The replacement-mixing conclusions follow from \Cref{thm:negative-exponent-replacement}.
    The stationary random state is \(\rho_\omega=I/d\) because every block and every one-step channel is bistochastic.
\end{proof}

%%%%%%%%%%%%%%%%%%%%%%%%%%%%%%%%%%%%%%%%%%%%%%%%%%%%%%%%%%%%%%%%%%%%%%%%%%%%
%%%%%%%%%%%%%%%%%%%%%%%%%%%%%%%%%%%%%%%%%%%%%%%%%%%%%%%%%%%%%%%%%%%%%%%%%%%

\end{appendix}

%%%%%%%%%%%%%%%%%%%%%%%%%%%%%%%%%%%%%%%%%%%%%%%%%%%%%%%%%%%%%%%%%%%%%%%%%%%%
%%%%%%%%%%%%%%%%%%%%%%%%%%%%%%%%%%%%%%%%%%%%%%%%%%%%%%%%%%%%%%%%%%%%%%%%%%%

\addcontentsline{toc}{section}{Bibliography}
\printbibliography

@Article{Stinespring1955,
  author    = {Stinespring, W. Forrest},
  journal   = {Proceedings of the American Mathematical Society},
  title     = {Positive Functions on C*-Algebras},
  year      = {1955},
  issn      = {0002-9939},
  month     = Apr,
  number    = {2},
  pages     = {211},
  volume    = {6},
  doi       = {10.2307/2032342},
  publisher = {JSTOR},
}

@Article{Choi1975,
  author    = {Choi, Man-Duen},
  journal   = {Linear Algebra and its Applications},
  title     = {Completely positive linear maps on complex matrices},
  year      = {1975},
  issn      = {0024-3795},
  month     = jun,
  number    = {3},
  pages     = {285--290},
  volume    = {10},
  doi       = {10.1016/0024-3795(75)90075-0},
  publisher = {Elsevier BV},
}

@book{Kraus1983, 
    title={States, Effects, and Operations Fundamental Notions of Quantum Theory}, 
    author={Kraus, Karl and B{\"o}hm, Arno and Dollard, John D and Wootters, WH}, 
    ISBN={9783540387251}, 
    url={http://dx.doi.org/10.1007/3-540-12732-1}, 
    DOI={10.1007/3-540-12732-1}, 
    journal={Lecture Notes in Physics}, 
    publisher={Springer Berlin Heidelberg}, year={1983} 
}

@Book{NielsenChuang2000,
  author    = {Nielsen, Michael A. and Chuang, Isaac L.},
  publisher = {Cambridge University Press},
  title     = {Quantum Computation and Quantum Information: 10th Anniversary Edition},
  year      = {2012},
  isbn      = {9780511976667},
  month     = Jun,
  doi       = {10.1017/cbo9780511976667},
}

@Book{Watrous2018,
  author    = {Watrous, John},
  publisher = {Cambridge University Press},
  title     = {The Theory of Quantum Information},
  year      = {2018},
  isbn      = {9781107180567},
  month     = apr,
  doi       = {10.1017/9781316848142},
}

@Article{CarbonePautrat2016,
  author    = {Carbone, Raffaella and Pautrat, Yan},
  journal   = {Reports on Mathematical Physics},
  title     = {Irreducible Decompositions and Stationary States of Quantum Channels},
  year      = {2016},
  issn      = {0034-4877},
  month     = Jun,
  number    = {3},
  pages     = {293--313},
  volume    = {77},
  doi       = {10.1016/s0034-4877(16)30032-5},
  publisher = {Elsevier BV},
}

@Article{BurgarthEtAl2013,
  author    = {Burgarth, D and Chiribella, G and Giovannetti, V and Perinotti, P and Yuasa, K},
  journal   = {New Journal of Physics},
  title     = {Ergodic and mixing quantum channels in finite dimensions},
  year      = {2013},
  issn      = {1367-2630},
  month     = Jul,
  number    = {7},
  pages     = {073045},
  volume    = {15},
  doi       = {10.1088/1367-2630/15/7/073045},
  publisher = {IOP Publishing},
}

@Article{FannesNachtergaeleWerner1992,
  author    = {Fannes, M. and Nachtergaele, B. and Werner, R. F.},
  journal   = {Communications in Mathematical Physics},
  title     = {Finitely correlated states on quantum spin chains},
  year      = {1992},
  issn      = {1432-0916},
  month     = Mar,
  number    = {3},
  pages     = {443--490},
  volume    = {144},
  doi       = {10.1007/bf02099178},
  publisher = {Springer Science and Business Media LLC},
}

@Article{EvansHoeghKrohn1978,
  author    = {David E. Evans and Raphael Høegh-Krohn},
  journal   = {Journal of the London Mathematical Society},
  title     = {Spectral Properties of Positive Maps on C*-Algebras},
  year      = {1978},
  issn      = {0024-6107},
  number    = {2},
  pages     = {345-355},
  volume    = {s2-17},
  doi       = {10.1112/jlms/s2-17.2.345},
  publisher = {Wiley},
  url       = {https://api.semanticscholar.org/CorpusID:56109123},
}

@Article{HiaiRuskai2016,
  author    = {Hiai, Fumio and Ruskai, Mary Beth},
  journal   = {Journal of Mathematical Physics},
  title     = {Contraction coefficients for noisy quantum channels},
  year      = {2015},
  issn      = {1089-7658},
  month     = Dec,
  number    = {1},
  volume    = {57},
  doi       = {10.1063/1.4936215},
  publisher = {AIP Publishing},
}

@InProceedings{Hirche2024,
  author    = {Hirche, Christoph},
  booktitle = {2024 IEEE International Symposium on Information Theory (ISIT)},
  title     = {Quantum Doeblin Coefficients: A Simple Upper Bound on Contraction Coefficients},
  year      = {2024},
  month     = Jul,
  pages     = {557--562},
  publisher = {IEEE},
  doi       = {10.1109/isit57864.2024.10619667},
}

@Article{Birkhoff1957,
  author    = {Birkhoff, Garrett},
  journal   = {Transactions of the American Mathematical Society},
  title     = {Extensions of Jentzsch’s Theorem},
  year      = {1957},
  issn      = {0002-9947},
  month     = May,
  number    = {1},
  pages     = {219},
  volume    = {85},
  doi       = {10.2307/1992971},
  publisher = {JSTOR},
}

@Article{Bushell1973,
  author    = {Bushell, P. J.},
  journal   = {Archive for Rational Mechanics and Analysis},
  title     = {Hilbert’s metric and positive contraction mappings in a Banach space},
  year      = {1973},
  issn      = {1432-0673},
  number    = {4},
  pages     = {330--338},
  volume    = {52},
  doi       = {10.1007/bf00247467},
  publisher = {Springer Science and Business Media LLC},
}

@Article{ReebKastoryanoWolf2011,
  author    = {Reeb, David and Kastoryano, Michael J. and Wolf, Michael M.},
  journal   = {Journal of Mathematical Physics},
  title     = {Hilbert’s projective metric in quantum information theory},
  year      = {2011},
  issn      = {1089-7658},
  month     = Aug,
  number    = {8},
  volume    = {52},
  doi       = {10.1063/1.3615729},
  publisher = {AIP Publishing},
}

@Article{BakshiLiuMoitraTang2025,
  author    = {Bakshi, Ainesh and Liu, Allen and Moitra, Ankur and Tang, Ewin},
  title     = {A Dobrushin condition for quantum Markov chains: Rapid mixing and conditional mutual information at high temperature},
  year      = {2025},
  copyright = {Creative Commons Attribution 4.0 International},
  doi       = {10.48550/ARXIV.2510.08542},
  publisher = {arXiv},
}

@Article{Kingman1973,
  author    = {Kingman, J. F. C.},
  journal   = {The Annals of Probability},
  title     = {Subadditive Ergodic Theory},
  year      = {1973},
  issn      = {0091-1798},
  month     = Dec,
  number    = {6},
  volume    = {1},
  doi       = {10.1214/aop/1176996798},
  publisher = {Institute of Mathematical Statistics},
}

@Article{Dobrushin1956,
  author    = {Dobrushin, R. L.},
  journal   = {Theory of Probability \&; Its Applications},
  title     = {Central Limit Theorem for Nonstationary Markov Chains. I},
  year      = {1956},
  issn      = {1095-7219},
  month     = Jan,
  number    = {1},
  pages     = {65--80},
  volume    = {1},
  doi       = {10.1137/1101006},
  publisher = {Society for Industrial \& Applied Mathematics (SIAM)},
}

@Article{MovassaghSchenker2021,
  author    = {Movassagh, Ramis and Schenker, Jeffrey},
  journal   = {Physical Review X},
  title     = {Theory of Ergodic Quantum Processes},
  year      = {2021},
  issn      = {2160-3308},
  number    = {4},
  pages     = {041001},
  volume    = {11},
  doi       = {10.1103/physrevx.11.041001},
  publisher = {American Physical Society (APS)},
}

@Article{MovassaghSchenker2022,
  author    = {Movassagh, Ramis and Schenker, Jeffrey},
  journal   = {Communications in Mathematical Physics},
  title     = {An Ergodic Theorem for Quantum Processes with Applications to Matrix Product States},
  year      = {2022},
  issn      = {1432-0916},
  month     = jul,
  number    = {3},
  pages     = {1175--1196},
  volume    = {395},
  doi       = {10.1007/s00220-022-04448-0},
  publisher = {Springer Science and Business Media LLC},
}

@Article{Souissi2026,
  author    = {Souissi, Abdessatar},
  title     = {Ergodic Theory of Inhomogeneous Quantum Processes},
  year      = {2025},
  copyright = {Creative Commons Attribution 4.0 International},
  doi       = {10.48550/ARXIV.2506.12280},
  publisher = {arXiv},
}

@Article{PerezGarciaVerstraeteWolfCirac2007,
  author    = {Perez-Garcia, D. and Verstraete, F. and Wolf, M.M. and Cirac, J.I.},
  journal   = {Quantum Information and Computation},
  title     = {Matrix product state representations},
  year      = {2007},
  issn      = {1533-7146},
  month     = Jul,
  number    = {5-6},
  pages     = {401--430},
  volume    = {7},
  doi       = {10.26421/qic7.5-6-1},
  publisher = {Rinton Press},
}

@Article{VerstraeteMurgCirac2008,
  author    = {Verstraete, F. and Murg, V. and Cirac, J.I.},
  journal   = {Advances in Physics},
  title     = {Matrix product states, projected entangled pair states, and variational renormalization group methods for quantum spin systems},
  year      = {2008},
  issn      = {1460-6976},
  month     = Mar,
  number    = {2},
  pages     = {143--224},
  volume    = {57},
  doi       = {10.1080/14789940801912366},
  publisher = {Informa UK Limited},
}

@Article{Dobrushin1970,
  author    = {Dobrushin, R. L.},
  journal   = {Theory of Probability \&; Its Applications},
  title     = {Prescribing a System of Random Variables by Conditional Distributions},
  year      = {1970},
  issn      = {1095-7219},
  month     = Jan,
  number    = {3},
  pages     = {458--486},
  volume    = {15},
  doi       = {10.1137/1115049},
  publisher = {Society for Industrial \& Applied Mathematics (SIAM)},
}

@Article{AccardiLuSouissi2021,
  author    = {Accardi, Luigi and Lu, Yun Gang and Souissi, Abdessatar},
  journal   = {Open Systems \&; Information Dynamics},
  title     = {A Markov–Dobrushin Inequality for Quantum Channels},
  year      = {2021},
  issn      = {1793-7191},
  month     = Dec,
  number    = {04},
  volume    = {28},
  doi       = {10.1142/s1230161221500189},
  publisher = {World Scientific Pub Co Pte Ltd},
}

@Article{SouissiBarhoumi2025,
  author    = {Souissi, Abdessatar and Barhoumi, Abdessatar},
  title     = {An Exponential Mixing Condition for Quantum Channels},
  year      = {2024},
  copyright = {arXiv.org perpetual, non-exclusive license},
  doi       = {10.48550/ARXIV.2407.16031},
  publisher = {arXiv},
}

@Article{Hastings2007,
  author    = {Hastings, M. B.},
  journal   = {Physical Review A},
  title     = {Random unitaries give quantum expanders},
  year      = {2007},
  issn      = {1094-1622},
  month     = Sep,
  number    = {3},
  pages     = {032315},
  volume    = {76},
  doi       = {10.1103/physreva.76.032315},
  publisher = {American Physical Society (APS)},
}

@Article{Lancien2024,
  author    = {Lancien, Cécilia},
  journal   = {International Mathematics Research Notices},
  title     = {Optimal Quantum (Tensor Product) Expanders From Unitary Designs},
  year      = {2026},
  issn      = {1687-0247},
  month     = Jan,
  number    = {2},
  volume    = {2026},
  doi       = {10.1093/imrn/rnaf383},
  publisher = {Oxford University Press (OUP)},
}

@Article{LancienPerezGarcia2022,
  author    = {Lancien, Cécilia and Pérez-García, David},
  journal   = {Annales Henri Poincaré},
  title     = {Correlation Length in Random MPS and PEPS},
  year      = {2021},
  issn      = {1424-0661},
  month     = Aug,
  number    = {1},
  pages     = {141--222},
  volume    = {23},
  doi       = {10.1007/s00023-021-01087-4},
  publisher = {Springer Science and Business Media LLC},
}

@Article{LancienYoussef2023,
  author    = {Lancien, Cécilia and Youssef, Pierre},
  journal   = {Letters in Mathematical Physics},
  title     = {A note on quantum expanders},
  year      = {2025},
  issn      = {1573-0530},
  month     = Nov,
  number    = {6},
  volume    = {115},
  doi       = {10.1007/s11005-025-02028-6},
  publisher = {Springer Science and Business Media LLC},
}

@Book{walters2000introduction,
  author    = {Walters, Peter},
  publisher = {Springer Science \& Business Media},
  title     = {An introduction to ergodic theory},
  year      = {2000},
  isbn      = {978-0-387-95152-2},
  volume    = {79},
}

@Article{bradley2005basic,
  author    = {Bradley, Richard C},
  journal   = {Probability Surveys},
  title     = {Basic Properties of Strong Mixing Conditions. A Survey and Some Open Questions},
  year      = {2005},
  issn      = {1549-5787},
  number    = {none},
  pages     = {107 -- 144},
  volume    = {2},
  doi       = {10.1214/154957805100000104},
  publisher = {Institute of Mathematical Statistics},
}

@book{bradley2007introduction,
   AUTHOR = {Bradley, Richard C.},
     TITLE = {Introduction to strong mixing conditions. {V}ol. 1},
 PUBLISHER = {Kendrick Press, Heber City, UT},
      YEAR = {2007},
     PAGES = {xviii+539},
      ISBN = {0-9740427-6-5},
   MRCLASS = {60-02 (37A25 37A50 60F05)},
  MRNUMBER = {2325294},
}

@Book{doukhan2012mixing,
  author    = {Doukhan, Paul},
  publisher = {Springer New York},
  title     = {Mixing: properties and examples},
  year      = {1994},
  isbn      = {9781461226420},
  volume    = {85},
  doi       = {10.1007/978-1-4612-2642-0},
  issn      = {0930-0325},
  journal   = {Lecture Notes in Statistics},
}

@book {doob1942stochastic,
    AUTHOR = {Doob, J. L.},
     TITLE = {Stochastic processes},
    SERIES = {Wiley Classics Library},
      NOTE = {Reprint of the 1953 original,
              A Wiley-Interscience Publication},
 PUBLISHER = {John Wiley \& Sons, Inc., New York},
      YEAR = {1990},
     PAGES = {viii+654},
      ISBN = {0-471-52369-0},
   MRCLASS = {60-02 (01A75 60-01 60Gxx 60Jxx)},
  MRNUMBER = {1038526},
}

@article{peligrad83noteon,
 ISSN = {00018678},
 URL = {http://www.jstor.org/stable/1426446},
 author = {Magda Peligrad},
 journal = {Advances in Applied Probability},
 number = {2},
 pages = {461--464},
 publisher = {Applied Probability Trust},
 title = {A Note on Two Measures of Dependence and Mixing Sequences},
 urldate = {2024-07-04},
 volume = {15},
 year = {1983},
DOI={https://doi.org/10.2307/1426446}
}

@Article{Souissi_2025,
  author    = {Souissi, Abdessatar and Barhoumi, Abdessatar},
  journal   = {Quantum Information Processing},
  title     = {An Exponential Mixing Condition for Quantum Channels: Application to Matrix Product States},
  year      = {2025},
  issn      = {1573-1332},
  month     = May,
  number    = {5},
  volume    = {24},
  doi       = {10.1007/s11128-025-04762-1},
  publisher = {Springer Science and Business Media LLC},
}

@Article{asym,
  author    = {Pathirana, Lubashan and Schenker, Jeffrey},
  journal   = {arXiv preprint arXiv:2509.08924},
  title     = {Asymptotic Behavior of Random Time-Inhomogeneous Markovian Quantum Dynamics},
  year      = {2025},
  copyright = {Creative Commons Attribution 4.0 International},
  doi       = {10.48550/ARXIV.2509.08924},
  publisher = {arXiv},
  URL       = {https://arxiv.org/abs/2509.08924},
}

@article{PS23,
    title = {{Law of large numbers and central limit theorem for ergodic quantum processes}},
    year = {2023},
    journal = {Journal of Mathematical Physics},
    author = {Pathirana, Lubashan and Schenker, Jeffrey},
    number = {8},
    month = {8},
    volume = {64},
    publisher = {AIP Publishing},
    url = {http://dx.doi.org/10.1063/5.0153483},
    doi = {10.1063/5.0153483},
    issn = {1089-7658}
}

@Article{MPS,
  author    = {Pathirana, Lubashan and Werner, Albert H.},
  journal   = {arXiv preprint arXiv:2510.07561},
  title     = {Correlation Lengths for Stochastic Matrix Product States},
  year      = {2025},
  copyright = {Creative Commons Attribution 4.0 International},
  doi       = {10.48550/ARXIV.2510.07561},
  publisher = {arXiv},
  URL       = {https://arxiv.org/abs/2510.07561},
}

@book{BratteliRobinson1997,
  author    = {Bratteli, Ola and Robinson, Derek W.},
  title     = {Operator Algebras and Quantum Statistical Mechanics:
               Equilibrium States. Models in Quantum Statistical Mechanics},
  edition   = {2},
  publisher = {Springer},
  year      = {1997},
  series    = {Theoretical and Mathematical Physics},
  doi       = {10.1007/978-3-662-03444-6}
}

@book{BratteliRobinson1987,
  author    = {Bratteli, Ola and Robinson, Derek W.},
  title     = {Operator Algebras and Quantum Statistical Mechanics 1:
               {C}*- and {W}*-Algebras, Symmetry Groups, Decomposition of States},
  edition   = {2},
  publisher = {Springer},
  year      = {1987},
  series    = {Texts and Monographs in Physics}
}

%%%%%%%%%%%%%%%%%%%%%%%%%%%%%%%%%%%%%%%%%%%%%%%%%%%%%%%%%%%%%%%%%%%%%%%%%%%
%%%%%%%%%%%%%%%%%%%%%%%%%%%%%%%%%%%%%%%%%%%%%%%%%%%%%%%%%%%%%%%%%%%%%%%%%%%

\end{document}